\documentclass{article}
\usepackage[margin=1in]{geometry}
\usepackage{setspace}
\usepackage[utf8]{inputenc}
\usepackage{lineno}
\usepackage{natbib}
\usepackage{graphicx}        
\usepackage{wrapfig}         
\usepackage{subcaption}      
\usepackage{amsmath}
\usepackage{amsthm}
\usepackage{amssymb}
\usepackage{algorithm}
\usepackage[noend]{algpseudocode}
\usepackage{caption}
\usepackage{setspace}
\usepackage{url}
\usepackage{enumitem}
\usepackage[dvipsnames,usenames]{xcolor}       
\usepackage{authblk}
\usepackage{orcidlink}

\usepackage{accents}
\newlength{\dhatheight}

\DeclareMathOperator{\dist}{dist}

\RequirePackage{hyperref}
\hypersetup{colorlinks=true, urlcolor=Blue, citecolor=Green, linkcolor=BrickRed, breaklinks, unicode}
\usepackage[capitalize]{cleveref} 

\usepackage{xcolor}  
\usepackage{marginnote}  

\newcommand{\VD}{\textsf{VD}}
\newcommand{\Vor}{\textsf{Vor}}
\newcommand{\Otild}{\tilde{O}}
\newcommand{\weight}{{\rm \omega}}
\newcommand{\Comp}{Y}
\newcommand{\bottom}{b}

\newcommand{\cotree}[1]{\hat{\mathcal T}^*_{#1}}

\makeatletter
\renewcommand{\ALG@name}{Procedure} 
\crefname{algorithm}{Procedure}{Procedures}
\Crefname{algorithm}{Procedure}{Procedures}
\makeatother

\theoremstyle{plain}
\newtheorem{theorem}{Theorem}[section]
\newtheorem{lemma}[theorem]{Lemma}
\newtheorem{proposition}[theorem]{Proposition}
\newtheorem{definition}[theorem]{Definition}
\newtheorem{corollary}[theorem]{Corollary}

\title{Distances in Planar Graphs are Almost for Free!\thanks{This research was partially funded by Israel Science Foundation grants Nos. 1867/20 and 810/21.}}

\author[1]{Shay Mozes}
\affil[1]{Reichman University, Israel\\
\href{smozes@runi.ac.il}{smozes@runi.ac.il}
\orcidlink{0000-0001-9262-1821}
}

\author[2]{Daniel Prigan}
\affil[2]{Reichman University, Israel\\
\href{daniel.prigan@post.runi.ac.il}{daniel.prigan@post.runi.ac.il}
\orcidlink{0009-0000-9285-3227}
}

\begin{document}
\date{}
\maketitle

\begin{abstract}
We prove that, up to subpolynomial or polylogarithmic factors, there is no tradeoff between preprocessing time, query time, and size of exact distance oracles for planar graphs.
Namely, we show how given an $n$-vertex weighted directed planar graph $G$, one can compute in $n^{1+o(1)}$ time and space a representation of $G$ from which one can extract the exact distance between any two vertices of $G$ in $\log^{2+o(1)}(n)$  time. 
Previously, it was only known how to construct oracles with these space and query time in $n^{3/2+o(1)}$ time [STOC 2019, SODA 2021, JACM 2023].
We show how to construct these oracles in $n^{1+o(1)}$ time.

\end{abstract}

\thispagestyle{empty}
\setcounter{page}{0}
\newpage

\section{Introduction}
Computing distances in graphs is one of the most fundamental, natural and important problems in algorithmic graph theory. 
Ideally, one would like to efficiently convert the local adjacency-based representation of a graph into a representation of the same size that allows constant-time access to distances between any pair of vertices. 
Making the above statement more precise, the ideal situation would be a linear-time algorithm that converts a graph of size $n$, represented by the adjacency list of each vertex, into a representation of size $O(n)$ that supports distance queries between any pair of vertices in $O(1)$ time.

In this paper we focus on the above problem in planar graphs, which is one of the graph families in which distance oracles were most extensively studied, both for its relevance in modeling various real-world networks, and as a theoretical model which can be generalized and extended to other graph families. 
A very long and extensive line of work on computing distances in planar graph exists. 
The earliest works considered algorithms for computing the distances (and shortest paths) for a specific pair of vertices, or rather the shortest path tree rooted at a specific vertex. 
It is easy to see that one must at least read the entire graph in $\Omega(n)$ time even for this more restricted problem, and even if one is willing to settle for approximate distances.
Then researchers began investigating the distance oracle problem, with dozens of papers studying the various tradeoffs between size, query time, preprocessing time, and approximation guarantees.
We show that, up to subpolynomial factors, no tradeoff is required; One can get almost the best we can hope for in all criteria simultaneously. In other words, once we invest slightly more than the linear time it takes to just read the graph, we can get any distance in the graph nearly for free!
\begin{theorem} \label{thm:main}
Let $G$ be a directed weighted planar graph with $n$ vertices (and hence $O(n)$ edges, and an adjacency input representation of size $O(n)$). There exists an $n^{1+o(1)}$-time algorithm that produces a representation of $G$ of size $n^{1+o(1)}$ that supports querying the exact distance between any pair of vertices of $G$ in $\log^{2+o(1)} n$ time. 

More generally, for any $m$ that is in $\omega(1) \cap o(\log n)$,
and any $k$ that is $O(\log n)$
we can construct in 
$n^{1+o(1)}$ time 
the oracles of \cite{ourJACM} that have 
$O(m k n^{1 + 1/m + 1/k})$ space 
and $O(2^m k \log^{2}{n} \log\log n)$ query time.
\end{theorem}

\cref{thm:main} can be thought of in some sense as an analogue of the fast Fourier transform (FFT) for distances in planar graphs. FFT is a nearly linear time  algorithm that converts an input sequence from a local representation (time or space domain) into a non-local representation (frequency domain) of the same size, from which one can query information in the non-local domain in constant time per query.
A naive algorithm for the Fourier transform would take quadratic time (as would a naive algorithm for all-pairs shortest paths in a planar graph).
We stress that we find this analogy appealing not in the technical sense of the structural properties or the algorithmic techniques, but in the general notion of converting in almost linear time between two useful representations of a dataset. 
Because of the simplicity of the representation (just an array), we do not usually think of FFT as a data structure problem. Moreover, we are by now quite used to the ``miracle'' of lack of tradeoff between preprocessing time, size of the transformed sequence and query time for the discrete Fourier transform problem.
Now we are finally able to show that an analogous ``miracle'' occurs for distances in planar graphs.

\paragraph{History of Exact Distance Oracles for Planar Graphs}

The first papers to discuss planar distance oracles date back to 1996.
Let the query time and the space of an oracle be denoted by $Q$ and $S$, respectively.
The early works ~\cite{ArikatiCCDSZ96,Djidjev96,ChenX00} mainly used planar separators and $r$-divisions~\cite{LiptonT80,Miller86,Frederickson87}
to achieve a space-query tradeoff of 
$Q=\tilde{O}(n/\sqrt{S})$ for $S\in [n^{4/3},n^2]$, 
and a weaker tradeoff of 
$Q=O(n^2/S)$ for $S \in [n,n^{4/3})$.  

In FOCS 2001 Fakcharoenphol and Rao~\cite{FakcharoenpholR06} 
used the \emph{Monge} property to obtain a distance oracle with $\tilde{O}(n)$ space and $\tilde{O}(\sqrt{n})$ query time.
In SODA 2012 Mozes and Sommer~\cite{MozesS2012} used  \cite{FakcharoenpholR06}, and  Klein's~\cite{Klein05,CabelloCE13} multiple source shortest path (MSSP) data structure,   to obtain the $Q=\tilde{O}(n/\sqrt{S})$ space-query tradeoff for nearly the full range $[n\log\log n, n^2]$. All of the above oracles can be  constructed in $\Otild(S)$ time. However, none of them provides polylogarithmic $\tilde O(1)$ query-time using truly subquadratic $O(n^{2-\varepsilon})$ space. 

In FOCS 2017, Cohen-Addad, Dahlgaard, and Wulff-Nilsen~\cite{Cohen-AddadDW17}, inspired by Cabello's use of \emph{additively weighted planar Voronoi diagrams} for the diameter problem, obtained the first exact distance oracle for planar graphs with subquadratic space and polylogarithmic query time. 
The Voronoi diagram based oracle in their breakthrough paper has a space-time tradeoff of $Q=\tilde{O}(n^{5/2}/S^{3/2})$
for $S\in[n^{3/2},n^{5/3}]$.
This came at the cost of increasing the preprocessing time from $\Otild(n)$ to $O(n^2)$, which was subsequently improved to match the space bound~\cite{ourDiameter}.

In SODA 2018, Gawrychowski, Mozes, Weimann, and Wulff-Nilsen~\cite{GawrychowskiMWW18} improved the space and preprocessing time to $\tilde{O}(n^{3/2})$ with $\tilde O(1)$ query-time (and the tradeoff to $Q=\tilde{O}(n^{3/2}/S)$ for $S\in[n,n^{3/2}]$) by defining a dual representation of Voronoi diagrams and developing an efficient {\em point-location} mechanism on top of it. 
In STOC 2019, Charalampopoulos, Gawrychowski, Mozes, and Weimann~\cite{CharalampopoulosGMW19} observed that the same point-location mechanism can be used on the Voronoi diagram for the complement of regions in the $r$-division. This observation alone suffices to improve the oracle size to $O(n^{4/3})$ with $\tilde O(1)$ query-time, at the cost of increasing the preprocessing time to $\Otild(n^{5/3})$.
By combining this idea with a sophisticated recursion (where a query at recursion level $i$ reduces to $O(\log n)$ queries at level $i+1$) they obtained an oracle of size $n^{1+o(1)}$ and  query-time $n^{o(1)}$. They were also able to speed up the construction time of their oracle to $O(n^{3/2 +o(1)})$ by developing an algorithm for computing the dual representation of a Voronoi diagram with $S$ sites in $\Otild(\sqrt{nS})$ time. This construction algorithm uses the algorithm of Fakcharoenphol and Rao~\cite{FakcharoenpholR06} to quickly zoom in and identify Voronoi vertices. 

Finally, in SODA 2021, 
Long and Pettie~\cite{LongPettie} showed that much of the point-location work can be done without recursion, and that only two (instead of $\log n$) recursive calls suffice. This led to the state-of-the-art oracle, requiring $n^{3/2 +o(1)}$ preprocessing time, $n^{1+o(1)}$ space and $\tilde O(1)$ query-time or $\tilde O(n)$ space and $n^{o(1)}$ query-time. See~\cite{ourJACM} for a JACM paper combining most of  the above Voronoi-based oracles.

The construction time in the state-of-the-art tradeoff is the only aspect in which the distance oracle problem was resolved in an almost optimal manner (i.e., up to subpolynomial or polylogarithmic terms). 
Settling the construction time problem is arguably the most important remaining open problem in the field of distance oracles for planar graphs.  

All of the distance oracles cited above report \emph{exact} distances. For $(1+\epsilon)$-approximate distances it has been known that, up to polylogarithmic factors, there is no tradeoff between preprocessing time, query time and size. Specifically,   Thorup~\cite{Thorup04} proved that in
directed planar graphs with non-negative integer weights bounded by $N$,
$(1+\epsilon)$-approximate distances 
can be reported
in $O(\epsilon^{-1}+\log\log nN)$ time by an oracle of space $O(n\epsilon^{-1}\log(n)\log(nN) )$. The construction time is $O(n(\log n)^3(\log(nN))\epsilon^{-2})$.
Refer to~\cite{Thorup04,KawarabayashiKS11,KawarabayashiST13,GuX19,Wulff-Nilsen16,ChanS19,LeW21}
for tradeoffs in the polylogarithmic terms in approximate distance oracles for \emph{undirected} planar graphs and to~\cite{LeW21} for \emph{directed} planar graphs.

\paragraph{Additional related work} 
Charalampopoulos, Gawrychowski, Mozes, and Weimann showed in ICALP 2021 \cite{CharalampopoulosGM21} an oracle satisfying \cref{thm:main} in the special case when the planar graph is an {\em alignment graph} of two strings. 
The alignment graph is a directed weighted acyclic grid graph that arises in string alignment problems. The simpler structure allowed the authors to adapt and specialize the state-of-the-art oracle of \cite{LongPettie} and, importantly, to develop a novel construction algorithm that runs in $O(n^{1+o(1)})$ time rather than in $O(n^{3/2+o(1)})$ as in \cite{LongPettie}.
Their approach exploits a monotonicity property of the interaction between shortest paths and small cycle separators in the alignment graph. Their algorithm constructs a representation of a Voronoi diagram with $S$ sites in time $\Otild(S)$ using binary search, instead of using the $\Otild(\sqrt{nS})$ algorithm of \cite{CharalampopoulosGMW19}.   
 
Very recently, Boneh, Golan, Mozes, Prigan and Weimann showed in ICALP 2025~\cite{icalp2025planar} how to achieve the same $\Otild(S)$-time construction of Voronoi diagrams in undirected weighted planar graphs by employing a similar (at least at a high level) binary search approach on shortest path separators instead of on small cycle separators as in \cite{CharalampopoulosGM21}. 
They used their Voronoi diagram construction algorithm to construct in $\Otild(n^{4/3})$ time the exact distance oracle with space $\Otild(n^{4/3})$ and query time $O(\log^2 n)$ of \cite{ourJACM}.
However, they were unable to combine their new construction algorithm with the recursive structure of the state-of-the-art oracles of \cite{ourJACM} because their algorithm relies on a data structure that they could not construct efficiently during the recursive calls. 
Our work provides a better approach for constructing Voronoi diagrams via a different binary search based on different structural properties. 
Our approach, which works for directed planar graphs, is arguably simpler than the approach of \cite{icalp2025planar}, and can be made to work (with significant effort) with the recursive structure of the state-of-the-art oracles of \cite{ourJACM} to achieve almost linear construction time.

\paragraph{Voronoi Diagrams}
The state-of-the-art distance oracles for planar graphs rely on point location in Voronoi Diagrams. Consider a planar graph $H$ and with an infinite face $h$ with $S$ vertices. 
The $S$ vertices of $h$ are called the {\em sites} of the Voronoi diagram, and each site $s\in h$ has a weight $\weight(s)\geq 0$ associated with it.
The additive distance between a site $s$ and a vertex $v\in H$, denoted by $\dist^\weight(s,v)$ is defined as $\weight(s)$ plus the length of the $s$-to-$v$ shortest path in $H$. 
The {\em additively weighted Voronoi diagram} $\VD(h,\weight)$ is a partition of the vertices of $H$ into pairwise disjoint sets, one set $\Vor(s)$ for each site $s$. 
The set $\Vor(s)$, called the {\em Voronoi cell} of $s$, contains all vertices in $H$ that are closer (w.r.t.~$\dist^\weight(\cdot,\cdot)$) to $s$ than to any other site.
A point location query for a vertex $v$ returns the site $s$ such that $v$ is in the Voronoi cell $\Vor(s)$.

\begin{figure}[htb]
\begin{minipage}[c]{0.25\textwidth}
    \includegraphics[width=\textwidth]{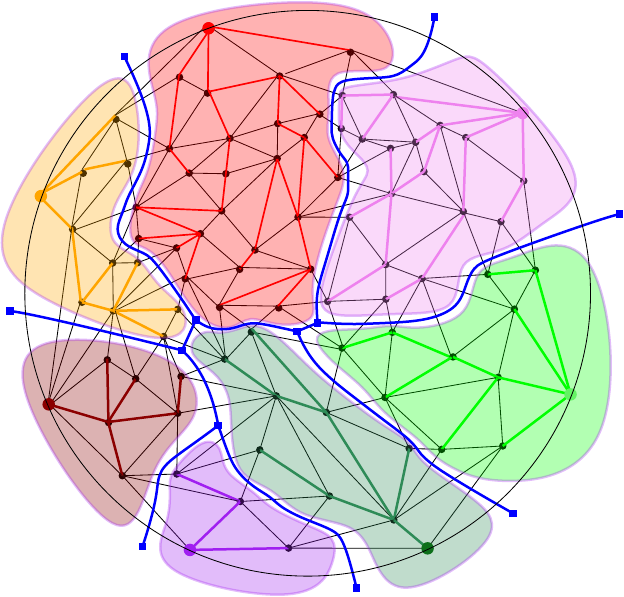}
  \end{minipage}\hfill
  \begin{minipage}[c]{0.7\textwidth}
    \caption{A graph (region) $H$. The vertices of $H$'s infinite face $h$ are the sites of the Voronoi diagram $\VD$. Each site is represented by a unique color, which is also used to shade its Voronoi cell. The dual representation $\VD^*$ is illustrated as the blue tree. The tree has 7 leaves (corresponding to 7 copies of $h^*$) and 5 internal nodes (corresponding to the 5  trichromatic faces of $\VD$). The edges of the tree correspond to contracted subpaths in $H^*$. \label{fig:VD}}
  \end{minipage}
\end{figure}

The main idea for using Voronoi diagrams in  a distance oracle for a graph $G$ is the following.  
Let $u$ be a vertex of $G$ and let $H$ be a subgraph of $G$ with a face $h$ such that $u \notin H$, and every path (in $G$) from $u$ to a vertex of $H$ intersects $h$. 
The oracle stores a Voronoi diagram for $H$ with additive distances $\weight(s) = \dist_G(u,s)$, where $\dist_G(u,s)$ denotes the distance in $G$ from $u$ to $s$. 
To answer a distance query from $u$ to a vertex $v \in H$ the oracle locates the site $s \in h$ such that $v \in \Vor(s)$ using a point location query for $v$. It then returns $\weight(s)$ plus $\dist_H(s,v)$. The former quantity is stored in the Voronoi diagram. The latter quantity is either available directly via an additional data structure, or obtained using recursion.
The various oracles described above differ in the way $G$ is decomposed into subgraphs, in the vertices for which Voronoi diagrams are stored, and in the use of additional data structures or recursion. 

We shall describe in detail the representation of Voronoi Diagrams in the technical part of the paper. For now let us just focus on the key concept of {\em trichromatic faces}. 
We assume, without loss of generality, that each face of the graph $H$ other than $h$ is a triangle. I.e., each face consists of 3 vertices. 
A face (other than $h$) is trichromatic iff its three vertices belong to distinct Voronoi cells.
Computing the representation of a Voronoi diagram essentially boils down to finding its trichromatic faces.
In the case of alignment graphs~\cite{CharalampopoulosGM21}, finding a trichromatic face was done by zooming in on gradually smaller rectangular portions of the grid. 
The crucial structural property of the alignment graph is that the boundary of any rectangle is partitioned by the Voronoi diagram into intervals of consecutive vertices that belong to the same Voronoi cell. 
There is at most one interval per cell, and the order of the intervals along the boundary of the rectangle is the same as the order of the sites along the face $h$. 
Moreover, it suffices to compare the intervals along the left-upper boundary of the rectangle and the lower-right boundary of the rectangle to determine if and which trichromatic faces exist inside a rectangle. 
This structure allows for a binary search procedure that identifies the endpoints of the intervals in $\Otild(S)$ time (recall that $S$ is the number of sites of the VD), and thereby for finding all $S$ trichromatic vertices in $\Otild(S)$ time.

Trying to extend these ideas, \cite{icalp2025planar} replaced the rectangles in the alignment grid with subgraphs bounded by two shortest paths (shortest path separators) in undirected planar graphs. 
The fact that each pair of shortest paths in an undirected graph crosses at most once implies that, under a certain technical restriction, a similar partition of the separator into nicely behaved intervals as in~\cite{CharalampopoulosGM21} exists also in \cite{icalp2025planar}.
This results in an algorithm that after $\Otild(|H|)$ preprocessing time, can compute any additively weighted Voronoi diagram of the $S$ sites of $H$ in $\Otild(S)$ time.
However, enforcing and working under the aforementioned technical restriction makes the algorithm and arguments quite complicated. More importantly, it requires computing certain data structures whose computation time in~\cite{icalp2025planar} takes $\Otild(|H|)$ time.
This suffices when the distance oracle does not use recursion. 
However, the state-of-the-art oracles are recursive, and do not allow for linear preprocessing time at each step of the recursion. 
A possible solution is to devise a way to efficiently reuse the data structures already computed for one level in subsequent recursive levels, but it is unknown whether this idea is applicable to the data structure required in~\cite{icalp2025planar}.

\subsection{Our Results and Techniques}
We had already mentioned that our main result (\cref{thm:main}) is achieved by developing a faster construction algorithm for additively weighted Voronoi diagrams. Using our algorithm, we give nearly optimal constructions for all the distance oracles presented in \cite{CharalampopoulosGMW19,LongPettie,ourJACM}.
Specifically, 
for any $m$ that is in $\omega(1) \cap o(\log n)$,
and any $k$ that is $O(\log n)$
we can construct in 
$n^{1+o(1)}$ time 
the oracles of \cite{LongPettie,ourJACM} with 
$O(m k n^{1 + 1/m + 1/k})$ space 
and $O(2^m k \log^{2}{n} \log\log n)$ query time.
Similarly, we can construct any distance oracle with space $S$ and query time $Q$ developed in \cite{CharalampopoulosGMW19}
in $\Otild(S\cdot Q)$ time. 

It is likely that, beyond settling the main open problem in the field, the techniques we develop here, and the fact that one can now access any pairwise distance in a planar graph almost for free, can lead to progress in other problems that involve distances in planar graphs. 
See for example the immediate implication for dynamic distance oracles we mention at the end of this section. 
Other likely candidate problems to benefit from this are other variants of distance oracles, and, in the longer term, perhaps the diameter problem (which seems to require significant additional ideas). 

It was shown in \cite{icalp2025planar} that computing the representation of a Voronoi diagram for a graph $H$ with a face $h$ with $S$ sites, essentially reduces to computing $\Otild(S)$ Voronoi diagrams for $H$ with just 3 sites each.
A Voronoi diagram with just 3 sites has at most a single trichromatic face.
Our main technical contribution is a completely new algorithm for identifying the (at most) single trichromatic face of a Voronoi diagram whose sites are only 3 of the vertices of the face $h$.

Unlike the previous efficient algorithms for locating trichromatic faces, which used small separators \cite{CharalampopoulosGM21} or shortest path separators \cite{icalp2025planar} to zoom in on the trichromatic face, our algorithm does not use separators at all. 
We use new monotonicity properties to design a new binary-search-based algorithm that 
finds the trichromatic face by eliminating portions of the shortest path tree rooted at one of the three sites. 
We identify a property that can be efficiently tested to determine whether a subtree contains the trichromatic face.
Our algorithm for finding a trichromatic face  is completely different from those devised in \cite{CharalampopoulosGM21,icalp2025planar}.

We provide here a description of the basic approach.
We call the three sites green ($g$), red ($r$), and blue ($b$).
Consider $T_g$, the shortest path tree of $H$ rooted at $g$. It is well known that the edges of $H$ not in $T_g$ form a spanning tree $T^*_g$ of the dual graph $H^*$. 
$T^*_g$ is often called the {\em cotree} of $T_g$. Consider the face $h^*$ as the root of $T^*_g$. 
For every edge $e\in T_g$, its dual edge $e^*$ defines a fundamental cycle $C^*_e$ in $T^*_g$ -- the cycle formed by $e^*$ and the two paths $P^*_1$ and $P^*_2$ in $T^*_g$ that go from the endpoints of $e^*$ towards the root.
We observe that for each of these paths $P^*_j$ ($j \in \{1,2\}$), the part of $P^*_j$ that lies in the green Voronoi cell forms a prefix of $P^*_j$. Hence this prefix can be found efficiently via binary search (it is "easy" to check if a specific vertex or edge are closer to $g$ than to the other two sites).

Our next important observation is that the green vertex of the trichromatic face (i.e., the vertex of the trichromatic face that is in the green Voronoi cell) is in the subtree of the edge $e$ in $T_g$ if and only if the above-mentioned green prefixes $P^*_1$ and $P^*_2$ are non-empty, and one of them leaves the green Voronoi cell into the red cell, while the other leaves the green Voronoi cell into the blue cell.
Thus, we have identified a local decision rule (local in the sense that it only involves inspecting the endpoints of the prefixes of $P_1^*$ and $P_2^*$) that allows us to eliminate one of the two subtrees one would obtain from $T_g$ by deleting the edge $e$.
This implies that the green vertex of the trichromatic face can be found using a binary search approach\footnote{We note that using tree centroids for binary search (centroid search) is a common technique that has also been used in~\cite{CharalampopoulosGMW19,ourJACM}. However, that use is in a different context (the query algorithm, not the construction), on a different tree (the dual Voronoi diagram, not the shortest path tree in the primal graph), and using different structural properties.} --  choose $e$ to balance the size of $T_g$ on either side of $e$ (a centroid), eliminate one side of $T_g$ using the local decision rule, and repeat the process on the side of $T_g$ that was not eliminated. 

Turning this into an efficient algorithm mainly depends on our claim that it is "easy" to check if a specific vertex is closer to $g$ than to the other two sites, and on our ability to efficiently perform random access of edges of the dual paths $P^*_1$ and $P^*_2$. 
If we are allowed to preprocess $H$ in $\Otild(|H|)$ time, then the MSSP data structure~\cite{Klein05, CabelloCE13} provides both operations with respect to the shortest path tree rooted at any site on the face $h$ in $O(\log|H|)$ time, and then finding the trichromatic face of any 3 sites and any set of additive weights takes $O(\log^2|H|)$ time. 
This simple and elegant approach, which is the basis of our result, is presented in \cref{sec:simple}.

Things are much more complicated when we cannot afford to invest $\Otild(|H|)$ preprocessing time. This is the case when constructing Voronoi diagrams for the complements of regions in a recursive decomposition of a graph $G$ into gradually smaller regions $H$, which is the recursive structure used in the state-of-the-art oracles of \cite{CharalampopoulosGMW19,LongPettie,ourJACM}. 
To overcome this issue we devise a mechanism that implements the tree elimination and fundamental dual cycle approach on coarse versions of the shortest path trees and their duals, at different levels of granularity. 
To achieve this we identify and exploit additional structure of the relationship between trees (both primal and dual) at different levels of coarseness, as well as of the relationship between coarse dual trees and Voronoi diagrams at different levels of the recursion. 
Interestingly, we are able to solve some of the problems that arise in this scheme using the same tree elimination approach described above, but applying it to the problem of finding a prefix of a dual path, rather than just to the original problem of finding a trichromatic face. 
The resulting recursive algorithm is quite involved. It finds a single trichromatic face in $(\log{n})^{O(m)}\cdot Q$ time, where $m$ is the depth of the recursive decomposition, and $Q$ is the query time of the oracle we are trying to construct. This algorithm is presented in \cref{sec:complete}.
This time dependency on $m$ gives $n^{1+o(1)}$ construction time for the entire oracle whenever 
$m=o(\log n / \log\log n)$, which includes all the tradeoffs in \cite{CharalampopoulosGMW19} (where $Q=O(\log^{m}{n})$), and part of the range of the tradeoffs in \cite{ourJACM} (where $Q=\Otild(2^{m})$).  
For the remainder of the range of tradeoffs in \cite{ourJACM}, i.e., $m \in \Omega(\log n / \log\log n) \cap o(\log n)$, we are also able, with some additional effort, to obtain $n^{1+o(1)}$ construction time. See \cref{sec:jacm_complete} for the details. 
In \cref{sec:jacm_complete} we also describe how to efficiently construct other parts of the oracles of~\cite{ourJACM} (i.e., parts that are not Voronoi diagrams). The construction of these parts poses no significant technical challenges, but must be included for completeness.

\paragraph{Application to dynamic {\em directed} distance oracles} It was pointed out in~\cite{icalp2025planar} that an $\Otild(1)$ time-algorithm for finding the trichromatic face of a 3-site Voronoi diagram implies a {\em dynamic} distance oracle that supports $\tilde O(1)$-time {\em directed} distance queries from a single-source $s$, and $\tilde O(n^{2/3})$-time updates consisting of edge insertions, edge deletions, and  changing the source $s$.
The dynamic distance oracle is obtained by simply plugging the Voronoi diagram construction algorithm into the oracle of Charalampopoulos and Karczmarz \cite{charalampopoulos2022ssp} (thus improving its update time from $\tilde O(n^{4/5})$ to $\tilde O(n^{2/3})$). 
The algorithm of~\cite{icalp2025planar} implied $\Otild(1)$-query and $\Otild(n^{2/3})$-time updates for undirected planar graphs, and our result now implies that same for directed planar graphs. 
As was also pointed out in \cite{icalp2025planar}, 
this new oracle achieves the same bounds as the current best {\em all-pairs} dynamic oracle by   Fakcharoenphol and Rao \cite{FakcharoenpholR06} but has the benefit that, between source changes, each query takes only $\tilde O(1)$ time instead of $\Otild(n^{2/3})$ time. It is known that achieving $\Otild(1)$-query and $O(n^{1/2-\epsilon})$-time updates for any constant $\epsilon>0$ is would refute the APSP conjecture~\cite{AbboudD16}.

\section{Preliminaries}

Throughout the paper we consider as input a weighted directed planar graph $G=(V(G),E(G))$, embedded in the plane. 
To be concrete, we assume the graph is given by adjacency lists. This representation can be converted into a combinatorial embedding of the graph in linear time \cite{hopcroftTarjanPlanarity1974}.
We use $|G|$ to denote the number of vertices in $G$. Since planar graphs are sparse, $|E(G)| = O(|G|)$ as well.
The dual of a planar graph $G$ is a planar graph $G^*$ whose vertices are the faces of $G$. Edges of $G^*$ are in one-to-one correspondence with edges of $G$; 
there is an edge $e^*$ from vertex $f^*$ to vertex $g^*$ of $G^*$
if and only if the corresponding faces $f$ and $g$ of $G$ are to the left and right of the edge $e$, respectively.
Thus, every edge $e$ in $G$ has a
corresponding dual edge $e^*$ in $G^*$.
Geometrically, we may consider both the primal graph $G$ and its dual $G^*$ as embedded as curves in the same plane.
The embeddings of a primal edge $e$ and its dual edge $e^*$ cross each other at a single point in the plane.
We may therefore refer to curves that are composed by concatenating a path $P$ of primal edges that ends with a primal edge $e$ and a path $Q^*$ of dual edges that starts with the dual edge $e^*$.
The concatenated curve starts along the primal path $P$ until the intersection point of $e$ and $e^*$, and then continues along the dual path $Q^*$.

We assume that the input graph has no negative length cycles.
We can transform the graph in a standard way in $O(n \frac{\log^2 n}{\log \log n})$ time~\cite{MozesW10}, so that all edge weights are non-negative and distances are preserved.
With another standard transformation we can guarantee that each vertex has constant degree and that the graph is triangulated, while distances are preserved and without increasing asymptotically the size of the graph.\footnote{We may assume the graph is simple (remove self loops and all parallel edges except the lightest one in each set of parallel edges). Replace each high degree vertex with a binary tree with zero length bidirectional edges. In the resulting graph the degree of each vertex is at most 3. Then, for every face $f$ that is not a triangle, let $v_1, v_2, v_3, v_4 \dots v_k$ be the vertices of $f$ in cyclic order. Add infinite length edges $v_2v_k, v_kv_3, v_3v_{k-1}, v_{k-1}v_4 \dots$. This replaces $f$ with $O(k)$ triangles, while increasing the degree of each vertex by a constant.}
We further assume that shortest paths are unique; this can be ensured in $O(n)$ time by a deterministic perturbation of the edge weights~\cite{EricksonFL18}.

For a rooted tree $T$, we use $T(v)$ to denote the subtree of $T$ rooted at $v$, and $T(u,v)$ to denote the path in $T$ from a vertex $u$ to its descendant $v$. 
Let $T$ be a spanning tree in a planar graph $G$. 
The edges of $G$ not in $T$ form a spanning tree of the dual graph $G^*$. This tree is denoted by $T^*$ and is called the {\em cotree} of $T$. We consider $T^*$ as a tree rooted at the infinite face of $G$.
For an edge $e \in T$, the {\em fundamental cycle} of $e^*$ w.r.t. $T^*$ is the cycle $C^*$ composed of $e^*$ and of the two rootward paths $P^*_1,P^*_2$ in $T^*$ from the endpoints of $e^*$ to their LCA $x^*$ in $T^*$.
The embedding of $G^*$ specifies a clockwise cyclic order on the edges incident to every vertex.
We say that $P^*_1$ is right of $P^*_2$ if, starting with the parent edge leading from the root of $T^*$ to the LCA $x^*$, we first encounter an edge of $P^*_1$ and then an edge of $P^*_2$. (To define the directions at the root we imagine embedding an artificial parent edge in the infinite face that enters the root.)

Since the graph has bounded degree, every spanning tree admits an {\em edge centroid decomposition}, defined as follows.
\begin{definition}[edge centroid decomposition]\label{def:centroidDecomposition}
Let $T$ be a tree with $n$ vertices and let $1<k\leq 2$ be some constant. A centroid of $T$ is an edge $e$ such that if $e$ is removed from $T$ then each resulting connected component has at most $\lfloor n/k \rfloor$ edges. A centroid decomposition of $T$ is a rooted binary tree $B$ defined recursively as follows:
\begin{itemize}
    \item If $T$ consists of a single edge $e$, then $B$ consists of a single node that corresponds to the edge $e$.
    \item Otherwise, let $e$ be a centroid of $T$, and let $T_1, T_2$ be the connected components of $T \setminus \{e\}$. The root of $B$ is a node corresponding to the edge $e$, and the two subtrees of the root of $B$ are the centroid decomposition trees for each $T_i$.
\end{itemize}
The centroid decomposition tree $B$ has height 
 $O(\log n)$.
\end{definition}

\paragraph{Multiple-source shortest paths.} For a planar graph $G$ with a distinguished face $h$, the multiple-source shortest paths (MSSP) data structure~\cite{Klein05,CabelloCE13} represents all shortest path trees rooted at the vertices of $h$, as well as all of their cotrees, using  persistent dynamic trees.
It can be constructed in $O(n \log n)$ time, requires $O(n\log n)$ space, and can support in $O(\log n)$ time standard tree operations on any of these trees such as reporting the distance from the root to any vertex, reporting whether a vertex is an ancestor of another vertex, reporting whether two vertices are left/right to one another in the tree, reporting an ancestor of a node at a specified depth, supporting leafmost/rootmost marked ancestor queries, and reporting the LCA of two nodes (cf. ~\cite[Theorem 18.0.2]{PlanarBook}). 
A variant parameterized by $k \geq 1$~\cite[Lemma 2.1]{ourJACM} achieves a construction time of $O(k n^{1+1/k})$, uses $O(k n^{1+1/k})$ space, and supports the same operations in $O(k \log\log n)$ time.
Furthermore, the MSSP structure can support maintaining a centroid decomposition of the shortest path trees \citep{goodrich1998dynamic}.

\paragraph{Separators and recursive decompositions.}

An $r$-division~\cite{Frederickson87} of a planar graph, for $r \in [1,n]$, is a decomposition of the graph into $O(n/r)$ regions, each of size $O(r)$, such that each region has $O(\sqrt{r})$ boundary vertices, i.e.~vertices that also belong to another region.
Another desired property of an $r$-division is that the boundary vertices lie on a constant number of faces of the region (called holes).
We denote the boundary vertices of a region $R$ by $\partial R$. 

It was shown in~\cite[Theorem 4]{KleinMS13} that, given a strictly decreasing sequence of numbers $(r_m, r_{m-1}, \ldots, r_1)$, where $r_1$ is a sufficiently large constant, and $r_m=n$, we can obtain the $r_i$-divisions for all $i$ in time $O(n)$ in total.
For $j<i$, the $r_j$-division is a refinement of the $r_i$-division.
For convenience we define the only region in the $r_m$ division to be $G$ itself, and the $r_0$ division to be single-vertex regions. 
We denote by $\mathcal{R}$ the set of all regions at all levels of this recursive decomposition, and by $\mathcal{R}_i$ the set of all regions of the $r_i$-division. 

We assume that  the boundary vertices of each region lie on a single hole (face of the region that is not a face of the original graph) that is a simple cycle. The outside of this hole with respect to a region $R$ is $G\setminus(R \setminus \partial R)$ and to simplify notation we denote it by $R^{out}$. 
The single hole assumption is not valid in general. 
However, for our case it can be made without loss of generality.
It is shown in \cite[Section 4.3]{CharalampopoulosGMW19} and \cite[Section 9]{ourJACM} that for the of oracles considered there (which are the oracles we construct in this paper), the case of multiple holes that are not necessarily simple cycles reduces, with no asymptotic overhead, to the single-hole case. The exact same reduction applies for the results in this paper. 

For $i<m$ and a region $R \in \mathcal R_i$, the parent of $R$ is the region $R' \in \mathcal R_{i+1}$ that contains $R$.

\paragraph {Voronoi diagrams.}
Recall that we assume that $G$ is connected, triangulated, and has constant degree. 
Let $H$ be a subgraph of $G$ with an infinite face $h$ such that $h$ is the only face of $H$ that is not necessarily a triangle.
Let $x$ be a vertex of $G$ such that $x \notin H$ or $x\in h$.
Note that every path (in $G$) from $x$ to a vertex of $H$ intersects $h$.
For every vertex $v \in H$, let $\weight_x(v)$ denote the distance $x$-to-$v$ in $G$.

\begin{definition}
The Voronoi diagram of $x$ within $H$ with additive weights $\weight_x(\cdot)$, denoted $\VD[H,h,\weight_x]$,
is a partition of $V(H)$ into pairwise disjoint sets, one set $\Vor(s)$ for each site $s \in h$. The set $\Vor(s)$, which is called the Voronoi cell of $s$, contains all vertices in $V(H)$ that are closer (w.r.t. additive weights) to $s$ than to any other site in the set of sites $S$.
Ties are always broken consistently, in favor of sites with larger additive weights --- formally, ties are broken with respect to reverse lexicographic order on $(\weight_x(s),s)$.
\end{definition}

We say that an edge $e$ is strictly inside $\Vor(s)$ of $\VD$ (the Voronoi diagram $\VD[H,h,\weight_x]$) if both primal vertices incident to $e$ are in $\Vor(s)$. Similarly, a face $f$ is strictly inside $\Vor(s)$ of $\VD$ if all 3 vertices incident to $f$ are in $\Vor(s)$.

\begin{lemma} [Lemma 2.1 of~\cite{GawrychowskiKMS18}] \label{lem:VD_shortest_path_color}
    For each $s\in h$, the vertices in $\Vor(s)$ form a connected subtree (rooted at $s$) of the shortest path tree $T_s$ of $H$.
\end{lemma}

Following \cite{CharalampopoulosGMW19}, we use a dual representation $\VD^*$ (i.e. $\VD^*[H, h, w_x]$) of $\VD$ as a tree with $O(|h|)$ vertices (see \cref{fig:VD}): 
Let $H^*$ be the planar dual of $H$. 
Let $\VD^*_0$ be the subgraph of $H^*$ consisting of the duals of edges $uv$ of $H$ such that $u$ and $v$ are in different Voronoi cells. 
Let $\VD^*_1$ be the graph obtained from $\VD^*_0$ by repeatedly contracting edges incident to degree-2 vertices. 
Finally, we define $\VD^*$ to be the tree obtained from $\VD^*_1$ by replacing the node $h^*$ by multiple copies, one for each edge incident to $h^*$. 
The vertices of $\VD^*$ (except the copies of $h^*$) are called {\em Voronoi vertices}, and the edges of $\VD^*$ are called {\em Voronoi edges}.
By definition, each Voronoi edge $\tilde e^*$ of $\VD^*$ corresponds to some path $P^*$ in $H^*$ that was contracted. We say that the edges of $P^*$ are {\em represented} by $\tilde e^*$, and that $\tilde e^*$ is a {\em coarsening} of $P^*$. We associate with $\tilde e^*$ the first and last edges of $P^*$. 
Each Voronoi vertex is dual to a face whose three vertices belong to three different Voronoi cells. 
We call such a face (and its corresponding dual vertex) {\em trichromatic}.
The complexity (i.e., the number of vertices and edges) of $\VD^*$ is $O(|h|)$. 

An important special case is a VD of just two sites $s$ and $t$. In such a VD there are no trichromatic faces. 
In this case, $\VD^*_0$ is just cycle in $H^*$ whose corresponding primal edges form an $st$-cut in $H$. 
We call this cycle the $st$-{\em bisector} and denote it by $\beta^*(s,t)$. 
In a Voronoi diagram with many sites, a Voronoi edge $\tilde e$ separating the cells of sites $s$ and $t$ is a coarsening of a subpath of the bisector $\beta^*(s,t)$. The Voronoi diagram stores a bi-directional mapping between the Voronoi edges $\tilde e^*$ and their corresponding pairs of sites $s$ and $t$.

Another important special case is a VD of three sites (we call this a {\em trichromatic VD}). 
Such a VD has at most one trichromatic face (in addition to face $h$ itself). 
Note that, in a VD with more than two sites,  a trichromatic vertex is a meeting point of the three bisectors between the corresponding pairs of vertices.

Let $T_x$ be the shortest path (in $G$) rooted at $x$.
Let $T^*_x$ be the cotree of $T_x$.
That is, $T^*_x$ is a spanning tree of $G^*$ that consists of all the edges of $G^*$ whose duals are not in $T_x$.
It would be useful for our purposes to relate $\VD^*$, the dual Voronoi diagram of $\VD[H,h,\weight_x]$, and the cotree $T^*_x$ of $T_x$.\footnote{While the definition of the dual Voronoi diagram $\VD^*$ is not new to this work, the characterization of the relations between $\VD^*$ and $T^*_x$ is new.} 
\begin{lemma}\label{lem:VD0-Tx}
    The edges of $\VD^*_0$ that do not belong to $T^*_x$ are exactly the set of edges dual to $\{uv \in T_x \cap H : v \in h\}$.
\end{lemma}
\begin{proof}
Note that the tie breaking rule in the definition of the Voronoi diagram implies that each vertex $v \in H$ belongs to the Voronoi cell of the last site on the shortest path from $x$ to $v$ in $T_x$.  
     Let $e^*$ be an edge of $\VD^*_0$. That is, $e=uv$ is such  that $u$ and $v$ are in distinct Voronoi cells. 
    If $e^* \notin T^*_x$ then $e \in T_x$, but then $u$ and $v$ are in the same Voronoi cell, unless, by the tie breaking rule, $v \in h$. 
\end{proof}

\begin{figure}[H]
    \centering
    \begin{minipage}{0.28\textwidth}
        \centering
        \includegraphics[width=\linewidth]{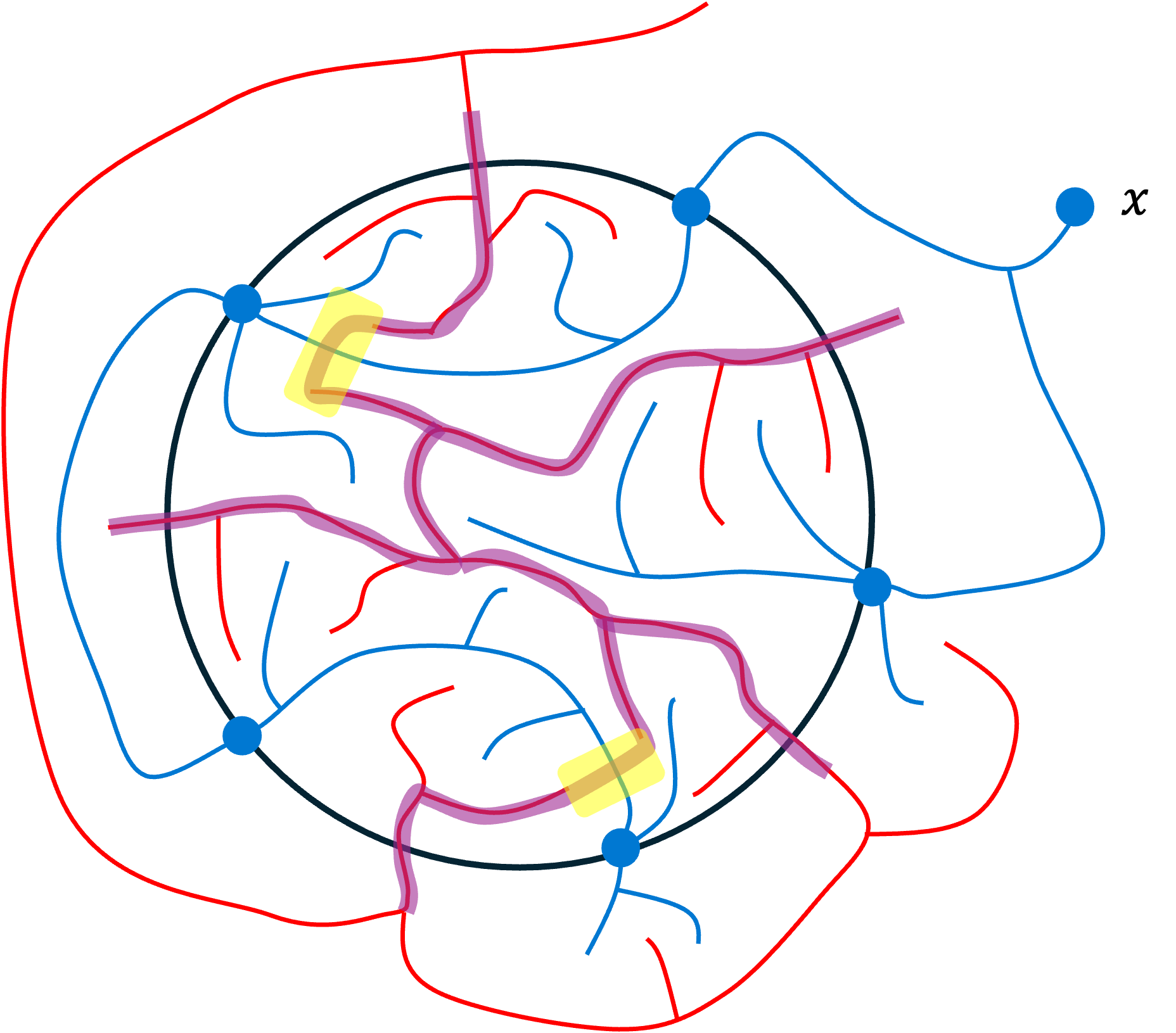}
    \end{minipage}%
    \hfill
    \begin{minipage}{0.7\textwidth}
        \captionof{figure}{%
            A schematic diagram for Lemma~\ref{lem:VD0-Tx}. 
            The shortest path tree $T_x$ is shown in blue, and its dual co-tree $T^*_x$ is shown in red. 
            The distinguished face $h$ is depicted as a black cycle, with five designated sites on $h$ indicated by solid blue nodes.
            The subgraph $\VD^*_0$ is highlighted in thick purple. 
            Edges of $\VD^*_0$ that do not belong to $T^*_x$ are additionally highlighted in yellow. The duals of these edges are exactly the edges of $T_x$ that enter a blue site.
            }
        \label{fig:VD0_tree-forest-edges}
    \end{minipage}
\end{figure}

\begin{lemma} \label{cor:VD-Tx}
An edge of $\VD^*$ that does not represent the dual to any edge of $T_x$ represents a path in $T^*$.

An edge of $\VD^*$ that represents a dual to an edge $e \in T_x$ represents a path of dual edges that comprises of two (possibly empty) subpaths of $T^*_x$ and of the edge $e^*$.
\end{lemma}
\begin{proof}
    Let $\tilde e^*$ be an edge of $\VD^*$. Let $\beta^*(s,t)$ be the bisector whose subpath is represented by $\tilde e^*$.
    For every edge $d^*$ that is represented by $\tilde e^*$, one endpoint of $d$ is in $\Vor(s)$ and the other is in $\Vor(t)$. 
    Hence, if $d \in T_x$ then $d \notin T^*_x$, so, by \cref{lem:VD0-Tx}, one endpoint of $d$ must be either $t$ or $s$.
    It follows that there is a path in $T_x$ between $s$ and $t$ that includes $d$.
    Since there is at most a single path between $s$ and $t$ in $T_x$, there can be at most one such edge $d^*$.
    Since, by \cref{lem:VD0-Tx} all the other edges represented by $\tilde e^*$ are edges of $T^*_x$, the lemma follows.
\end{proof}

\paragraph{Point Location in Voronoi Diagrams} \label{par:point_location_discussion} 
A point location query for a vertex $v$ in a Voronoi diagram stored for some vertex $x$ returns the site $s$ s.t. $v \in \Vor(s)$, and the additive distance from $s$ to $v$. 
In~\cite{ourJACM}, point location queries are implemented in the context of a recursive decomposition of $G$ by a recursive procedure called {\sc CentroidSearch}~\cite[Algorithm~4]{ourJACM}.
Let $R_i$ denote the region of $\mathcal{R}_i$ that contains $x$. 
The point location procedure of~\cite{ourJACM} actually identifies (and can return), for each level $i$ of the recursive decomposition such that the shortest $x$-to-$v$ path crosses $\partial R_i$, the last site $s_i \in \partial R_i$ along the shortest path to $v$, together with the additive distance from $s_i$ to $v$.

\paragraph{Components of the oracle of \cite{ourJACM}}

To establish \cref{thm:main} our algorithm efficiently computes all components of the distance oracle of \cite{ourJACM}.
We now describe components (A)–(C), which are the relevant parts for the computation of the Voronoi diagrams used by the oracle.
The remaining components (D)–(E) are not related to the construction of the Voronoi diagrams. 
We defer to \cref{sec:jacm_complete} the description of these components and their efficient construction, which does not pose significant technical challenges.

The oracle works with a recursive $r$-division with $m$ levels, as described above, with $r_i = n^{i/m}$. For the regions in this recursive decomposition the oracle includes the following parts:

 \begin{enumerate}
    \item[(A)] \textbf{MSSP Structures.} 
    For each $i\in[0,m-1]$ and each region $R_i\in\mathcal{R}_i$ with parent $R_{i+1}\in\mathcal{R}_{i+1}$, 
    the oracle includes an MSSP data structure for 
    the graph $R_i^{out}$, and source set $\partial R_i$.
    However, the structure only answers queries 
    for $s\in\partial R_i$ and $u,v \in R_i^{out} \cap R_{i+1}$.
    Rather than represent the \emph{full} SSSP tree from each root on $s\in \partial R_i$, 
    this MSSP data structure only stores the tree induced by $R_i^{out} \cap R_{i+1}$, i.e.,
    the parent of any vertex $v\in R_i^{out} \cap R_{i+1}$ is its nearest ancestor $v'$
    in the SSSP tree such that $v' \in R_i^{out} \cap R_{i+1}$.
    If $v'v$ is a "shortcut" edge corresponding to a path in $R_{i+1}^{out}$, it has weight equal to the length of that path.
    Cf. \cref{fig:PartA_SSSP-Tree} for the SSSP tree and its shortcut edges.

\begin{figure}[H]
    \centering
    \begin{minipage}{0.32\textwidth}
        \centering
        \includegraphics[width=\linewidth]{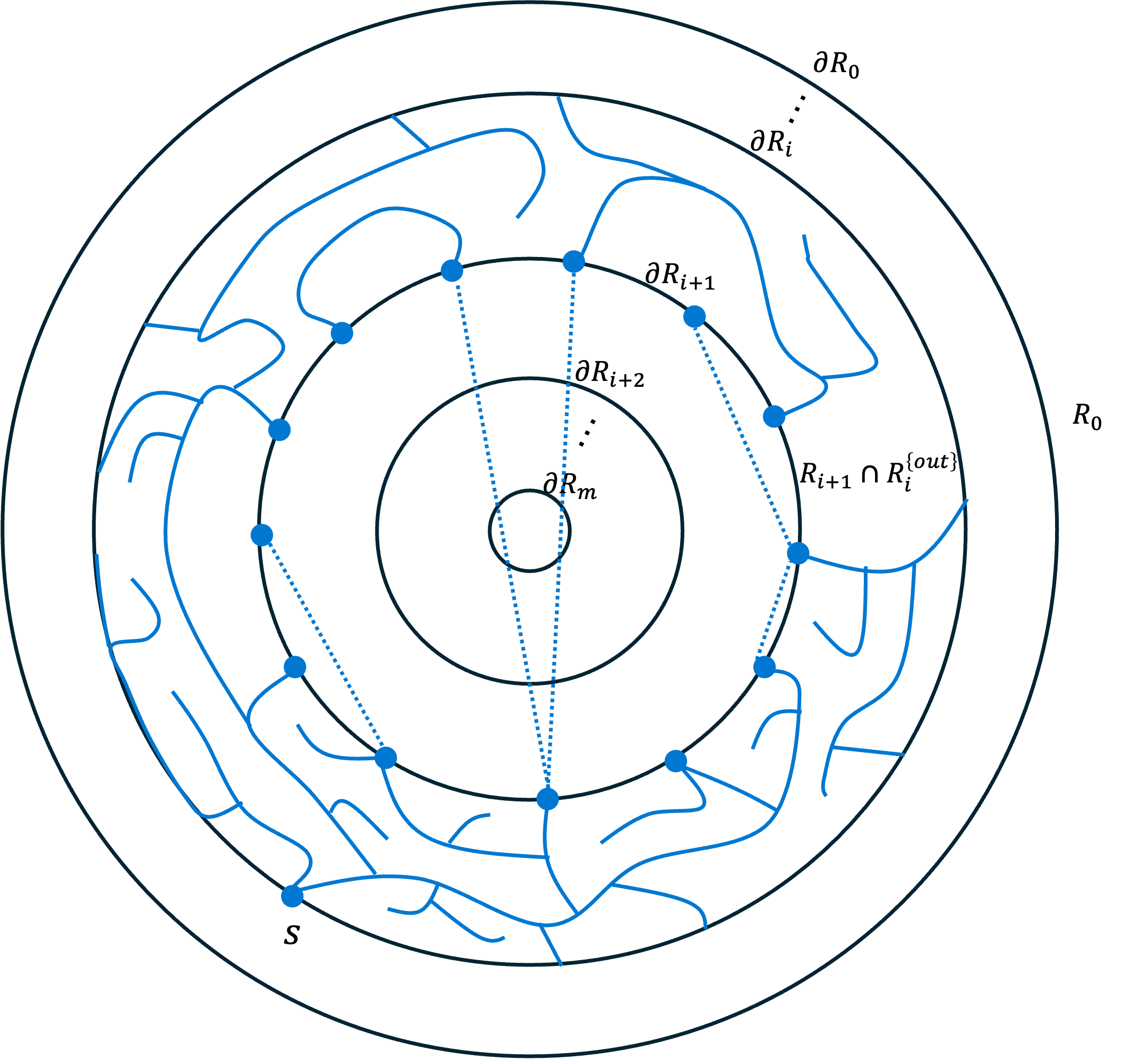}
    \end{minipage}%
    \hfill
    \begin{minipage}{0.66\textwidth}
        \small
            The diagram shows a sequence of regions $\{R_j\}_{j=0}^m$ at different levels of the recursive decomposition, where $R_0$ is some region at the finest level $\mathcal R_0$ of the recursive decomposition, and $R_{j+1}$ is the region of $\mathcal R_{j+1}$ that contains $R_j$.
            The regions' boundaries $\partial R_j$ are depicted as concentric cycles.
            Somewhat non-intuitively, it is easier to depict an embedding in which each region $R_j$ is embedded outside $\partial R_j$;  for instance, $R_0$ corresponds to the infinite face of the diagram, and $R_0^{out}$ corresponds to the disk bounded by $\partial R_0$.
            A vertex $s \in \partial R_i$ is indicated explicitly, and the vertices of $\partial R_{i+1}$ are shown as blue nodes.
            The SSSP tree rooted at $s$ in $R_i^{\text{out}}$, is shown in blue lines. Regular edges in $R_{i}^{out} \cap R_{i+1} $ are indicated by solid blue edges, while shortcut edges in $R_{i+1}^{out}$ are indicated by dotted blue lines.
            Note that the embedding of the shortcut edges does not capture the corresponding shortest path in  $R_{i+1}^{out}$ (neither in the diagram, nor in the MSSP data structure of part (A)).
        
    \end{minipage}
    \captionof{figure}{%
    \label{fig:PartA_SSSP-Tree}
            A schematic diagram of an SSSP tree in part (A) of the oracle from~\cite{ourJACM}.}
\end{figure}

  \item[(B)] \textbf{Voronoi Diagrams.} For each level $i\in[0,m-2]$, for every region $R_i\in\mathcal{R}_i$ with parent $R_{i+1}\in\mathcal{R}_{i+1}$ and for every vertex $q\in\partial R_i$, the oracle includes the dual representation of the Voronoi diagram
  \[
  \VD^*(q,R^{out}_{i+1}) = \VD^*\Bigl[R_{i+1}^{\mathrm{out}},\,\partial R_{i+1},\,\weight_q\Bigr],
  \]
  where the additive weight function is defined by
  \[
  \weight_q(s)=\dist_G(q,s)\quad \text{for } s\in\partial R_{i+1}.
  \]
  
  \item[(C)] \textbf{Additional Voronoi Diagrams.} For each level $i\in[1,m-1]$, for every region $R_i\in\mathcal{R}_i$ and every $q\in\partial R_i$, the oracle also includes
  \[
  \VD^*(q,R_i^{out}) = \VD^*\Bigl[R_i^{\mathrm{out}},\,\partial R_i,\,\weight_q\Bigr],
  \]
  with $\weight_q(s)=\dist_G(q,s)$.
  
\end{enumerate}

\section{Computing a Voronoi Diagram of Three Sites } \label{sec:simple}

We present an efficient algorithm for identifying the trichromatic face $\tilde f$ of a Voronoi diagram with 3 sites $r,g,b$ incident to a face $h$ in a graph $H$. 
As we mentioned in the introduction, this is the basis of the results in this paper since computing a Voronoi diagram with $S$ sites reduces to finding the trichromatic face of $\Otild(S)$ Voronoi diagrams with just 3 sites each.

It was shown~\cite{Cabello19} that such a Voronoi diagram with just 3 sites has at most a single trichromatic face (other than $h$ itself).
To avoid clutter we assume that the trichromatic face indeed exists, but our algorithm does identify the case that there is no trichromatic face. 
We also assume in this section that we are given an MSSP for the sites $r,g,b$ with respect to $H$, which represents the shortest path trees $T_r, T_g, T_b$, rooted at the sites, and answers distance 
queries on the primal trees, and level ancestor queries on their cotrees $T^*_r, T^*_g, T^*_b$ in $D$ time per query (when the trees are represented using persistent link-cut trees, $D=O(\log n)$). This assumption will be removed in \cref{sec:complete}.

\begin{theorem}\label{thm:SimpleVoronoiSite}
Let $H$ be a directed planar graph with real arc lengths, $n$ vertices, and no negative length cycles. 
Let $r, g, b$ be three sites that lie on a single face $h$ of $H$. Assume that all faces of $H$ except for $h$ have size 3. Given an MSSP for the sites in $H$,
one can find the trichromatic face $\tilde f$ of $\VD^*[H, \{r, g, b\},\omega]$ in $O(\log^2(|H|)\cdot D)$ time.
\end{theorem}

We prove~\cref{thm:SimpleVoronoiSite} by providing an algorithm, described in \cref{alg:s_centsearch}. 
We say that a vertex $v$ is {\em green} if $v$ is closer to the site $g$ (w.r.t. additive weights $\omega$) than to the sites $b,r$.
Note that checking whether a vertex is green can be done in $O(D)$ time using 3 MSSP distance queries, one from each site.
We say that an edge $e$ of $H$ is green if both endpoints of $e$ are green.
We say that a dual edge $e^*$ of $H^*$ is green if the primal edge $e$ is green.

To find the trichromatic face $\tilde f$, the algorithm finds the green vertex of $\tilde f$ by successively eliminating portions of the shortest path tree $T_g$ rooted at the green site $g$ using a centroid decomposition of $T_g$.
For each centroid edge $e=(uv)$, the algorithm determines whether $T^e \stackrel{\text{def}}{=}  T_g(v) \cup e$, the union of the centroid edge and the subtree of $T_g$ rooted at the leafward endpoint $v$ of $e$, contains an edge incident to the green vertex of $\tilde f$ or not. If it does, the algorithm recurses on $T^e$. Otherwise it recurses on the subtree $T_g \setminus T_g(v)$. 

We next describe how to decide which subtree of $T_g$ to recurse on. 
The subtree $T_g(v)$ induces a partition of $H$ into two connected components; The subgraph of $H$ induced by the vertices of $T_g(v)$, and the subgraph of $H$ induced by the remaining vertices of $H$.  
The two components are separated by the fundamental cycle of $e^*$ with respect to $T^*_g$, which is comprised of two rootward dual paths $P^*_1, P^*_2$.

\paragraph{Critical edges}\label{par:critical_edge_definition}
We prove in \cref{lem:TwoColorMonotonicity} that each of the two rootward dual paths $P^*_j$ has a (possibly empty) prefix of green edges. Furthermore, no edge of the remaining suffix of $P^*_j$ is green.
The procedure {\sc SimpleFindCritical}, described in \cref{alg:s_CriticalEdges}, utilizes this monotonicity property to identify, using binary search on $P^*_j$, the first non-green edge of $P^*_j$, if it exists. We call this edge the {\em critical edge} of $P^*_j$. Together we call the (at most) two critical edges, one of $P^*_1$ and the other of $P^*_2$, the critical edges of $e$ w.r.t. $T_g$. Refer to \cref{fig:VoronoiVertex_positive_a} for and illustration.

We can decide which subtree of $T_g$ to recurse on by inspecting the critical edges. However, since in degenerate cases the critical edges may not exist, we must consider several cases.\footnote{To assist the reader, we note that case (1) is the natural case to consider, and the others may be considered edge cases.  
Understanding just case (1) gives the main intuition for the algorithm and its correctness.} 
Let $q^*$ denote the common dual vertex of $P^*_1$ and $P^*_2$. That is, $q^*$ is the LCA of the endpoints of $e^*$ in $T^*_g$.

\begin{enumerate}
    \item Both critical edges exist.
    \item Exactly one critical edge exists (w.l.o.g. only $e_1$ exists) and $q^*= h^*$.
    \item Exactly one critical edge exists (w.l.o.g. only $e_1$ exists) and $q^* \neq h^*$.
    \item None of the critical edges exists and $q^* = h^*$.
    \item None of the critical edges exists and $q^* \neq h^*$.
\end{enumerate}

Choosing which subtree to recurse on is done by first examining if $u$ (the rootward endpoint of $e$) is a green vertex. If not, then the subtree $T^e$ contains no green vertices nor green edges, and we continue on $T \setminus T_g(v)$.
Else, $v$ is a green vertex, and we evaluate the following simple decision rule {\sc TrichromaticDecisionRule} (\cref{alg:triDec}). The recursion should continue on $T^e$ if exactly one of the following conditions is satisfied:
\begin{enumerate}[itemsep=0pt, topsep=4pt]

    \item Both critical edges exist. Among the primal endpoints of the critical edges $e_1,e_2$ of $P^*_1$ and $P^*_2$, there is a red vertex, a green vertex and a blue vertex.
    Cf. \cref{fig:VoronoiVertex_positive_a}.
    
    \item Only $e_1$ exists, and $q^* = h^*$.
    Let $e'_2\in h$ be the primal edge of the last dual edge of $P^*_2$ leading to $h^*$. Let $s_1$, $s_2$ be the sites nearest to $e'_2$ along $h$.
    In the set of vertices containing $s_1$, $s_2$ and the primal endpoints of $e_1$, there is a red vertex, a green vertex and a blue vertex.
    Cf. \cref{fig:VoronoiVertex_positive_b}.

    \item Just $e_1$ exists, and $q^* \neq h^*$.
    Let $S_1$ be the set of two primal endpoints of $e_1$.
    Let $S_2$ be the set of three primal vertices incident to $q^*$.
    In the set $S_1 \cup S_2$ there is a red vertex, a green vertex and a blue vertex. The diagram shown in \cref{fig:VoronoiVertex_positive_a} illustrates this case as well. (Note that the diagram does not indicate where $P^*_1$ and $P^*_2$ meet, so for this case consider that $q^*$ is either on the red-green bisector or on the blue-green bisector, so one of the $P^*_j$'s is entirely green.)

    \item Neither critical edge exists, and $q^* = h^*$. 
    Let $e'_1, e'_2\in h$ be the primal edges of the last dual edges of $P^*_1$, $P^*_2$ leading to $h^*$, respectively. In the set of sites containing the two nearest to $e'_1$ along $h$, and the two nearest to $e'_2$ along $h$, there is a red vertex, a green vertex and a blue vertex.
    Cf. \cref{fig:VoronoiVertex_positive_c}.
\end{enumerate}
Otherwise, the recursion should continue on $T \setminus T_g(v)$.
 
When \cref{alg:s_centsearch} reaches a subtree that consists of a single edge $e=uv$, it concludes that the trichromatic face $\tilde f$, if exists, is incident to either the vertex $u$ or $v$.\footnote{We terminate the search at a single edge rather than a single vertex to make the description in this section more similar to the description in \cref{sec:complete}.} The procedure then returns the edge and finally we test the $O(1)$ neighboring faces of both vertex candidates, $u$ and $v$, using the MSSP, and either find the trichromatic face $\tilde f$, or determine that $\tilde f$ does not exist.

The tree elimination procedure on $T_g$ consists of $O(\log|H|)$ steps, each invoking a binary search that performs a total of $O(\log |H|)$ queries to the MSSP.\footnote{
At each inspected subtree, computing the decision rule {\sc TrichromaticDecisionRule} (\cref{alg:triDec}) may require identifying the nearest site among $\{r, g, b\}$ to an edge $e'\in h$ at the end of $P^*_1$ or $P^*_2$. 
To support this we assign in the preprocessing stage to each vertex and edge on $h$ its index in the cyclic order along $h$.
At query time, the nearest site is found in $O(1)$ time by checking which of the indices of the $r,g,b$ sites minimizes the difference (modulo the size of $h$) with the index of $e'$.}
This proves the running time stated in \cref{thm:SimpleVoronoiSite}. To conclude the proof of the theorem it only remains to establish the prefix property (in \cref{lem:TwoColorMonotonicity}) and the correctness of the decision rule of \cref{alg:triDec} (in \cref{lem:SimpleTrichromaticDecisionCorrectness}).

\begin{lemma}\label{lem:TwoColorMonotonicity} 
Let $e \in T_g$ be a green edge, and let $P^*_j$ be the two paths comprising the fundamental cycle of $e$ w.r.t. $T^*_g$. 
For each $j\in {1,2}$, the set of edges 
$\{ e' \in P^*_j \mid e' \text{ is green} \}
$
forms a prefix of $P^*_j$.
\end{lemma}

\begin{proof}
Note that for any edge $a^*$ of $P^*_j$, the fundamental cycle $C_a$ of $a$ w.r.t. $T_g$ encloses the prefix of $P^*_j$ ending at $a^*$. If $a$ is a green edge then both (primal) endpoints of $a$ are green, hence $C_a$ is comprised of two shortest paths, each from $g$ to a vertex in $\Vor(g)$. Thus, following from \cref{lem:VD_shortest_path_color}, $C_a$ encloses only green vertices. Cf. \cref{fig:TwoColorMonotonicity_main-figure}
\end{proof}

\begin{figure}[H]
    \centering
    \begin{minipage}{0.28\textwidth}
        \centering
        \includegraphics[width=\linewidth]{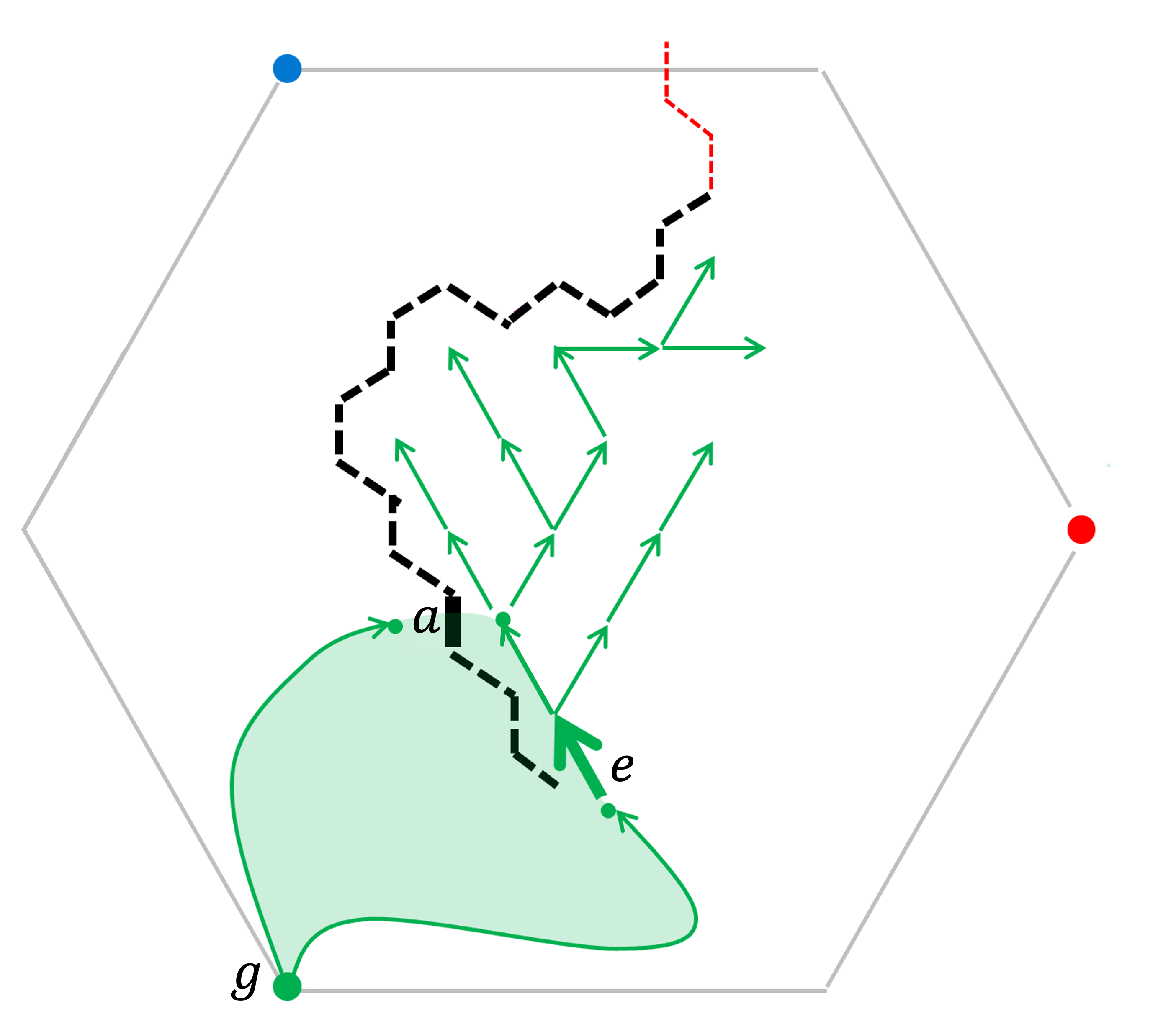}
    \end{minipage}%
    \hfill
    \begin{minipage}{0.7\textwidth}
        \captionof{figure}{%
            A schematic diagram of the Proof of Lemma~\ref{lem:TwoColorMonotonicity}.
            A green edge $e \in T_g$ is indicated by the bold green arrow. One of the paths comprising the fundamental cycle of $e^*$ w.r.t. $T^*_g$, say, $P_1^*$, is indicated by a dashed black line (the remainder of its rootward path in $T^*_g$ is indicated by a dashed red line).
            A green edge $a^*$ on $P^*_1$ is marked in a thick black line.
            The fundamental cycle $C_a$ of $a$ w.r.t. $T_g$ is shaded green. It contains only green vertices and encloses the prefix of $P^*_1$ ending at $a^*$.
        }
        \label{fig:TwoColorMonotonicity_main-figure}
    \end{minipage}
\end{figure}

\alglanguage{pseudocode}
\begin{algorithm}[h]
\caption{\textsc{SimpleTreeElimination}$(T)$
\label{alg:s_centsearch}}
\textbf{Input:} A subtree $T$ of $T_g$, s.t. if there exists a trichromatic face $\tilde f$ in $\VD(\{r,g,b\})$, then $T$ contains the green vertex of $\tilde f$.
\\
\textbf{Output:} An edge $e$ incident to green vertex $v$ if $\tilde f$ exists. 
\begin{algorithmic}[1]
    \If{$T$ consists of a single edge $e$}
        \Return $e$
    \EndIf
    \State Let $e=uv$ be a centroid edge of $T$ \Comment{$u$ is rootward}
    \If{$u$ is not a green vertex} \label{ln:SCS_analogy}
        \State \Return \textsc{SimpleTreeElimination}($T \setminus T_g(v)$)
    \EndIf
    \State $e_1,e_2$ = {\sc SimpleFindCritical($e$)}
    \If{\sc TrichromaticDecisionRule($e_1, e_2$)}
        \State \Return \textsc{SimpleTreeElimination}($T_g(v) \cup e$)
    \Else
        \State \Return \textsc{SimpleTreeElimination}($T \setminus T_g(v)$)
    \EndIf
    \Statex
\end{algorithmic}
\vspace{-0.4cm}%
\end{algorithm}

\alglanguage{pseudocode}
\begin{algorithm}[h]
\caption{\textsc{SimpleFindCritical}$(e)$  
\label{alg:s_CriticalEdges}}  
\textbf{Input:} An edge $e \in T_g$  
\\  
\textbf{Output:} The critical edges of $e$ w.r.t. $T_g$.  
\begin{algorithmic}[1]
    \For{$j \in \{1,2\}$}
        \While{multiple edges remain in $P^*_j$} \Comment{binary search}
            \State Select middle edge $e'_j$ along $P^*_j$
            \If{$e'_j$ is a green edge}
                \State Eliminate the prefix of $P^*_j$ up to and including $e'_j$
            \Else
                \State Eliminate the suffix of $P^*_j$ starting immediately after $e'_j$
            \EndIf
        \EndWhile
        \State Let $e'_j$ be the remaining edge in $P^*_j$ \Comment{null if no edges left}
    \EndFor
    \State \Return $\{e'_1, e'_2\}$
\end{algorithmic}
\end{algorithm}

Now, we turn to address the correctness of the decision rule.

\alglanguage{pseudocode}
\begin{algorithm}[h]
\caption{\textsc{TrichromaticDecisionRule}$(e_1,e_2)$}
\label{alg:triDec}
\textbf{Input:} The critical edge $e_1$ with respect to $P^*_1$ and the critical edge $e_2$ with respect to $P^*_2$.
\\
\textbf{Output:} If the trichromatic face $\tilde f$ of $\VD(\{r,g,b\})$ exists then the output is true if and only if $T^e$ contains the green vertex of $\tilde f$.
Otherwise, the output is either true or false.
\begin{algorithmic}[1]
    \State Let $S$ be the set of primal endpoints of $e_1, e_2$
    \For{$j \in \{1,2\}$}
        \State let $q^*$ be the common dual vertex of $P^*_1$ and $P^*_2$.
        \State let $h^*$ denote the dual vertex corresponding to face $h$.
        \If{($e_j$ is null \textbf{and} $q^*$ is $h^*$)}
            \State Let $e'_j$ be the primal edge of the last dual edge on $P^*_j$ leading to $h^*$. \Comment{$e'_j$ is on $h$}
            \State Let $s^j_1, s^j_2$ be the nearest sites (of $\{ r,g,b \}$) to $e'_j$ along face $h$
            \State Add $s^j_1$ and $s^j_2$ to set $S$
        \ElsIf{($e_j$ is null \textbf{and} $q^*$ is not $h^*$)}
            \State Let $S_j$ be set of vertices incident to $q^*$
            \State Add the vertices of $S_j$ to set $S$
        \EndIf
    \EndFor
    \State \Return true iff in $S$ there is a red vertex, green vertex and a blue vertex.
\end{algorithmic}
\end{algorithm}

\begin{lemma}\label{lem:SimpleTrichromaticDecisionCorrectness}
Let $e_1,e_2$ be the critical edges of $e=uv$ w.r.t. $T_g$, such that $u$ is a green vertex and rootward of $v$. Assume the trichromatic face $\tilde f$ exists. \cref{alg:triDec} returns true if and only if there is an edge $\tilde e \in T^e$ such that one of the endpoints of $\tilde e$ is the green vertex of $\tilde f$.
\end{lemma}

\begin{figure}[h]
    \centering
    \begin{subfigure}{0.3\textwidth}
        \includegraphics[width=\textwidth]{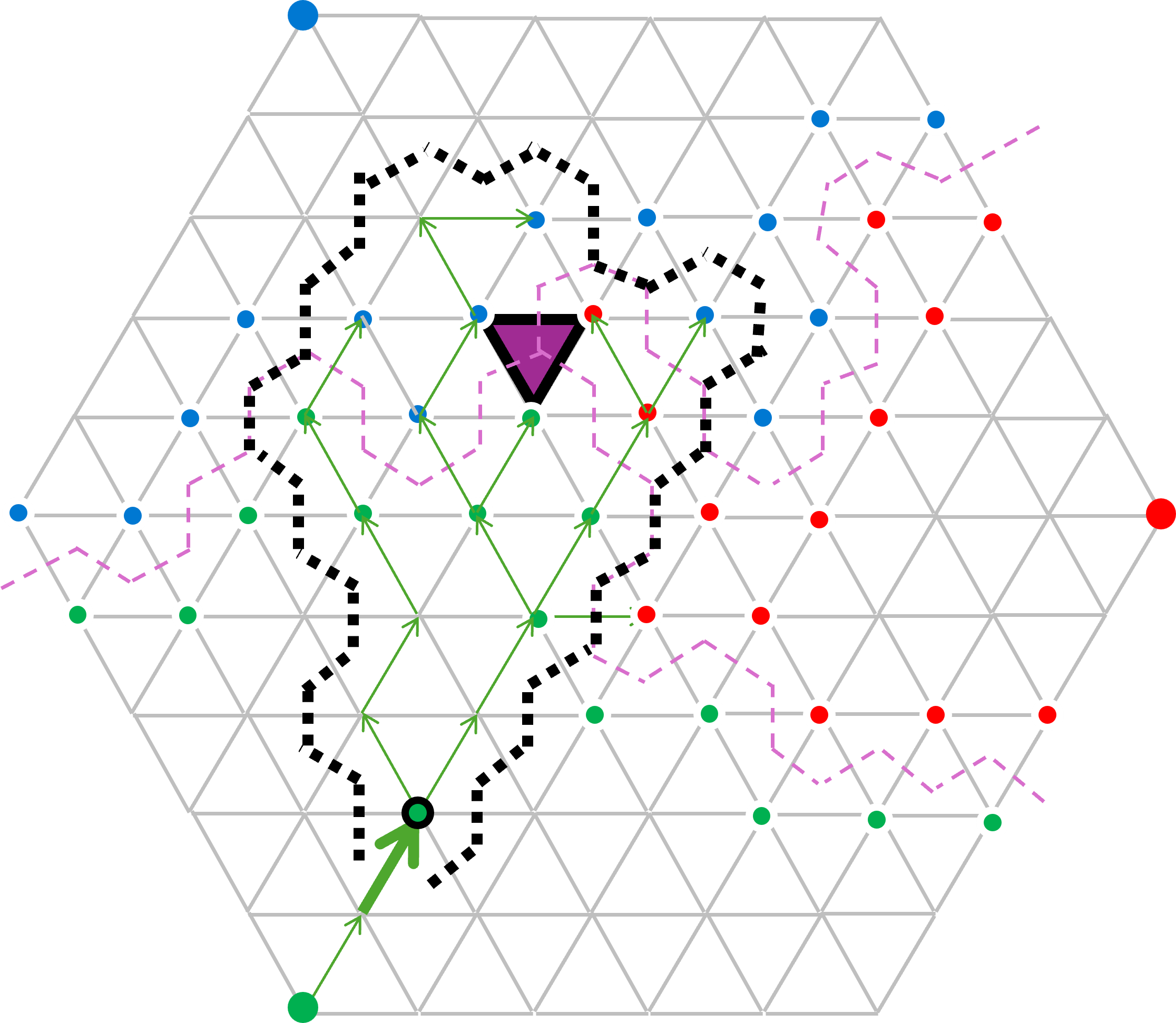}
        \caption{positive example}
        \label{fig:VoronoiVertex_positive_a}
    \end{subfigure}
    \hfill
    \begin{subfigure}{0.3\textwidth}
        \includegraphics[width=\textwidth]{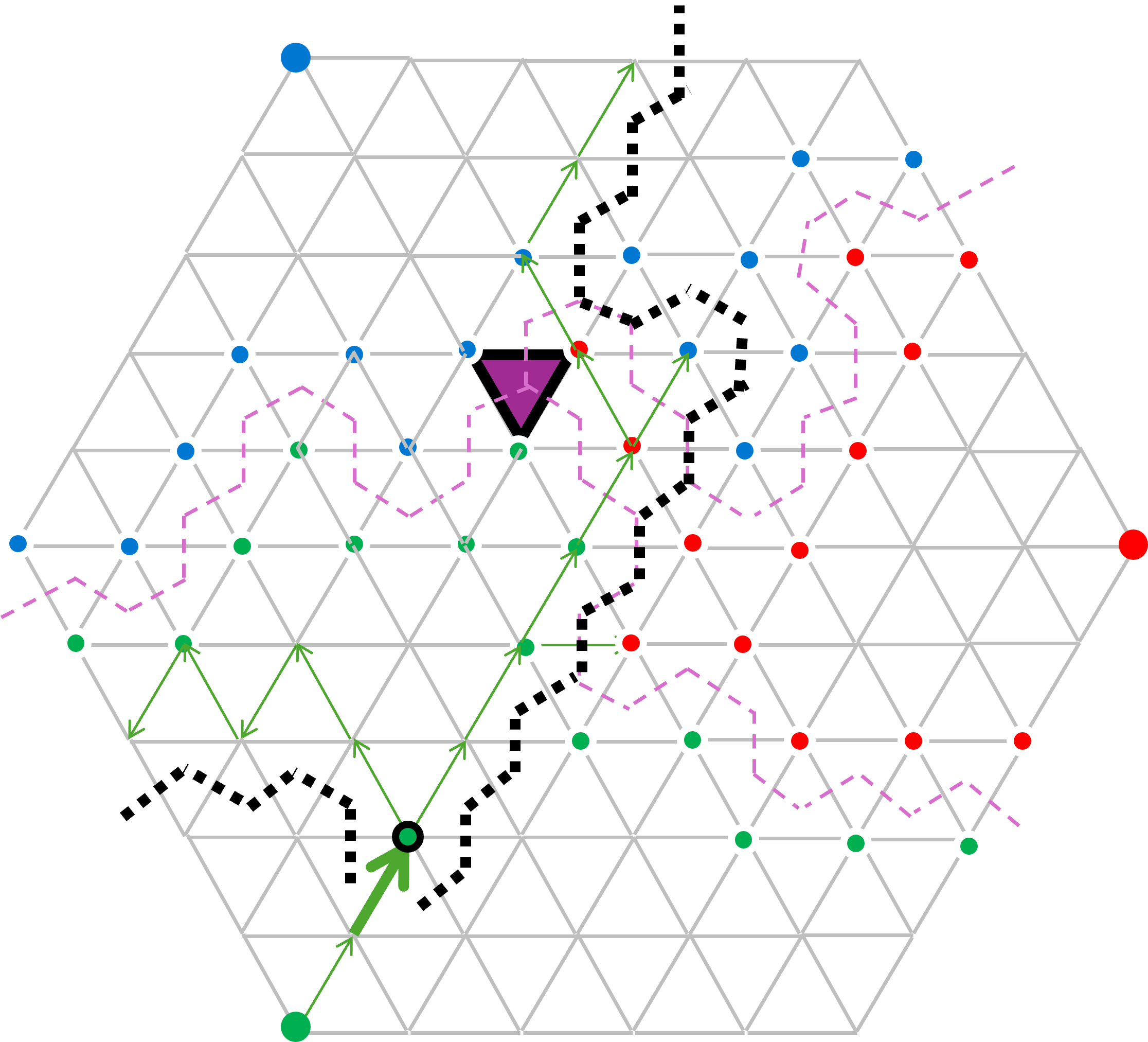}
        \caption{positive example}
        \label{fig:VoronoiVertex_positive_b}
    \end{subfigure}
    \hfill
    \begin{subfigure}{0.3\textwidth}
        \includegraphics[width=\textwidth]{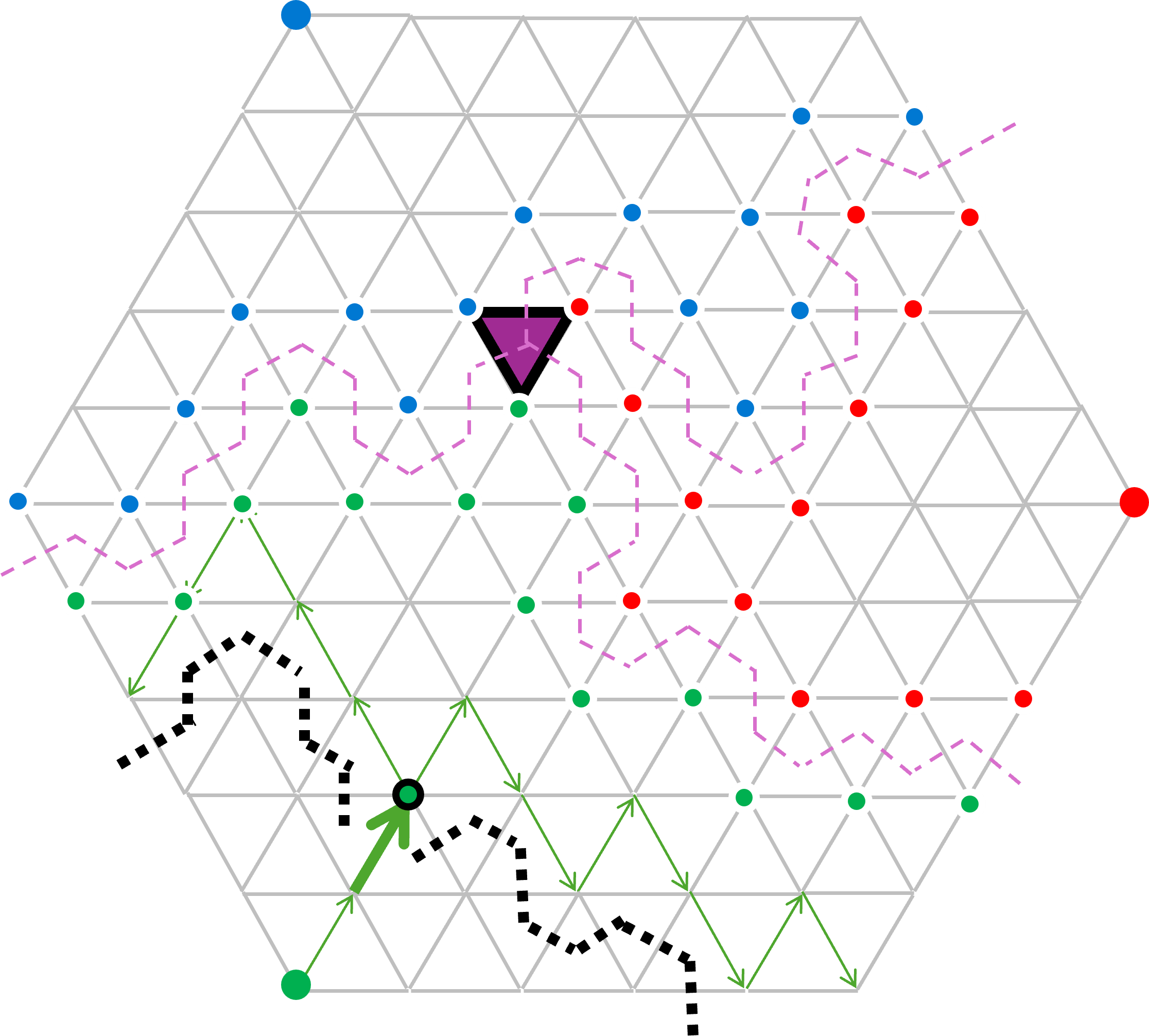}
        \caption{positive example}
        \label{fig:VoronoiVertex_positive_c}
    \end{subfigure}

    \vspace{1em} 

    \begin{subfigure}{0.3\textwidth}
        \includegraphics[width=\textwidth]{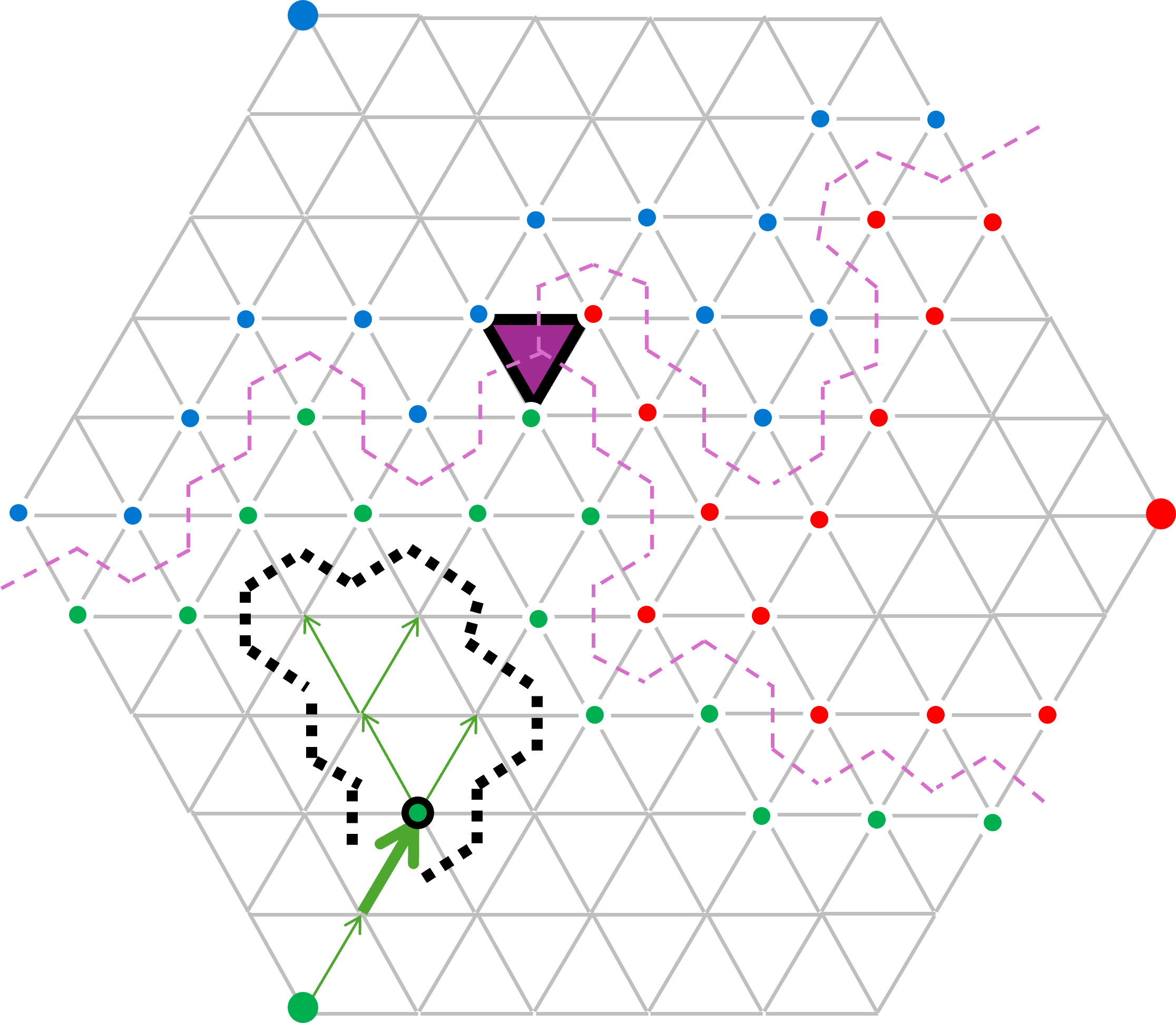}
        \caption{negative example}
        \label{fig:VoronoiVertex_negative_a}
    \end{subfigure}
    \hfill
    \begin{subfigure}{0.3\textwidth}
        \includegraphics[width=\textwidth]{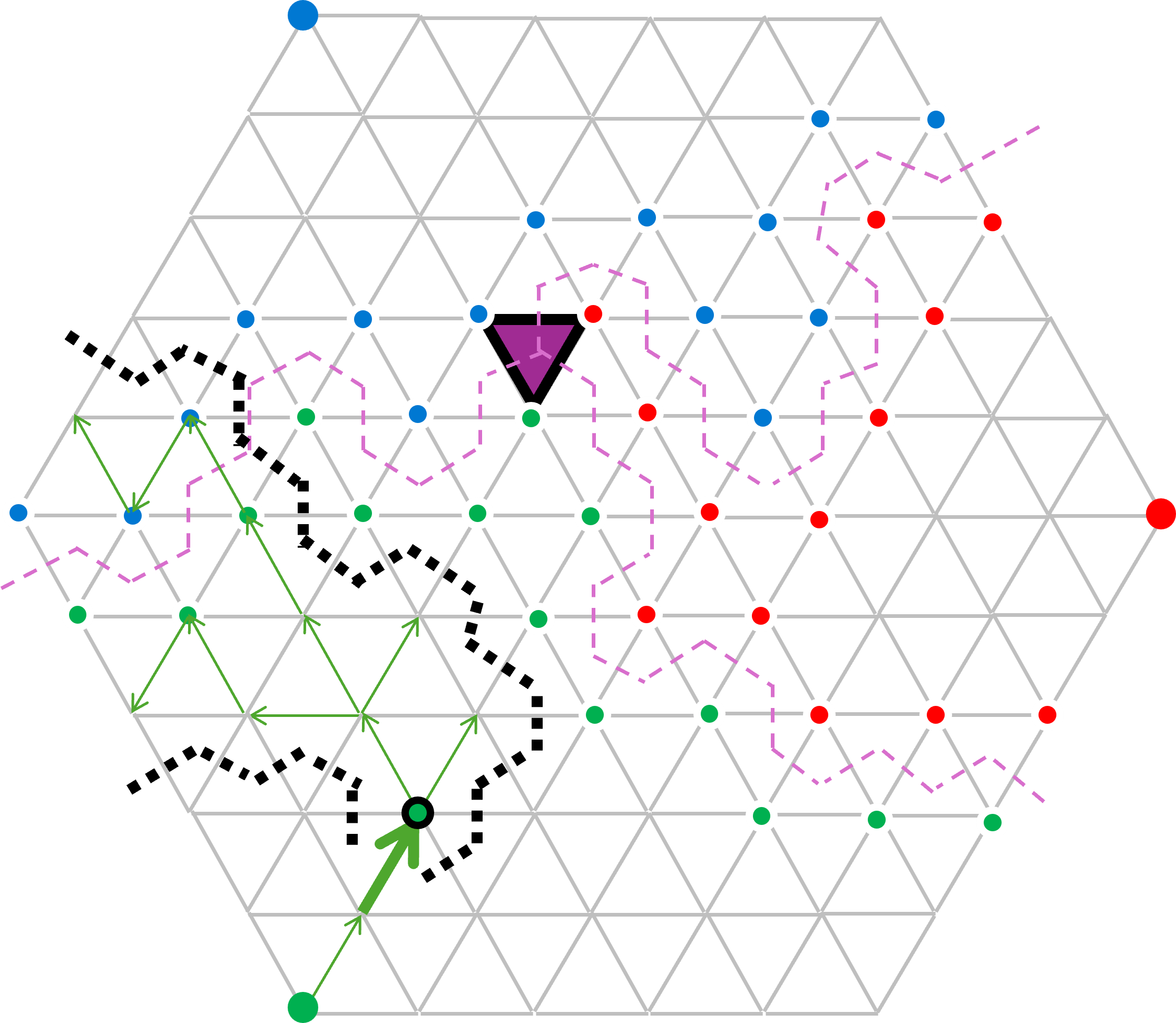}
        \caption{negative example}
        \label{fig:VoronoiVertex_negative_b}
    \end{subfigure}
    \hfill
    \begin{subfigure}{0.3\textwidth}
        \includegraphics[width=\textwidth]{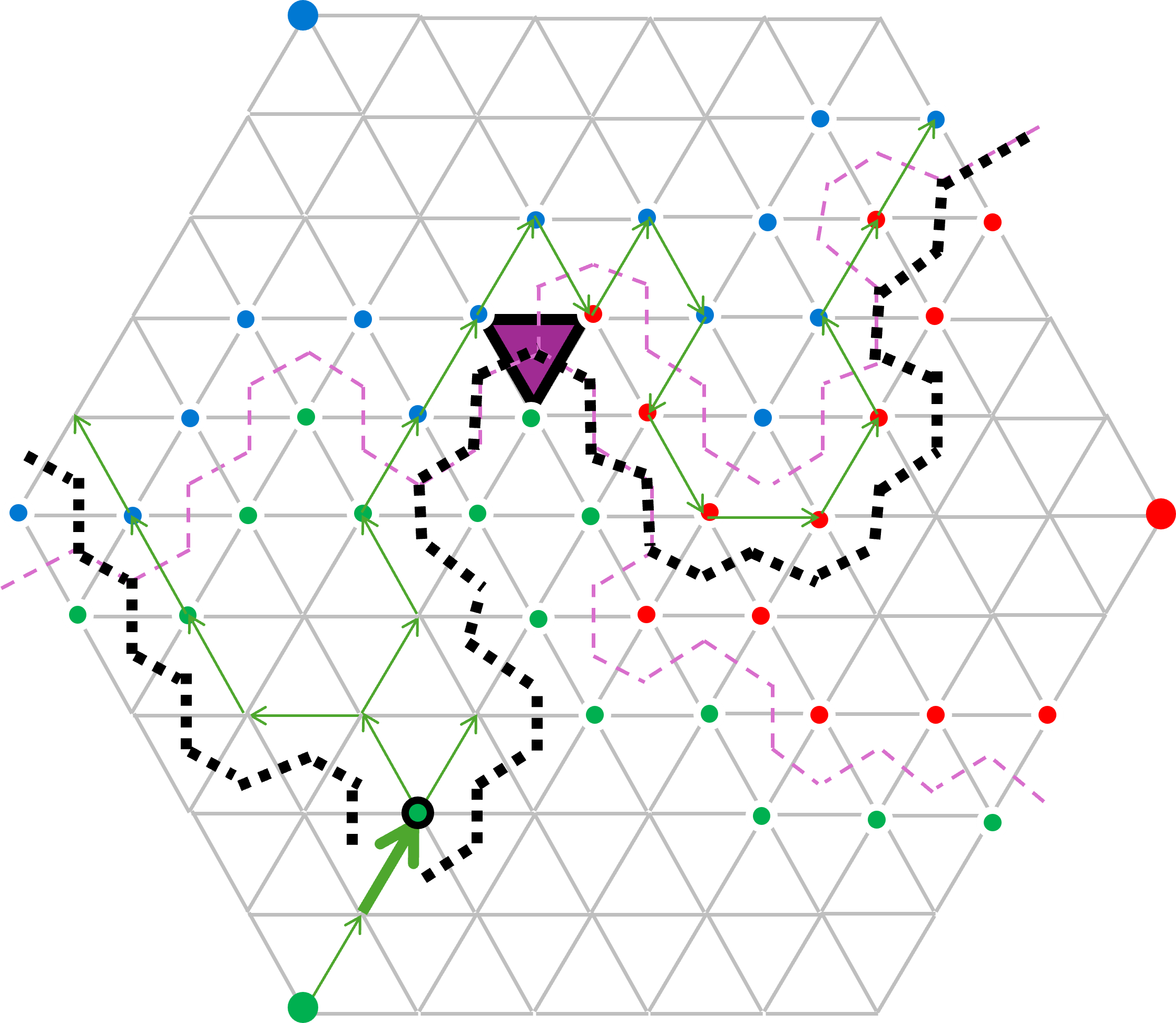}
        \caption{negative example}
        \label{fig:VoronoiVertex_negative_c}
    \end{subfigure}

    \caption{\textbf{Possible cases in the Proof of Lemma~\ref{lem:SimpleTrichromaticDecisionCorrectness}.}
            The bold green arrow represents edge $e$ and the dashed black line represents the fundamental cycle of $e^*$ w.r.t. $T^*_g$. The bisectors are indicated by a thin dashed purple line and the trichromatic face $\tilde f$ is the bold purple triangle.
            In cases (a)-(c) the green vertex of $\tilde f$ belongs to $T^e$, whereas in cases (d)-(f) the green vertex of $\tilde f$ does not belong to $T^e$.}
    \label{fig:VoronoiVertex_main-figure}
\end{figure}

\begin{proof}
Let $S$ be the set represented by the variable $S$ of \cref{alg:triDec} after lines 1--9 have been executed.
That is, $S$ is the set defined as follows. For each $j \in \{1,2\}$,
\begin{itemize}
    \item If the critical edge $e_j$ exists then $S$ contains its primal endpoints. 
    \item Else, if $P^*_j$ ends with $h^*$, let $e'_j$ be the primal edge corresponding to the last dual edge on $P^*_j$. 
    Let $s^j_1, s^j_2$ be the nearest sites (of $\{ r,g,b \}$) to $e'_1$ along face $h$.
    The set $S$ contains the sites $s^j_1, s^j_2$.
    \item Else (i.e., $e_j$ does not exist and  $P^*_j$ does not end with $h^*$), the set $S$ contains the primal vertices incident to the lowest common dual vertex of $P^*_1$ and $P^*_2$.
\end{itemize}

For the first direction, we assume that \cref{alg:triDec} returns true. 
This happens only if all three colors are present among vertices of $S$.
We claim that exactly one of the primal endpoints of $e_j$ (if $e_j$ exists) must be green. 
If the green prefix is not empty then the edge preceding $e^*_j$ on $P^*_j$ is green.
If the green prefix is empty, then $e^*_j$ is the first edge of $P^*_j$, which is incident to $e^*$.
In both cases, since the graph is triangulated, $e_j$ has one primal endpoint (and exactly one since by definition $e_j$ is not green).

We need to consider the various cases. In each case we define appropriate curves that will be used in a unified way to prove that the green vertex of $\tilde f$ is in $T^e$.

\begin{enumerate}
    \item Both critical edges $e_1$ and $e_2$ exist. 
    In this case $S$ consists of just the primal endpoints of $e_1$ and $e_2$. 
    Since each $e_j$ has exactly one green endpoint we may assume w.l.o.g that one endpoint of $e_1$ is red and one endpoint of $e_2$ is blue.
    
    We define a $GR$-curve that starts with $T_g(g,u)$.
    Then continues along $P^*_1$ until reaching $e_1$, then continues along the red-green bisector until reaching the face $h$.
    Similarly, we define a $GB$-curve using $P^*_2$.
    Refer to \cref{fig:Voronoi_Simple_Proof_Regions-figure}.

    \item Exactly one of $e_1$ and $e_2$ exists and $q^*$ is $h^*$. W.l.o.g $e_1$ exists and $P^*_2$ is entirely green. 
    Then the site $g$ is one of the nearest sites to $e'_2$. 
    In this case $S$ consists of the two primal vertices of $e_1$ and the two primal vertices of $e'_2$.
    The edge $e_1$ has exactly one green endpoint, so we assume w.l.o.g. that the other endpoint of $e_1$ is red. 
    This implies that the nearest sites to $e'_2$ are green and blue.

    We define the $GR$-curve as in case (1).
    The $GB$-curve is defined to start with $T_g(g,u)$ and then continue along $P^*_2$ until reaching $h^*$.
    
    \item Exactly one of $e_1$ and $e_2$ exists and $q^*$ is not $h^*$. 
    W.l.o.g $e_1$ exists and $P^*_2$ is entirely green. 
    In this case $S$ consists of the primal endpoints of $e_1$ and the primal endpoints of $q^*$. 
    The edge $e_1$ has exactly one green endpoint, so we assume w.l.o.g. that the other endpoint of $e_1$ is red. 
    Therefore, $q^*$ must have a blue primal vertex (since $S$ has all three colors).
    Note that $q^*$ cannot be the trichromatic face because in this case, $q^* \neq h^*$ is a triangle, so if $q^*$ is trichromatic both $e_1$ and $e_2$ exist. 
    Hence, $q^*$ has exactly two green vertices and one blue vertex.  
    In particular, two dual edges incident to $q^*$ belong to the green-blue bisector.

    We define the $GR$-curve as in case (1).
    The $GB$-curve is defined to start with $T_g(g,u)$ and then continue along $P^*_2$ until reaching $q^*$. Then, the $GB$-curve continues along the blue-green bisector until reaching the $h^*$.

    \item Both $e_1$ and $e_2$ do not exist and $q^*=h^*$. In this case $S$ consists of the nearest sites to $e'_1$ and $e'_2$. 
    The site $g$ is one of the nearest sites to both, and we assume w.l.o.g. that the other nearest site to $e'_1$ is red and the other nearest site to $e'_2$ is blue.  

    The $GR$-curve is defined to start with $T_g(g,u)$ and then continue along $P^*_1$ until reaching $h^*$.   
    The $GB$-curve is defined similarly using $P^*_2$.   

    \item Both $e_1$ and $e_2$ do not exist and $q^* \neq h^*$. 
    We claim this case leads to a contradiction. 
    Since $e$ itself is green, $P^*_1$ and $P^*_2$ are entirely green. 
    Hence $q^*$ has three green vertices.
    Since $q^* \neq h^*$, $S$ would consist of just the primal vertices of $q^*$, which, by triangulation, are all green, contradicting the fact that all 3 colors exist in $S$.

\end{enumerate}

Having defined the curves in each case we complete the proof of this direction in a unified argument that captures all the cases.
Observe that the intersection of the $GB$-curve and the $GR$-curve is exactly the $T_g(g,u)$ prefix of both curves. 
Since both curves start at $g$ and end at $h^*$, the two curves partition $H$ into three parts. 
One part, which we call the $GB$-part is bounded by the the $GB$-curve and the subpath of $h$ going back to $g$ that does not contain the red and blue sites.
The other part, which we call the $GR$-part is bounded by the the $GR$-curve and the subpath of $h$ going back to $g$ that does not contain the red and blue sites.
The third part, which we call the $RB$-part, is the remainder of $H$, and is bounded by the non-common suffixes of the $GR$-curve and the $GB$-curve, connected by the subpath of $h$ that does not contain the green site.

Observe that, by \cref{lem:VD_shortest_path_color}, all vertices in the $GB$-part are green, and no neighbor of a vertex in the $GB$-part is red. 
This is because the red and blue sites do not belong to the $GB$-part, and because all the edges of the $T_g(g,u)$ part of the $GB$-curve are green, the green prefix of $P^*_2$ consists only of green edges, and all the remaining edges along the $GB$-curve either belong to the green-blue bisector (and hence are not incident to a red vertex), or belong to $h$.

The fact that no neighbor of a vertex in the $GB$-part is red implies that the green vertex of the trichromatic face $\tilde f$ is not in the $GB$-part.
A symmetric argument shows that no neighbor of a vertex in the $GR$-part is blue, which implies that the green vertex of the trichromatic face $\tilde f$ is not in the $GR$-part.

It follows that the green vertex of $\tilde f$ can only be in the $RB$-part. 
Since the green site does not belong to the $RB$-part and since the only edge of $T_g$ that crosses the boundary of the $RB$-part is $e$ itself, it follows that the green vertex of $\tilde f$ belongs to $T^e$.
This concludes the first direction of the proof.

\begin{figure}[h]
    \centering
    \begin{subfigure}{0.3\textwidth}
        \includegraphics[width=\textwidth]{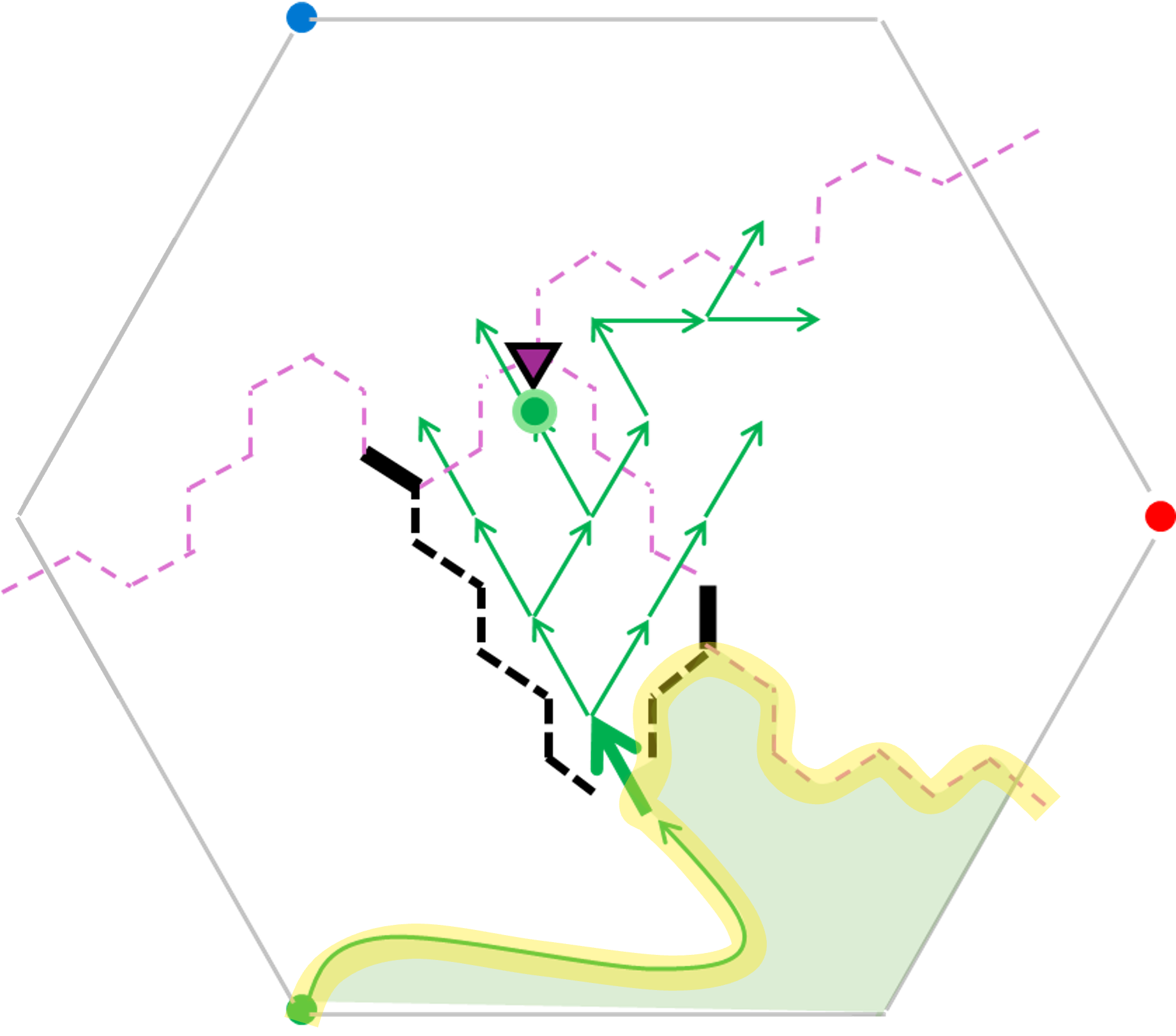}
        \caption{The $GR$-part}
        \label{fig:Voronoi_Simple_Proof_Regions_b}
    \end{subfigure}
    \hfill
    \begin{subfigure}{0.3\textwidth}
        \includegraphics[width=\textwidth]{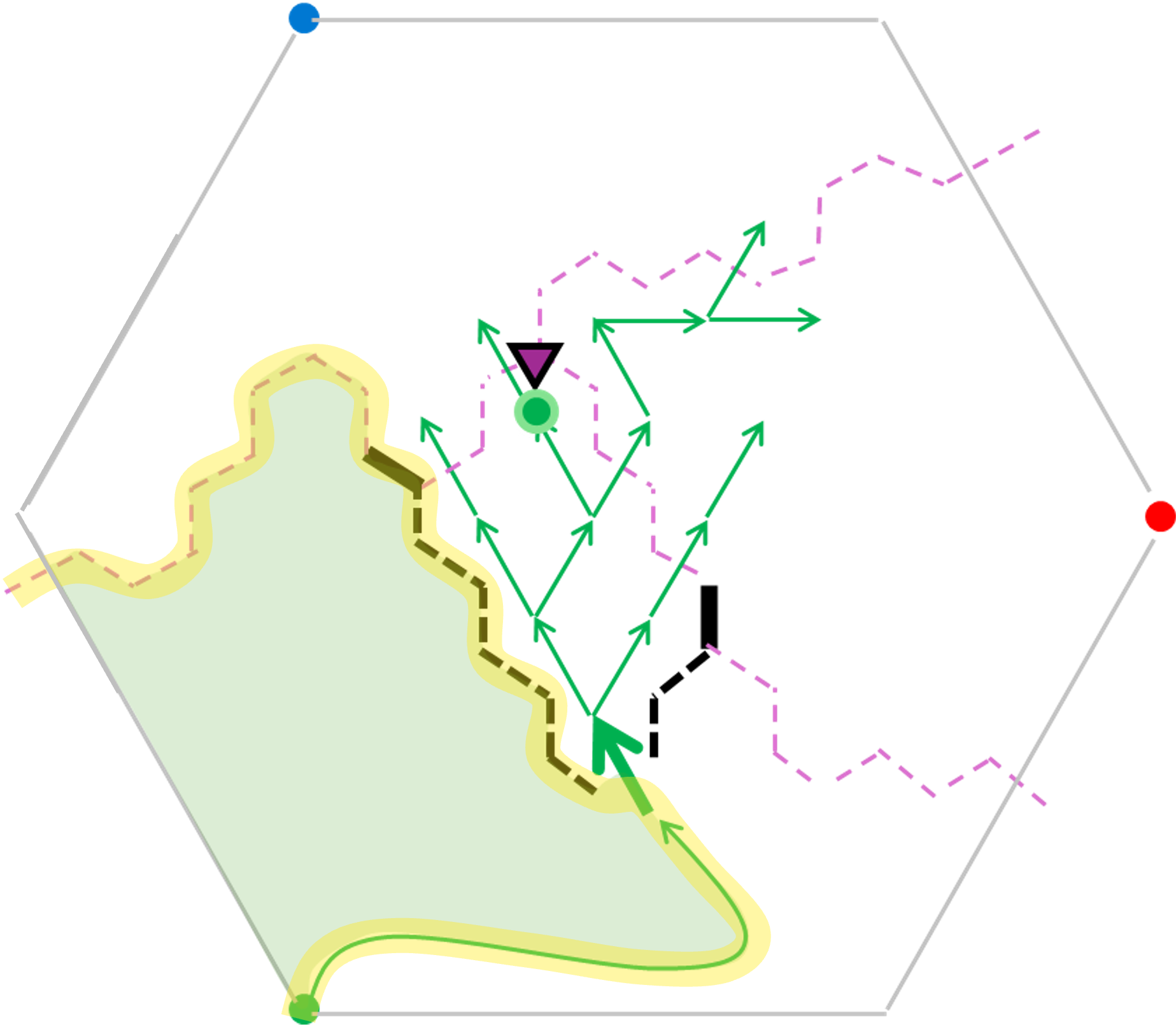}
        \caption{The $GB$-part}
        \label{fig:Voronoi_Simple_Proof_Regions_a}
    \end{subfigure}
    \hfill
    \begin{subfigure}{0.3\textwidth}
        \includegraphics[width=\textwidth]{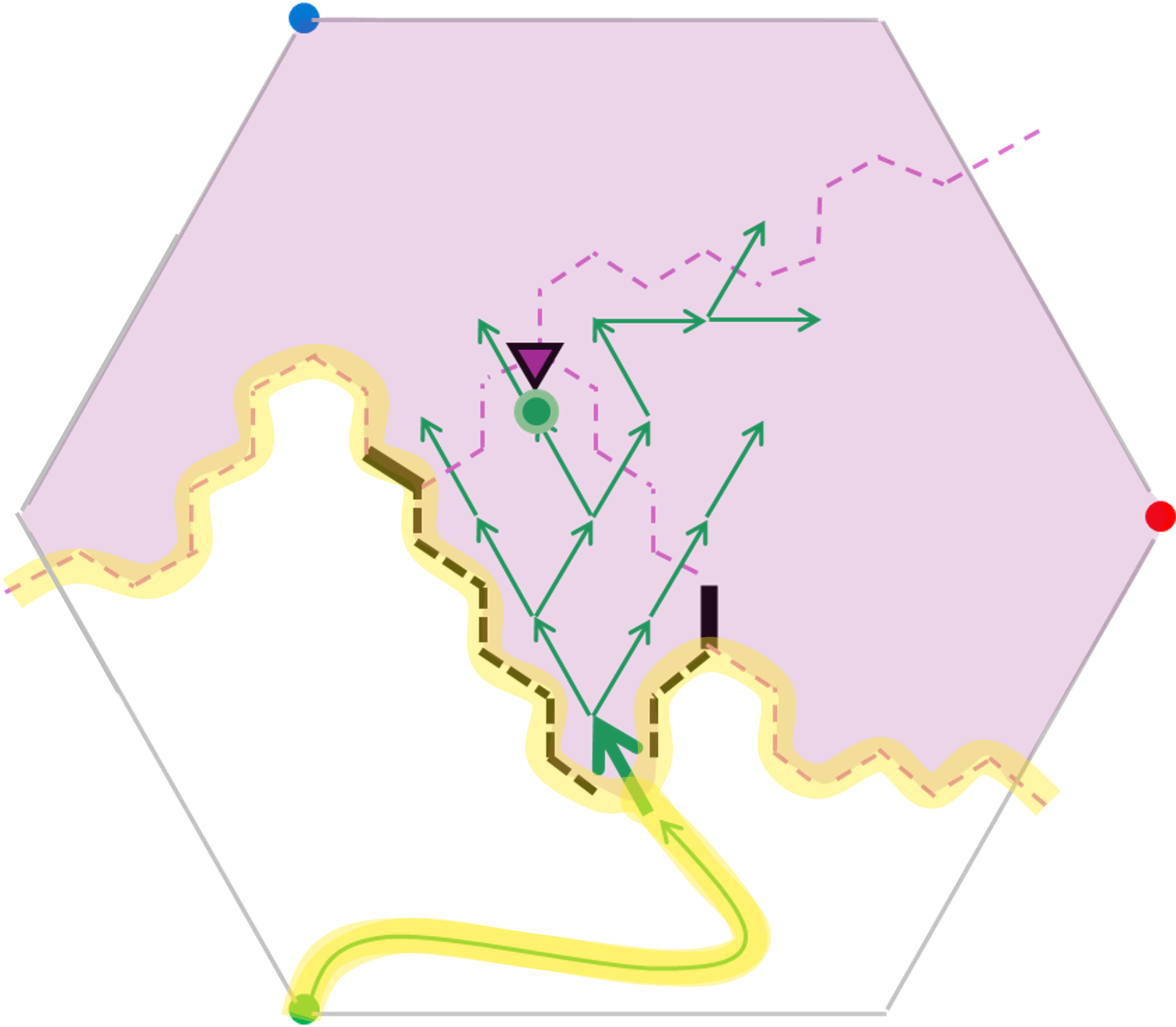}
        \caption{The $RB$-part}
        \label{fig:Voronoi_Simple_Proof_Regions_c}
    \end{subfigure}

    \caption{\textbf{The partition described in the Proof of Lemma~\ref{lem:SimpleTrichromaticDecisionCorrectness}.}
            The shaded areas indicate the parts. 
            The $GR$-curve and $GB$-curve are highlighted in yellow.
            The bold green arrow represents edge $e$ and the dashed black lines represent the the green prefixes of paths $P^*_1, P^*_2$ of the fundamental cycle of $e^*$ w.r.t. $T^*_g$.
            The bisectors are indicated by thin dashed purple lines, and the trichromatic face $\tilde f$ is the bold purple triangle.
            The highlighted green node is the green vertex of $\tilde f$.
            The figures correspond to case (1) where both critical edges exist.
}
    \label{fig:Voronoi_Simple_Proof_Regions-figure}
\end{figure}

For the second proof direction, we assume that \cref{alg:triDec} returns false. 
This happens only if not all three colors are present among vertices of $S$.

We consider the various cases. 
For each case we define an appropriate cycle $C$ in the plane that separates the green vertices of $T_g(v)$ from the rest of $H$. 
The cycle $C$ is then used in a unified way to prove that the green vertex of $\tilde f$ is not in $T^e$.

\begin{enumerate}
    \item Both critical edges $e_1$ and $e_2$ exist.
    Since not all 3 colors appear in $S$, the endpoints of both $e^*_1$ and $e^*_2$ have the same two colors, one being green, and the other, w.l.o.g., blue. 
 
    We define the cycle $C$ to start with $e^*$, then continue along $P^*_1$ until reaching $e^*_1$, then continue along the green-blue bisector until reaching $e^*_2$, finally along the reverse of $P^*_2$ back to $e^*$.
    Cf. \cref{fig:VoronoiVertex_negative_c}.

    \item Only $e_1$ exists and $q^*=h^*$.
    Since not all 3 colors appear in $S$, the endpoints of $e^*_1$ and the two nearest sites to $e'_2$ have the same two colors, one being green, and the other, w.l.o.g., blue. 

    We define the cycle $C$ to start with $e^*$, then continue along $P^*_1$ until $e^*_1$. 
    Then, continue along the green-blue bisector until reaching $h^*$. Then continue along $h$ towards the green site until reaching $e'_2$. Finally, continue along the reverse of $P^*_2$ until getting back to $e^*$.
    Cf. \cref{fig:VoronoiVertex_negative_b}.

    \item Only $e_1$ exists and $q^* \neq h^*$.
    Assume w.l.o.g. that $e^*_1$ has one green vertex and one blue vertex. Moreover, the last edge of $P^*_1$ is not green.
    Since not all 3 colors appear in $S$, none of the primal vertices of $q^*$ is colored red. Since the last edge of $P^*_1$ is not green, $q^*$ has a green vertex and a blue vertex. 

    We define the cycle $C$ to start with $e^*$, then continue along $P^*_1$ until $e^*_1$. 
    Then, continue along the green-blue bisector until reaching $q^*$. Finally, continue along the reverse of $P^*_2$ until getting back to $e^*$.

    \item Both $e_1$ and $e_2$ do not exist and $q^* = h^*$.
    Since not all 3 colors appear in $S$, the two nearest sites to $e'_1$ and to $e'_2$ are the same, one being the green site (since $e'_1$ and $e'_2$ are green edges), and the other, w.l.o.g., the blue site. 

    We define the cycle $C$ to start with $e^*$, then continue along $P^*_1$ until reaching $h$ at $e'_1$.
    Then, continue along $h$ until $e'_2$ using the subpath of $h$ that does not pass through any site. 
    Finally, continue along the reverse of $P^*_2$ back to $e^*$.

    \item Both $e_1$ and $e_2$ do not exist and $q^* \neq h^*$. 
    We define the cycle $C$ to be the fundamental cycle of $e^*$ w.r.t. $T^*_g$. Note that in this case all the edges of $C$ are green. Cf. \cref{fig:VoronoiVertex_negative_a}.

\end{enumerate}

In all cases, by \cref{lem:VD_shortest_path_color}, the cycle $C$ encloses exactly the green vertices of $T_g(v)$, and no neighbor of a vertex enclosed by $C$ is red. 
This is because all edges of $C$ are either green, or edges of the green-blue bisector.  
Since all vertices of $T_g(v)$ have no red neighbor, no vertex of $T_g(v)$ can belong to the trichromatic face $\tilde f$.
\end{proof}

The following corollary is an immediate consequence of using \cref{lem:SimpleTrichromaticDecisionCorrectness} to direct the tree elimination algorithm in {\sc SimpleTreeElimination}. This corollary concludes the correctness proof of \cref{thm:SimpleVoronoiSite}.

\begin{corollary}\label{cor:SimpleCriteriaCorrectness}
{\sc SimpleTreeElimination} correctly finds an edge incident to the green vertex of the trichromatic face $\hat f$, if it exists.
\end{corollary}

\section{Construction of the Voronoi Diagrams of the oracle of \cite{ourJACM}}
\label{sec:complete}

We now turn to using the approach of \cref{sec:simple} to efficiently construct components (A)-(C) of the recursive distance oracle of \cite{ourJACM}.
These components were described in the preliminaries.
Component (A) is the MSSP data structure for $R_i^{out} \cap R_{i+1}$ for all regions at levels $0\leq i < m$. 
Constructing it is quite simple. 
The structure for a particular $R_i^{out} \cap R_{i+1}$ can be constructed in $\Otild(|R_{i+1}|)$ time by a bottom up computation, starting from regions at level $i=m-1$ for decreasing $i$. 
When handling $R_i^{out}$ we already have the distances in $R_{i+1}^{out}$ among the vertices of $\partial R_{i+1}$ from the MSSP structure for the previous level.
We use those distances to construct a planar emulator~\cite{emulator, GoranciHP20} in $\Otild(|\partial R_{i+1}|^2) = \Otild(|R_{i+1}|)$ time, and then run the MSSP algorithm \cite{CabelloCE13} on the union of the emulator and of $R_i^{out} \cap R_{i+1}$.
The running time is $\Otild(|R_{i+1}|)$. 
The emulator concisely represents the parts of the SSSP trees in $R_{i+1}^{out}$, so the distances represented by the MSSP data structure are in $R_i^{out}$, but the total size is just $\Otild(|\partial R_{i+1}|^2)$ as specified in the definition of part (A).
Thus the total time for constructing all of part (A) is proportional, up to polylogarithmic factors to its space requirement.

The main challenge we need to overcome is the computation of parts (B)-(C).
Namely, the Voronoi diagrams for $R_i^{out}$ for all regions $R_i$ at all levels of the recursive decomposition.
Computing the additive distances for all the Voronoi diagrams in $\Otild(mn^{1+1/(2m)})$ (i.e., bounded by the size of the oracle) is done exactly as described in~\cite[Section 8]{ourJACM}.
As we show in \cref{sec:jacm_complete}, computing parts (D)-(E) does not pose additional significant challenges.
The Voronoi diagrams of parts (B)-(C) are constructed starting with those for the outside of regions at level $m$ of the decomposition, and ending with those for the outside of regions at level 1. 
This way, when working on level $i$ (i.e., when computing $\VD^*(q,R_i^{\mathrm{out}})$ for some vertex $q$), we have already computed all the components (i.e., parts (A)-(E)) for $R_{i'}^{\mathrm{out}}$ for every $i'>i$.
In particular, we already have all the components necessary for executing the point location algorithm of~\cite{CharalampopoulosGMW19} for any vertex $v \in R_{j+1}^{out}$ in $\VD^*(q,R_{j+1}^{out})$ for any $j \geq i$ and $q \in \partial R_j$, and (together with the MSSP for $R_j^{out} \cap R_{j+1}$) to execute the distance query algorithm of~\cite{CharalampopoulosGMW19} from any such $q$ to any $u \in R_j^{out}$.

The computation of the Voronoi diagram $\VD^*(q,R_i^{\mathrm{out}})$ is performed using the divide and conquer mechanism in \cite[Section 5]{icalp2025planar}. 
Note that while the main result of \cite{icalp2025planar} is for undirected planar graphs, this mechanism applies without any changes to directed planar. 
We summarize this mechanism in the following lemma.

\begin{lemma} \label{lem:d-and-c-VD}
There exists an algorithm that constructs an additively weighted Voronoi diagram of a graph $H$ with sites $S$ lying on a single face $h$ of $H$ using $\Otild(|S|)$ many calls to each of the following operations: 
\begin{enumerate}
    \item Reporting the distance in $H$ from any site in $S$ to any vertex of $H$.
    \item Finding the single trichromatic face with respect to any set of 3 sites in $S$.
\end{enumerate}
\end{lemma}

As we argued above, item (1) is provided by the already computed components of the oracle in time $O(Q)$, where $Q$ is the query time of the distance oracle.

In the rest of this section we describe how to implement item (2) in the context of the recursive decomposition used by the oracle. 
Namely, finding the single trichromatic face w.r.t. a set of 3 sites on $\partial R_i$ when $H = R_i^{out}$, for some region $R_i \in \mathcal R_i$.
Let $T$ be the time required for computing such a trichromatic face. 
We show that $T = (\log n)^{O(m)} \cdot Q$ (\cref{thm:tc-time}).
With this, we can analyze the total running time required for constructing all the Voronoi diagrams comprising item (B) of \cite{CharalampopoulosGMW19}, and parts (B) and (C) of \cite{ourJACM}.
According to \cref{lem:d-and-c-VD}, the construction of a Voronoi diagram with $S$ sites requires $\Otild(|S|\cdot (Q+T))$ time.
Hence, the total construction time is $\Otild(|S_{tot}|\cdot (Q+T))$ where $S_{tot}$ is the total number of sites in all of the constructed Voronoi diagrams, which is bounded by the space required by the oracles of \cite{CharalampopoulosGMW19} and \cite{ourJACM}.
Substituting $T = (\log n)^{O(m)} \cdot Q$, the total running time is $\Otild(|S_{tot}|\cdot (\log n)^{O(m)} \cdot Q)$.
For all $m = o(\log n / \log \log n)$, 
this is $n^{1+o(1)}$.
We address the rest of the tradeoff, when $m \in \Omega(\log n / \log \log n) \cap o(\log n)$ in \cref{sec:jacm_complete}.

The $n^{1+o(1)}$ bound on the construction time of parts (A), (B) and (C), together with the description of the construction of parts (D) and (E) in \cref{sec:jacm_complete} will prove \cref{thm:main}.

For the remainder of this section, for $i \leq j \leq m$, let $R_j$ be the region of $\mathcal R_j$ that contains $R_i$.
The mechanism for computing the single trichromatic face works along the lines of {\sc SimpleTreeElimination} described in \cref{sec:simple}. 
In fact, if $i=m-1$, then we use the exact same procedure since for every region $R_{m-1}$ in $\mathcal R_{m-1}$ we have MSSPs for 
$R_{m-1}^{out} \cap R_m = R_{m-1}^{out}$ (as $R_m= G$).
However, when $i < m-1$, we do not have an MSSP data structure for $R_i^{out}$, and hence cannot directly access shortest path trees and their duals, so we cannot implement {\sc SimpleTreeElimination} directly.
To overcome this difficulty, we use a coarse representation of the primal and dual trees. 
The main procedure, {\sc TreeElimination}, performs the search on the coarse representation at various levels of granularity.
When the search eliminates all but a single coarse edge at some coarseness level $i$, we recurse into level $i+1$ to refine that coarse edge and continue the search on this refinement.
This approach, which will be described in the rest of this section will allow us to find the trichromatic face in a 3-site Voronoi diagram for a region at level $i$ of the recursive decomposition in time  
$(\log n)^{O(m)} \cdot Q$.

\subsection{Coarse Trees and Dual Coarse Trees} \label{sec:dualtreedef}

\paragraph{Coarse Trees.}
For a site $x$ on face $\partial R_{i}$ the \emph{Coarse Shortest Path Tree} $\hat{T}_x$ is defined as the tree obtained from $T_x$, the shortest path tree in $R_i^{out}$ rooted at $x$ by repeatedly contracting an edge incident to a leaf that does not belong to $\partial R_{i+1} \cup \partial R_i$ or an edge with both endpoints having degree 2 and at least one endpoint not in $\partial R_{i+1} \cup \partial R_i$.\footnote{To be clear, this contraction process is described just for the sake of defining the coarse tree $\hat T_x$. It is not actually executed in this way at any point.}
Cf.~\cref{fig:CoarseTree_Primal}.
We say that the level of $\hat T_x$ is $i$.
The coarse tree can be computed in $\Otild(|\hat T_x|)=\Otild(|\partial R_{i+1}|)$ time in a bottom up computation along the recursion (starting from level $i=m-1$ and ending with $i=1$) as follows.

\begin{figure}[htb]
    \centering
    \includegraphics[width=0.35\textwidth]{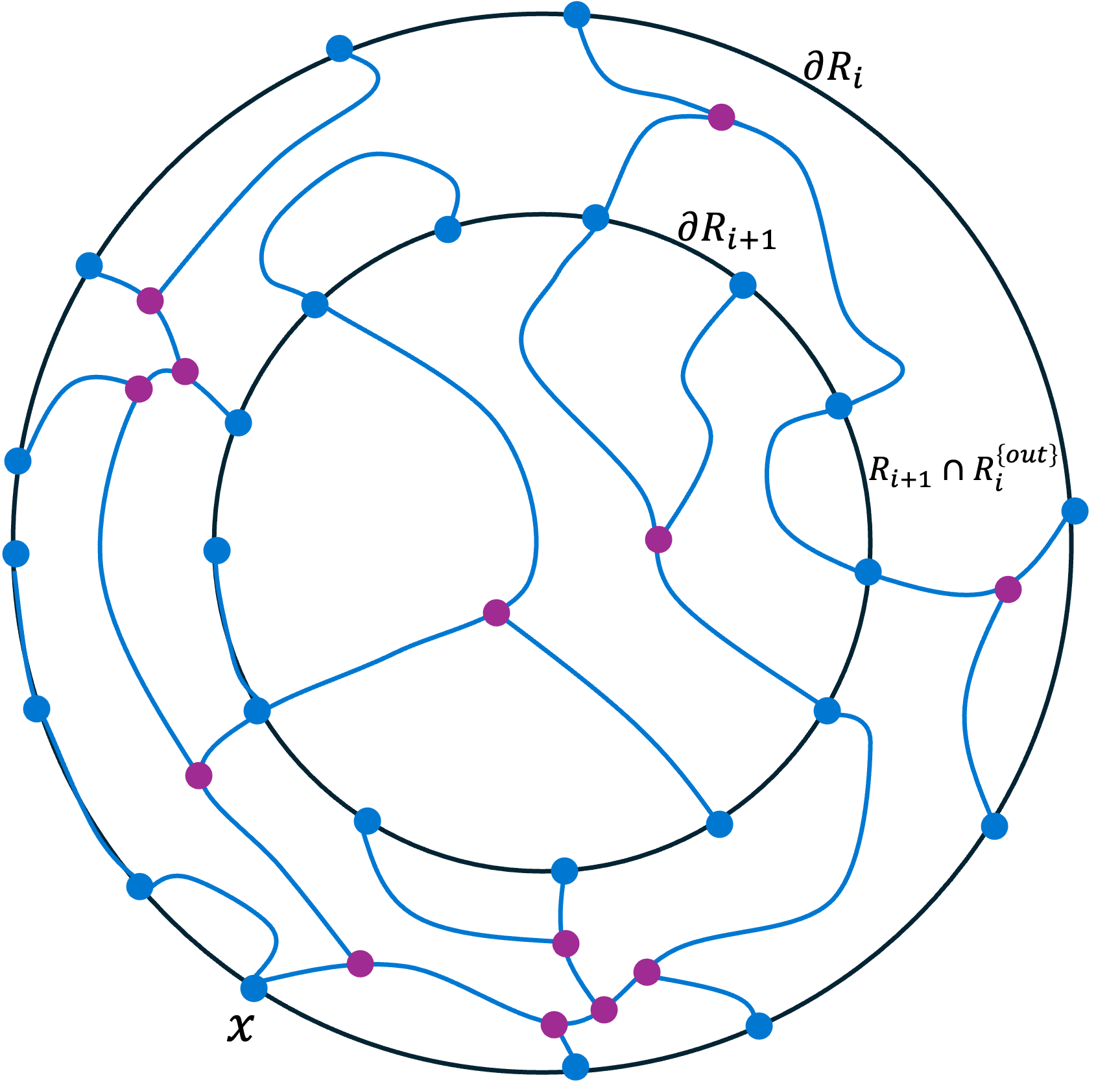}
    \caption{
    A schematic illustration of the primal coarse tree $\hat T_x$. 
    The source $x \in \partial R_i$ of the underlying fine shortest path tree $T_x$ is labeled explicitly.
    The boundary vertices of $R_i$ and $R_{i+1}$ are shown as blue nodes.
    The edges of the coarse tree $\hat T_x$ are indicated by blue lines.
    Vertices in $R_i^{out}$ that belong to $\hat T_x$ but do not lie on $\partial R_i$ or $\partial R_{i+1}$ are shown as purple nodes..
    }
    \label{fig:CoarseTree_Primal}
 \end{figure}
 
Suppose we had already constructed all coarse trees at levels greater than $i$ and we wish to construct a coarse tree $\hat T_x$ at level $i$. 
We begin by constructing a non-fully contracted version of $T_x$ by linking together portions of $T_x$ from elements that are already available at this stage. The parts of $T_x$ in $R_i^{out} \cap R_{i+1}$ are available from the MSSP tree of $x$ for $R_i^{out} \cap R_{i+1}$ (from part (A)). 
The parts of $T_x$ in $R_{i+1}^{out}$ are 
available as subtrees of the coarse trees at level $i+1$ already computed for the vertices  of $\partial R_{i+1}$. We use the artificial edges in the MSSP tree of $x$ in $R_i^{out} \cap R_{i+1}$ to identify which subtrees of the coarse trees at level $i+1$ should be used.
We use $O(|\partial R_{i+1}|)$ link and cut operations to splice the appropriate subtrees and link them all together into a single tree rooted at $x$. 
Then we use $O(|\partial R_{i+1}|)$ LCA queries and marked ancestor queries on this tree to identify the endpoints of each contracted path of $T_x$ that forms an edge of the coarse tree $\hat T_x$.
Marked ancestor queries are used to obtain the vertices of $\partial R_{i+1}$ in the tree (there are $O(m)$ different types of marks, each vertex that is a boundary vertex of a region in the recursive decomposition is marked according to the level of the region).
LCA queries are used to identify the remaining vertices of $\hat T_x$.
This way we identify, and explicitly create and store all edges of $\hat T_x$ in a data structure that supports LCA and marked ancestor queries in $\Otild(1)$ time. 
Constructing $\hat T_x$ in this way thus takes $\Otild(|\hat T_x|)$ time.

We call the edges of the coarse tree coarse edges. We call edges of the original graph fine edges. By definition, every coarse edge $\hat e$ of $\hat T_x$ corresponds to a path $Q$ of fine edges in $T_x$. We associate and store with $\hat e$ the first (i.e., root-most) fine edge $e=uv$ of $Q$. This association can be computed with no asymptotic overhead during the computation of tree $\hat T_x$.
We say that $\hat e$ is a {\em coarsening} of $Q$ and that  every fine edge $e\in Q$ is {\em represented} by the coarse edge $\hat e$.

The following proposition will be useful in viewing the same path at  different levels of coarseness.
\begin{proposition}\label{pro:gamma_subtree_level_association}
     Let $e=uv$ be a fine edge associated with a coarse edge $\hat e_i$ in $\hat T_x$ such that $e$ resides in $R_{i+1}^{out}$.
     Let $\gamma$ be the lowest ancestor of $u$ in $\hat{T}_x$ that belongs to $\partial R_{i+1}$ (possibly $\gamma=u$).
     Then $e$ is also associated with a unique coarse edge $\hat e_{i+1}$ of the level-$(i+1)$ coarse tree $\hat T_{\gamma}$.
     Moreover, during the construction of the coarse trees we can store with $\hat e_i$ a pointer to $\hat e_{i+1}$.
\end{proposition}

\paragraph{Dual Coarse Trees}
We would like to define a coarse representation $\hat T_x^*$ of the dual tree $T_x^*$ so that coarse edges of $\hat T_x^*$ are contractions of paths of fine edges in $T^*_x$. It would be useful that the coarse edges of $\hat T_x^*$ will be objects that are available to us at the current stage of the construction algorithm. 
One such object are the Voronoi edges of $\VD^*(x,R_{i+1}^{out})$ (which were already computed). The other object is the cotree of $\hat T_x$, which we denote by $\cotree{x}$ and define next.

Let $\hat H_x$ be the graph whose vertex set is the vertex set of $\hat T_x$, and whose edges are the edges of $\hat{T}_x$ and the edges of $\partial R_i \cup \partial R_{i+1}$. 
Let $\hat H_x^*$ be the planar dual of $\hat H_x$.
The coarse cotree $\cotree{x}$ is  the spanning tree of $\hat H_x^*$ that consists of all the edges whose primal counterparts are not in $\hat T_x$. 
Note that, by definition of $\hat H_x$, for any dual coarse edge $\hat e^* \in \cotree{x}$, the corresponding primal edge $\hat e$ of $\hat H_x$ is in fact a fine edge of $\partial R_i$ or  of $\partial R_{i+1}$. 
In this sense, any coarse edge $\hat e^* \in \cotree{x}$ is a trivial coarsening of a fine dual edge $e^* \in T^*_x$ s.t. $e\in \partial R_i \cup \partial R_{i+1}$.
Hence, $\cotree{x}$ is a coarsening of $T^*_x$. However $\cotree{x}$ is too coarse for our purposes, so we shall describe how to refine it in the next paragraph. 
The tree $\cotree{x}$ is represented explicitly by a data structure that supports LCA, level/marked ancestor, and left/right queries in $\Otild(1)$ time. The computation time of $\cotree{x}$  is $\Otild(|\hat H_x|) = \Otild(|\cotree{x}|)$. 

We define the dual coarse tree $\hat{T}_x^*$ as the  refinement of $\cotree{x}$, as follows. 
Cf.~\cref{fig:CoarseTree_Dual}. 
We replace the parts of $\cotree{x}$ embedded in $R_{i+1}^{out}$ with parts of $VD^*(x,R_{i+1}^{out})$. However, we need to do this carefully because there is a difference between the structure of $\cotree{x}$ and $\VD^*(x,R_{i+1}^{out})$  that arises from 
the fact that in $\VD(x,R_{i+1}^{out})$, each site was artificially made to belong to its own Voronoi cell (cf. \cref{lem:VD0-Tx,cor:VD-Tx}) and from the fact that in $\VD^*(x,R_{i+1}^{out})$ there are many copies of the face corresponding to $R_i^{out} \cap R_{i+1}$.

\begin{figure}[H]
    \centering
    \includegraphics[width=0.3\textwidth]{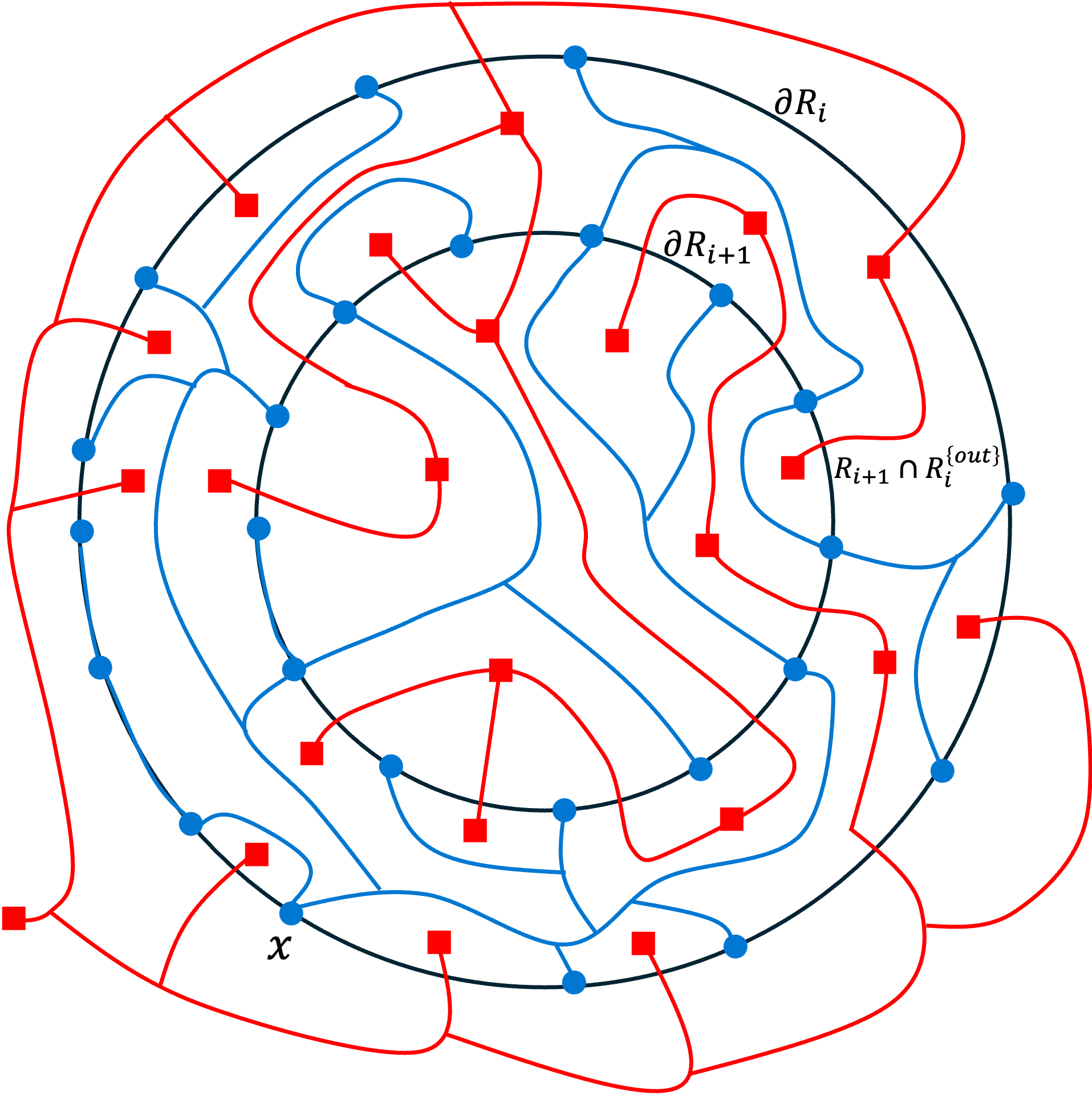}
    \hspace{0.5in}
    \includegraphics[width=0.3\textwidth]{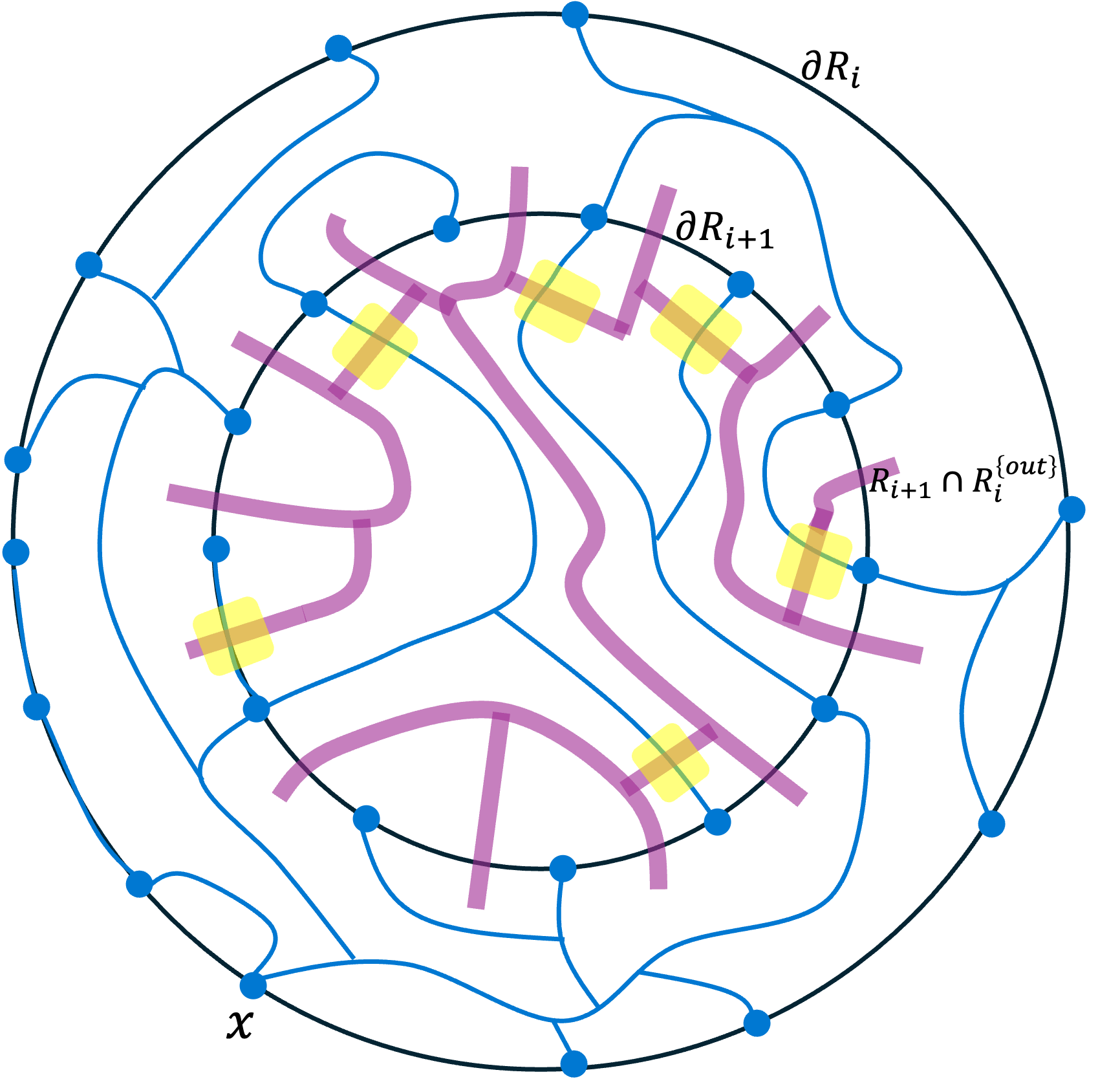}
    \hspace{0.5in}
    \includegraphics[width=0.4\textwidth]{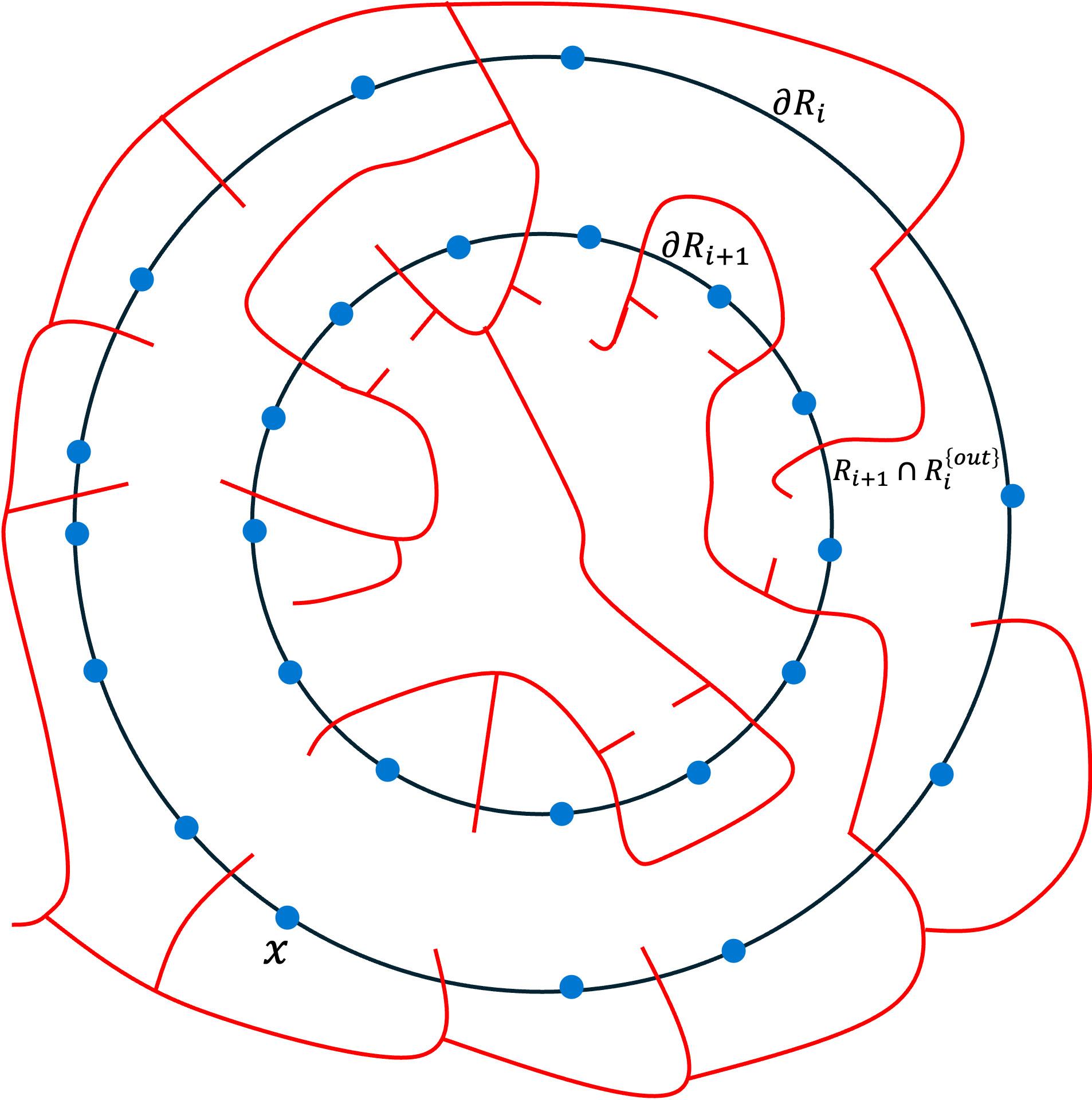}
    \caption{
    A schematic illustration of the dual coarse tree $\hat T^*_x$. 
    The source $x \in \partial R_i$ of the underlying fine shortest path tree $T_x$ is labeled explicitly.
    The boundary vertices of $R_i$ and $R_{i+1}$ are shown as blue nodes.
    Top-Left: The coarse tree $\hat T_x$ is shown in blue, and its dual co-tree $\cotree{x}$ (in $\hat H^*_x$) is shown in red.
    Top-Right: The coarse tree $\hat T_x$ is again shown in blue. The Voronoi diagram $\VD^*(x, R_{i+1}^{out})$ is drawn in purple, with its edges that belong to $T_x$ highlighted in yellow.
    Bottom: The dual coarse tree $\hat T^*_x$ is shown in red. 
    The tree $\hat T^*_x$ is obtained by "gluing" the red coarse tree $\cotree{x}$ and the purple VD, after breaking each highlighted edge.
    }
    \label{fig:CoarseTree_Dual}
 \end{figure}

Specifically, we go over the Voronoi edges of $\VD^*(x,R_{i+1}^{out})$. 
Recall that each Voronoi edge $\tilde e^*$ corresponds to a fragment of a bisector between two sites $s,t$ of $\partial R_{i+1}$. We check if $s$ and $t$ are connected in $\hat T_x$ by a path that does not leave $R_{i+1}^{out}$. 
This can be done in $\Otild(1)$ time by marking, during construction, all the edges of $\hat T_x$ that are not in $R_{i+1}^{out}$, and using marked ancestor queries.
If the condition is satisfied then the Voronoi edge $\tilde e^*$ must represent an edge whose primal counterpart is in $T_x$. Hence, we cut the Voronoi edge $\tilde e^*$
into two edges $\tilde e^*_1$ and $\tilde e^*_2$. Each of the $\tilde e^*_i$'s represents one of the two (possibly empty) subpaths of $T^*_x$ described in \cref{cor:VD-Tx}. 
Each of the $\tilde e^*_i$'s remains connected to one of the original endpoints of the edge $\tilde e^*$.
This process turns $\VD^*(x,R_{i+1}^{out})$ from a tree to a forest that represents no edge whose primal counterpart is in $T_x$.
Each vertex $\hat v^*$ of $\hat H^*$ that is embedded in $R_{i+1}^{out}$ corresponds to exactly one connected component of this forest.
We replace in $\cotree{x}$ each such vertex $\hat v^*$ with its corresponding component of the forest $\VD^*(x,R_{i+1}^{out})$.
We connect each Voronoi edge of this component of $\VD^*(x,R_{i+1}^{out})$ that represents a contracted path ending with an edge $e'^*$ of $\partial R_{i+1}$ with the coarse edge of $\cotree{x}$ which was previously incident to $\hat v^*$ and  represents the fine edge $e'^*$.
We note that formally, the edge $e'^*$ is now represented twice in $\hat T^*_x$. Once by the Voronoi edge and once by the edge of $\cotree{x}$. To guarantee that it will only be represented once, we ignore edges of $\partial R_{i+1}$ represented by Voronoi edges.

Since $\cotree{x}$ is a coarsening of $T^*_x$ and since the above process guarantees that the remaining edges of $\VD^*(x,R_{i+1}^{out})$ represent no edges whose dual is in $T_x$, we obtain
\begin{lemma}\label{lem:dual-coarse-tree}
    The dual coarse tree $\hat T^*_x$ is a coarsening of $T^*_x$.
\end{lemma}

\subsection{Decision and Objective Predicates}\label{subsec:preduality}

We would like to generalize our view of {\sc SimpleTreeElimination} to support operations that will be required when working with the recursive decomposition of the graph.

First we generalize the notion of a vertex being green, red or blue.
We call $\Comp$ the competitor set.
For $c \in \{x\}\cup\Comp$ we say that a vertex $v$ is $c$-colored (w.r.t. $\{x\}\cup\Comp$) if $v$ is closer (w.r.t. additive weights) to $c$ than to any other site in $\{x\}\cup\Comp$.
This definition implies a generalization of the critical edges of an edge $e$ w.r.t. $T_x$; 
Let $P^*_j$, $j \in \{1,2\}$ be the two dual paths forming the fundamental cycle of $e^*$ w.r.t. $T_x^*$. 
The ($x$-$\Comp$)-critical edge of $P^*_j$ is the first edge of $P^*_j$ that is not $x$-colored.
Note that the prefix property of \cref{lem:TwoColorMonotonicity} applies in the generalized setting by the same proof.

The second, more significant generalization we make, is to use the same tree elimination procedure to be able to identify not just an edge $\tilde e$ incident to the green vertex of the trichromatic face, but to also identify edges satisfying other desired properties.
This is done by considering the decision rule as a decision predicate $P_{x,\Comp}(e_1,e_2)$ on the critical edges $e_1,e_2$ in the context of the set of competitors $\{x,Y\}$, and considering the desired property as an objective predicate $Q_{x,\Comp}(\tilde e)$ on a single edge.
Then, we can generalize the correctness argument of {\sc SimpleTreeElimination} established by 
\cref{lem:SimpleTrichromaticDecisionCorrectness} and \cref{cor:SimpleCriteriaCorrectness} in the following lemma.

\begin{lemma}\label{lem:P_Q_formulation}
Consider a generalization of {\sc SimpleTreeElimination} in which we replace $g$ by $x$, $\{r,b\}$ by $Y$, and ${\sc TrichromatichDecisionRule}$ by a decision predicate $P_{x,\Comp}(\cdot,\cdot)$. 
Let $Q_{x,\Comp}(\cdot)$ be an objective predicate. 
Suppose that any edge $e \in T_x$ where the rootward vertex of $e$ is $x$-colored, satisfies that $P_{x,\Comp}(e_1,e_2) \leftrightarrow \exists\  \tilde e \in \textrm{ subtree of } e \textrm { in } T_{x} \ : Q_{x,\Comp}(\tilde e)$, where $e_1$ and $e_2$ are the critical edges of $e$ w.r.t. $T_x$.
Then, this version of {\sc SimpleTreeElimination} correctly returns an edge $e$ for which $Q(e)$ holds.
Furthermore, the running time of this version of {\sc SimpleTreeElimination} is dominated by the time for $\log(|H|)$ calls to {\sc SimpleFindCritical} and to the decision predicate $P$.
\end{lemma}

In terms of this generalization, \cref{thm:SimpleVoronoiSite} of the previous section used the formulation of \cref{lem:P_Q_formulation} with the decision predicate $P = P_{trichromatic}$ specified in \cref{alg:triDec}, and the objective predicate $Q=Q_{trichromatic}$ being an edge incident to the green vertex of the trichromatic face $\tilde f$.

We will later introduce another pair of predicates which will be used to locate the last $x$-colored edge along a prefix of a bisector.

\subsection{The Procedure {\sc TreeElimination}}
In its general form, the goal of the procedure  {\sc TreeElimination} is as follows. Let $x$ be some site, and let $T_x$ be the shortest path tree rooted at $x$. The goal of the procedure {\sc TreeElimination} is to find some edge $\tilde e$ that satisfies an objective predicate $Q(\tilde e)$ in a subtree of $T_x$ that is represented by a coarse tree $\hat T$ at some level $\ell$ of the recursion. 
Thus, for example, to achieve our main goal and find the trichromatic face $\tilde f$ of three sites of a Voronoi diagram at recursion level $i$, we will use {\sc TreeElimination} on the entire tree $T_x$ which is represented by the level $i$ coarse tree $\hat T_x$ where $x$ is one of the sites, with the decision predicate $P_{trichromatic}$ (where the other two sites are the competitor set).

At each step of the search we work with some subtree $\hat T$ of $\hat T_x$ that has not yet been eliminated. 
The level $\ell$ of coarseness of $\hat T$ may be $i$ or greater.
Let $\hat e=uv$ be a centroid edge of $\hat T$.
Let $e$ be the fine edge associated with $\hat e$ in $\hat T$.
The algorithm calls the procedure {\sc FindCritical} to obtain the two fine critical edges on the fundamental cycle of $e$ w.r.t. the fine tree $T_x$.
The main difficulty in finding the critical edges is that, unlike {\sc SimpleFindCritical} of~\cref{sec:simple}, we only have access to the coarse dual tree, but we can not afford to inspect the entire fine tree. 
Nonetheless the procedure finds and returns a pair of {\em fine} critical edges.
We describe this procedure in~\cref{subsec:findCritical}.

Having found the critical edges, the procedure {\sc TreeElimination} then decides whether to recurse on $ \{\hat e\} \cup \hat T(v)$ or on $\hat T \setminus \hat T(v)$ using the appropriate decision rule. 
See lines \ref{ln:CS_endblock_start}-\ref{ln:CS_endblock_end}  in \cref{alg:centsearch}
(\cref{alg:CS-stack} is analogous to \cref{ln:SCS_analogy} of {\sc SimpleTreeElimination}).

Eventually, {\sc TreeElimination} is called on a coarse subtree consisting of just a single coarse edge $\hat e=uv$.
The following lemma states how the desired fine edge $\tilde e$ relates to the single coarse edge
that has not been eliminated. This dictates how the search should be continued on a refined coarse tree.
\begin{lemma} \label{lem:coarse-4-subtrees}
    Suppose {\sc TreeElimination} is called on a coarse subtree consisting of just a single coarse edge $\hat e=uv$.
    Let $\tilde T$ be the maximal subtree of $T_x(u)$ that does not contain any vertex of $\hat H_x$ other than $u$ and $v$ (see~\cref{fig:T_tilde}).
    The fine edge $\tilde e \in T_x$ that satisfies the objective predicate $Q_P$, if it exists, is in $\tilde T$. 
\end{lemma}

\begin{figure}[H]
    \centering
    \begin{minipage}{0.28\textwidth}
        \centering
        \includegraphics[width=\linewidth]{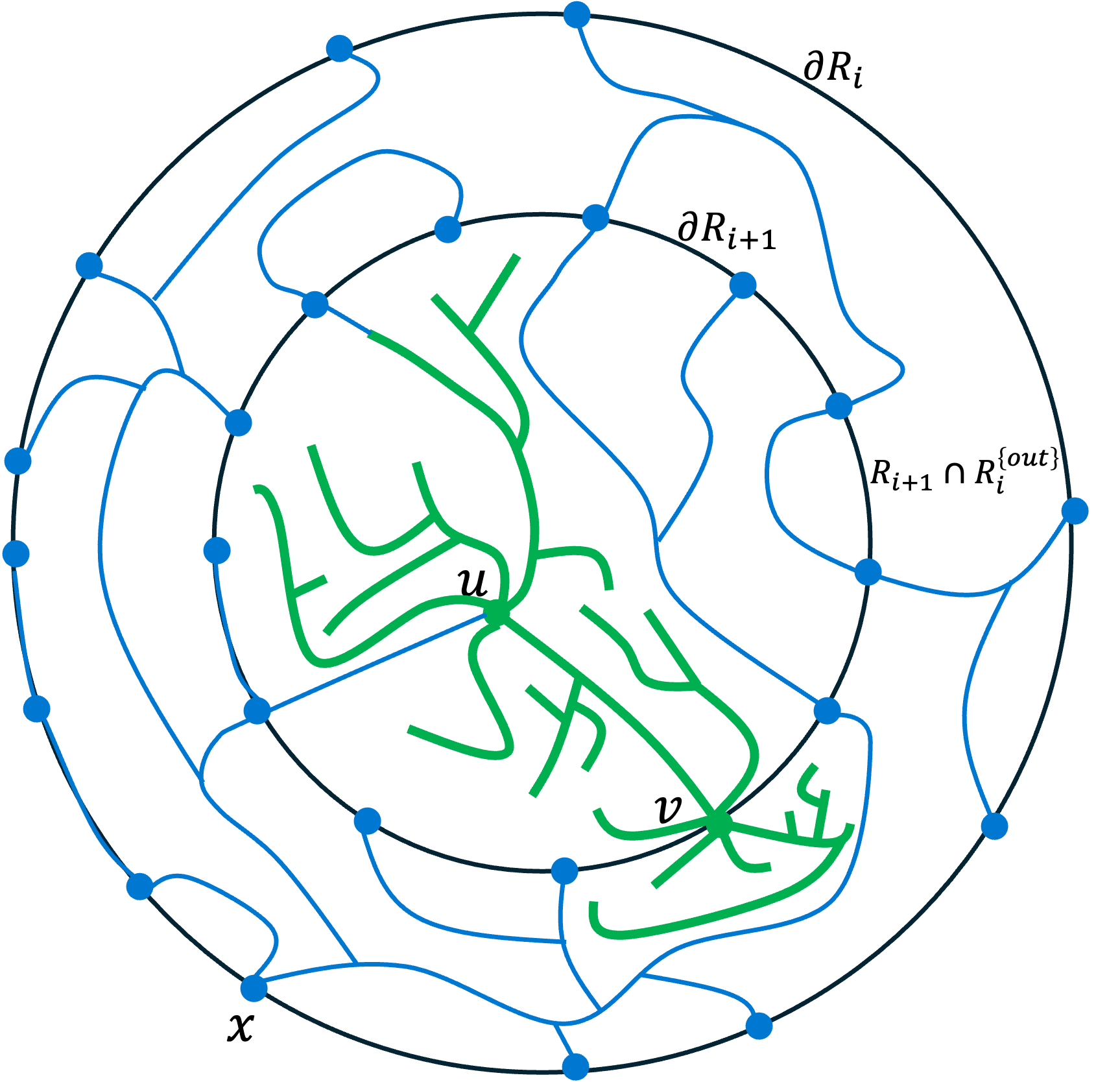}
    \end{minipage}%
    \hfill
    \begin{minipage}{0.7\textwidth}
        \captionof{figure}{%
            A schematic diagram of $\tilde T$ from Lemma~\ref{lem:coarse-4-subtrees}. 
            The boundary vertices of $R_i$ and $R_{i+1}$ are shown as blue nodes.
            The vertices $u, v$ are indicated by green nodes.
            The edges of the coarse tree $\hat T_x$ which are not in $\tilde T$ are indicated by blue lines.
            The edges of $\tilde T$ are indicated by green lines.
            }
        \label{fig:T_tilde}
    \end{minipage}
\end{figure}

\begin{proof}
    Recall that we assume that the predicate $P$ used in procedure {\sc TreeElimination} satisfies the following property: 
    For any edge $e \in T_x$ whose rootward vertex is $x$-colored,  $P_{x,\Comp}(e_1,e_2) \leftrightarrow \exists\  \tilde e \in \textrm{ subtree of } e \textrm { in } T_{x} \ : Q_{x,\Comp}(\tilde e)$ (for $P_{trichromatic}$ this was already established in \cref{lem:SimpleTrichromaticDecisionCorrectness}, for the other predicate that was not yet introduced, this will be shown in \cref{lem:BisectorPrefixDecisionCorrectness}).
    This guarantees that whenever {\sc TreeElimination} recurses on $\hat T(\hat v) \cup \{\hat e\}$ (in \cref{ln:CS_endblock_if_then}), then $\tilde e$ is in the subtree of $T_x$ rooted at $e$ (defined in \cref{ln:CS_endblock_start} as the fine edge associated with $\hat e$).
    Conversely, if {\sc TreeElimination} recurses on $\hat T \setminus \hat T(\hat v)$ then $\tilde e$ is not in the subtree of $T_x$ rooted at $e$.

    It follows that if only one coarse edge remains, then for every coarse edge $\hat d$ that is an ancestor of $\hat e$ in $\hat T$, $\tilde e$ is in the subtree of $T_x$ rooted at the fine edge $d$ associated with $\hat d$.
    Conversely, for every coarse edge $\hat d$ that is not an ancestor of $\hat e$ in $\hat T$, $\tilde e$ is not in the subtree of $T_x$ rooted at the fine edge $d$ associated with $\hat d$.

    Assume towards contradiction that $\tilde e \notin \tilde T$.
    
    If $\tilde e$ is a descendant of $u$, then let $w$ be the lowest ancestor of $\tilde e$ in $T_x$ that belongs to $\hat H_x$. 
    Since $\tilde e \notin \tilde T$, $w\neq v$ and $w$ is a descedent of $u$.
    Hence, the parent edge $\hat d$ of $w$ in $\hat T$ is not an ancestor edge of $u$, and therefore $\tilde e$ is not in the subtree of $T_x$ rooted at the fine edge $d$ associated with $\hat d$. 
    However, $w$ is in the subtree of $T_x$ rooted at the $d$, so $\tilde e$ must also be in that subtree, a contradiction. 
    
    Otherwise, $\tilde e$ is not a descendant of $u$.
    Let $w$ be a common ancestor of $u$ and $\tilde e$ that belongs to $\hat H_x$.
    Since $\tilde e \notin \tilde T$, $w$ is a strict ancestor of $u$ in $\hat T$.
    The coarse edge $\hat d$ leading from $w$ to $u$ in $\hat T$ exists but was eliminated.
    Since $\hat e$ was not eliminated, then $\hat d$ was eliminated when {\sc TreeElimination} was called on some strict descendant $\hat d'$ of $\hat d$ that is an ancestor of $\hat e$ and that $\tilde e$ is in the subtree of $T_x$ rooted at the fine edge $d'$ associated with $\hat d'$, a contradiction.
    
\end{proof}

To continue the search for the fine edge $\tilde e$ that satisfies the objective predicate, {\sc TreeElimination} should now recurse on a coarse tree at a finer level of coarseness that will lead to $\tilde e$.
By \cref{lem:coarse-4-subtrees}, we focus on the subtree $\tilde T$, the maximal subtree of $T_x(u)$ that does not contain any vertex of $\hat H_x$ other than $u$ and $v$.
Note that parts of $\tilde T$ might not even be represented by the coarse tree $\hat T_x$ at level $\ell$.
Instead, we consider each of the fine edges $e'$ incident to $u$ or to $v$.
Recall that the degree of each vertex is $O(1)$. Therefore, there are only $O(1)$ such fine edges $e'$. 

We would like to identify the fine edge $e'$ to whose subtree $\tilde e$ belongs.
Before describing how this is done, we need to discuss how to keep track of the levels of coarseness along the computation. 
To this end, the procedure {\sc TreeElimination} maintains a stack $X$ with a sequence of sites of increasingly finer  resolution (and hence that belong to regions of increasing levels) that starts with the original site $x$ with which the base recursive call was made. With each site in the stack we also specify the Voronoi diagram to which it belongs. See the text after \cref{prop:refinement-correctness} for the description of how this stack is populated.

Now, going back to identifying the fine edge $e'$ to whose subtree $\tilde e$ belongs, we first identify, for each of the $O(1)$ candidates $e'$, the minimal level $\ell'\geq \ell$ at which $e'$ is the associated edge with some coarse edge $\hat e'$ (recall the association of fine edges to coarse edges from \cref{sec:dualtreedef}).
This is done as follows.
We create a temporary copy $X'$ of the stack $X$.
We go over the levels, one after the other, starting at level $\ell$.
At each new level $k>\ell$ we push onto $X'$ the tuple $(x_k, \VD^*(x_{k-1}, R_k^{out}))$, where $x_k$ is the lowest ancestor of $u$ (or $v$) in $\hat T_{x_{k-1}}$ that belongs to $\partial R_k$.
We stop when we have reached the level $\ell'$ at which $e'$ is the associated fine edge of one of the $O(1)$ coarse edges $\hat e'$ incident to $u$ or $v$ in the coarse tree $\hat T_{x_{\ell'}}$ at the level $\ell'$.
Cf. \cref{alg:CS-ell} of {\sc TreeElimination}.

Having found $\ell'$, we call {\sc FindCritical} on $\hat e'$ with stack $X'$ (\cref{alg:CS-ell-call}) and apply the decision rule on the critical edges (\cref{alg:CS-ell-apply}).
In this way we identify the single fine edge $e'$ of $\tilde T$, incident to either $u$ or $v$, for which the decision rule is satisfied.
Note that in fact if $e'$ is incident to $v$ then the decision rule is satisfied also for the fine edge associated with the coarse edge $uv$.
If this is the case we take $e'$ to be the leafward edge, i.e. the one incident to $v$.
From the definition of $\tilde T$ it follows that if $e'$ is in $R_\ell^{out}\cap R_{\ell+1}$ then $\tilde e$ is also in $R_\ell^{out}\cap R_{\ell+1}$, and if $e'$ is in $R_{\ell+1}^{out}$ then $\tilde e$ is also in $R_{\ell+1}^{out}$.

If $e'$ is in $R_\ell^{out}\cap R_{\ell+1}$, then we continue the tree elimination procedure on the fine shortest path tree of $T_x$ in $R_\ell^{out}$, which is available in the MSSP data structure for $R_\ell^{out}\cap R_{\ell+1}$ from part (A) (lines \ref{ln:case_annulus_start}-\ref{ln:case_annulus_end}).

Otherwise,  $e'$ is in $R_{\ell+1}^{out}$. 
We assume that $e'$ is incident to $u$ (otherwise just replace $u$ with $v$ in the following description).  
Let $\gamma$ be lowest ancestor of $u$ in $\hat{T}_x$ that belongs to $\partial R_{\ell+1}$ (possibly $\gamma=u$).
The subtree of $\tilde T$ rooted at $e'$ is a contained in the subtree of $T_x$ rooted at $\gamma$ confined to $R_{\ell+1}^{out}$.
Therefore, it suffices to search for $\tilde e$ in the latter.
Let $T_\gamma$ be the shortest path tree in $R_{\ell+1}^{out}$ rooted at $\gamma$.
Note that $T_\gamma$ is not necessarily a subtree of $T_x$ (though they necessarily agree on the path from $\gamma$ to $u$).
However, by definition of Voronoi cells, the following holds:
The subtree of $T_x$ rooted at $\gamma$ confined to $R_{\ell+1}^{out}$ is exactly the subtree of $T_\gamma$ that does not leave the Voronoi cell $\Vor(\gamma)$ of $\gamma$ in $\VD(x,R_{\ell+1}^{out})$.
Hence we obtain:
\begin{proposition}\label{prop:refinement-correctness}
The edge $\tilde e$ is contained in the subtree of $T_\gamma$ that does not leave the Voronoi cell $\Vor(\gamma)$ in $\VD(x,R_{\ell+1}^{out})$.
\end{proposition}

\alglanguage{pseudocode}
\begin{algorithm}[t!]
\caption{\textsc{TreeElimination}$(\ell,x_b,Y,\hat T,X,P)$
\label{alg:centsearch}}
\textbf{Input:} 
\begin{itemize}[itemsep=0pt, topsep=4pt]
    \item $\ell$  - recursive level
    \item $X$ - stack of pairs $(x_{b+1}, \VD^*(x_b, R_{b+1}^{out})), \dots , (x_{\ell}, \VD^*(x_{\ell-1}, R_{\ell}^{out}))$ 
    \item $x_b$ - root site at level $b$
    \item $Y$ - set of competitor sites
    \item $\hat T$ - subtree of shortest path tree rooted at $x_\ell$
    \item $P$ - decision predicate
\end{itemize}
\textbf{Output:} A fine edge $e$ in the intersection of Voronoi cells in the stack $X$ that is represented in $\hat T$ and satisfies the objective predicate $Q_P$. 
\begin{algorithmic}[1]
    \If{$\hat T$ consists of a single edge $\hat e=(u, v)$}
            \If {$\hat e$ is a fine edge in $R_\ell^{out} \cap R_{\ell+1}$}\Comment{base of recursion}
                \State \Return $\hat e$
            \EndIf
            \For{$e'=(u',v') \leftarrow$ fine edges incident to $u$ or $v$}
                \State Let $\ell'\geq \ell$ be the smallest level at which $e'$ is the associated edge of some coarse edge $\hat e'$ \label{alg:CS-ell}
                \State Let $X'$ be the corresponding stack upon reaching $\ell'$ in the search for $e'$
                \State $e_1,e_2$ = {\sc FindCritical}$(\ell', X', x_b, e', \hat e', Y)$ \label{alg:CS-ell-call}
                \If{P($e_1, e_2$)} \label{alg:CS-ell-apply}
                    \If{$e' \in R_\ell^{out} \cap R_{\ell+1}$} \label{ln:case_annulus_start}
                        \State $\hat T \leftarrow $ shortest path tree of $R_\ell^{out}$ rooted at $x$ from MSSP structure (part (A))
                        \State $\hat T \leftarrow $ subtree of $\hat T$ rooted at $u'$
                        \State \Return {\sc  TreeElimination}$(\ell,x_b,Y,\hat T,X,P)$ \label{ln:case_annulus_end}
                    \Else \Comment{advance in $R_{\ell+1}^{out}$} \label{ln:case_out_start}
                        \State $x_{\ell+1} \leftarrow $ lowest ancestor of $u$ on $\partial R_{\ell+1}$
                        \State $\hat T \leftarrow $ shortest path tree of $R_{\ell+1}^{out}$ rooted at $x_{\ell+1}$
                        \State push $(x_{\ell+1}, \VD^*(x_\ell, R_{\ell+1}^{out}) )$ onto $X$
                        \State \Return {\sc  TreeElimination}$(\ell+1,x_b,Y,\hat T,X,P)$ \label{ln:case_out_end}
                    \EndIf
                \EndIf
            \EndFor
    \EndIf
    \State Let $e=uv$ be the fine edge associated with the centroid edge $\hat e = (u,\hat v)$ of $\hat T$ \Comment{$u$ is rootward} \label{ln:CS_endblock_start}
    \If {($v$ is not $x_b$-colored) or ($v$ is not in $\Vor(x_i)$ in $\VD^*(x_{i-1}, R_{i}^{out})$ for some $x_i \in X)$} \label{alg:CS-stack}
        \State \Return \textsc{TreeElimination}$(\ell,x_b,Y,\hat T \setminus \hat T(\hat v),X,P)$
    \EndIf
    \State $e_1,e_2$ = {\sc FindCritical}$(\ell, X, x_b, e, \hat e, Y)$
    \If{P($e_1, e_2$)}
        \State \Return \textsc{TreeElimination}$(\ell,x_b,Y,\hat T(\hat v) \cup \{\hat e\},X,P)$ \label{ln:CS_endblock_if_then}
    \Else
        \State \Return \textsc{TreeElimination}$(\ell,x_b,Y,\hat T \setminus \hat T(\hat v),X,P)$ \label{ln:CS_endblock_end}
    \EndIf
    \Statex
\end{algorithmic}
\vspace{-0.4cm}%
\end{algorithm}

Accordingly, we recurse on the coarse tree $\hat T_\gamma$ at level $\ell+1$, and indicate that for the rest of the recursion we are only interested in edges that belong to $\Vor(\gamma)$ by pushing the information to the stack $X$ of Voronoi cells that is passed to the recursive call (lines \ref{ln:case_out_start}-\ref{ln:case_out_end}). 
Each entry in the stack $X$ is a tuple $(\gamma,\VD^*(x,R_{\ell+1}^{out}))$ containing the identity of site $\gamma$ and the identity of the respective $\VD^*(x,R_{\ell+1}^{out})$ diagram.

We therefore need to make our procedures aware of the stack $X$ by the following adaptations. 
\begin{itemize}
    \item When we handle a new centroid edge $\hat e'$ of $\hat{T}$, we must verify that $\hat e'$ indeed belongs to all the Voronoi cells on the stack $X$. See \cref{alg:CS-stack} of {\sc TreeElimination}. 
    This is done by calling, for each entry $(x_{i+1},\VD^*(x_i,R_{i+1}^{out}))$ on the stack, the procedure {\sc PointLocate}$(\VD^*(x_i,R_{i+1}^{out}), \hat v')$, for each endpoint $\hat v'$ of $\hat e'$, which returns the site of $\partial R_{i+1}$ in whose cell $\hat v'$ resides. Note that all the Voronoi diagrams $\VD^*(x_i,R_{i+1}^{out})$ in stack $X$ were already computed earlier in the construction algorithm.
    
    \item In the procedure {\sc FindCritical}, we must make sure that the fundamental cycle of the fine edge $e^*$ we consider is with respect to $T^*_x$, not with respect to $T^*_\gamma$ (or to any finer tree $T^*_{\gamma'}$ that we may work with deeper in the recursion). 
    This will be done by "shortcutting" the fundamental cycle of $e^*$ (w.r.t. the tree of the current recursive level) along the boundary of the Voronoi cells of the sites at earlier recursive levels through which the recursion reached the current edge. The description of this structure, its correctness, and use in the procedure {\sc FindCritical} is described in \cref{subsec:recFunCyc}.
\end{itemize}

Analogously to \cref{lem:P_Q_formulation} and from the entire discussion in the current section we obtain that the following lemma holds.

\begin{lemma}\label{lem:TreeElimination_correctness}
Let $P_{x,\Comp}(\cdot,\cdot)$ and $Q_{x,\Comp}(\cdot)$ be a pair of decision and objective predicates. Suppose that any edge $e \in T_x$ where the rootward vertex of $e$ is $x$-colored, satisfies that $P_{x,\Comp}(e_1,e_2) \leftrightarrow \exists\  \tilde e \in \textrm{ subtree of } e \textrm { in } T_{x} \ : Q_{x,\Comp}(\tilde e)$, where $e_1$ and $e_2$ are the critical edges of $e$ w.r.t. $T_x$.
Further suppose that the procedure {\sc FindCritical} is implemented correctly.
Then, {\sc TreeElimination} correctly returns an edge $e$ for which $Q(e)$ holds.

\end{lemma}

\subsection{The procedure {\sc FindCritical}}\label{subsec:findCritical}

The objective of {\sc FindCritical} is to find $e^*_j, j\in {1,2}$, the fine ($x_b$-$Y$)-critical edge on each of the rootward paths $P^*_j$ of $C^*_{x_b}$, the fundamental cycle of $e$ w.r.t $T^*_{x_b}$. The main difference in comparison to {\sc SimpleFindCritical} is that we do not have the fine representation of $C^*_{x_b}$, but only a coarse one.

The procedure (provided in \cref{alg:findCritical}) first decomposes $P^*_j$ into $O(m)$ segments identifying the fine edges that delimit each segment. 
This is done using the procedure {\sc SegmentBreakdown} which is described in \cref{subsec:recFunCyc}.
Then, the procedure preforms a binary search on the (fine) edges that delimit the segments to locate the first edge among them that is not $x_b$-colored.
The color of an edge can be determined using $O(1)$ distance queries from the competing sites, using Voronoi diagrams that were computed at earlier stages.

This way a single segment, $s$, that contains the fine edge $e^*_j$ is identified.
Next, $e^*_j$ is found by the procedure {\sc FindCriticalOnSegment} on segment $s$, which is described in \cref{subsec:findCriticalOnSegment}.

\alglanguage{pseudocode}
\begin{algorithm}[H]
\caption{\textsc{FindCritical}$(i,X,x_b,e,\hat e,Y)$
\label{alg:findCritical}}
\textbf{Input:}
\begin{itemize}[itemsep=0pt, topsep=4pt]
    \item $i$ - recursive level
    \item $X$ - stack of pairs $(x_{b+1}, \VD^*(x_b, R_{b+1}^{out})), \dots , (x_{i}, \VD^*(x_{i-1}, R_{i}^{out}))$ 
    \item $x_b$ - root site at level $b$
    \item $e$ - a fine edge in $T_{x_i}$ that is in the intersection of the Voronoi cells in the stack $X$
    \item $\hat e$ - the primal coarse edge in  $\hat T_{x_i}$ that represents  the fine edge $e$.
    \item $Y$ - set of competitor sites of $x_b$
\end{itemize}
\textbf{Output:} 
The two ($x_b$-$Y$)-critical edges $(e_1,e_2)$ of $e$ w.r.t. $T_{x_b}$.
\begin{algorithmic}[1]
    \For{$P^*_j, j \in \{1,2\}$} \Comment{paths of the fundamental cycle $C^*_{e}$}
        \State $S \leftarrow$ {\sc SegmentBreakdown} of $P^*_j$
        \State Using binary search eliminate all segments in $S$ except a single segment $s$
        \State $e_j \leftarrow$ {\sc FindCriticalOnSegment} on $s$ \label{line:FCcallsFCOS}
    \EndFor
    \State \Return $\{e_1, e_2\}$
\end{algorithmic}
\end{algorithm}

\subsection{Decomposing the Fundamental Cycle}\label{subsec:recFunCyc}

Before describing {\sc SegmentBreakdown} we establish some required structural properties.
Recall that $i$ is the current recursive level at which {\sc FindCritical} was called, $b$ is the level of the  base call to {\sc TreeElimination} (i.e. the level of the vertex w.r.t. which the first Voronoi diagram on the stack is defined).
Fix a fine edge $e=uv$ of $T_{x_b}$ that is in $R_i^{out}$.
Let $m > \ell \geq i$ be the level at which the edge $e$ is an edge of an MSSP structure (i.e. $e \in R_\ell^{out} \cap R_{\ell+1}$).
Recall that $x_b, \dots x_\ell$ are such that for every $b \leq k \leq \ell$, $x_k$ is the last vertex of $\partial R_k$ on $T_{x_b}(x_b,u)$.
Note that this implies the following:
\begin{enumerate}
    \item For any $b\leq k \leq j \leq \ell$, the vertex $x_j$ is the last vertex of $\partial R_j$ on $T_{x_k}(x_k,u)$. 
    \item For any $b\leq k \leq j \leq \ell$, $T_{x_j}(x_j,v)$ ends with $e$ and $T_{x_j}(x_j,v)$ is a suffix of $T_{x_k}(x_k,v)$.
    \item For any $b< k \leq \ell$, $e$ is strictly inside $\Vor(x_k)$ in $\VD(x_{k-1},R_k^{out})$.
\end{enumerate}

The following lemmas describe the structural relations between fundamental cycles of  $e^*$ w.r.t. dual trees $T^*_{x_k}$'s at different levels (i.e., for different $x_k$'s), and the relations between fundamental cycles of $e^*$ w.r.t. a dual tree $T^*_{x_k}$ and the Voronoi diagrams of sites $x_j$ for $j<k$.
This structure will allow us to represent the fundamental cycle of $e^*$ w.r.t. to $T^*_{x_b}$ using segments of the coarse trees $\hat T^*_{x_k}$ for $k$' in the range $[b,\ell]$ that we already have available.

For any $k$, let $P^{*k}$ denote the right rootward path of the fundamental cycle of $e^*$ w.r.t. $T^*_{x_k}$. 
The arguments for the left rootward paths are identical.
For the remainder of this section when we write $Vor(x_k)$ we mean $\Vor(x_k)$ in $\VD(x_{k-1},R_k^{out})$.

\begin{proposition}\label{cor:dual-fund-cycle-in-vor-cell-coincide-for-consecutive-levels}
The prefixes of $P^{*k}$ and $P^{*(k-1)}$
that are strictly in $\Vor(x_k)$ are the same. 
\end{proposition}
\begin{proof}
By definition of $\Vor(x_k)$, $T_{x_k} \cap \Vor(x_k)$ and $T_{x_{k-1}} \cap \Vor(x_k)$ are identical.
Since $e$ is an edge of $T_{x_k}$ strictly inside $\Vor(x_k)$, the prefixes in $\Vor(x_k)$ of the fundamental cycle of $e^*$ w.r.t. both trees are the same.
\end{proof}

By applying \cref{cor:dual-fund-cycle-in-vor-cell-coincide-for-consecutive-levels} inductively, we obtain
\begin{corollary}\label{cor:decomp}
    Let $b\leq k'\leq \ell$. The prefix of $P^{*k}$ that is strictly in 
    $\bigcap_{b < j \leq k'} \Vor(x_j)$
    is the same for all $b \leq k \leq {k'}$.

\end{corollary}

\begin{proof}
Applying \cref{cor:dual-fund-cycle-in-vor-cell-coincide-for-consecutive-levels} at each level $k$ between $k'$ and $b$ we get that:

The prefix of $P^{*k}$ in $\Vor(x_{k})$ equals the prefix of $P^{*(k-1)}$ in $\Vor(x_{k})$.

\noindent But therefore, for every $b < k \leq k'$:

The prefix of $P^{*k}$ in $\bigcap_{b < j \leq k'} \Vor(x_j)$ equals the prefix of $P^{*(k-1)}$ in $\bigcap_{b < j \leq k'} \Vor(x_j)$.

\noindent Hence the corollary follows.

\end{proof}

\paragraph{The Decomposition.}
\Cref{cor:decomp} implies a decomposition of $P^* = P^{*b}$ into $k\leq \ell-b$ segments $(s_\ell, s_{i_1}, \dots, s_{i_k})$, as follows.
Cf. \cref{fig:Rootward_Decomposition}.
Let $i_0=\ell$.
The first segment, $s_{\ell}=s_{i_0}$, is the maximal prefix of $P^*$ strictly in all $\Vor(x_j)$ for all $b \leq j \leq {\ell}$. 
We define the remaining segments inductively.
Suppose that we have already identified $s_{i_0}, s_{i_1}, ..., s_{i_j}$.
Consider the suffix ${P^*}'$ of $P^*$ that was not yet assigned to any segment.
Let $i_{j+1} < i_j$ be such that the first Voronoi cell that ${P^*}'$ leaves is $\Vor(x_{i_{j+1}})$. 
The segment $s_{i_j}$ is defined to be the maximal prefix of ${P^*}'$ strictly in all $\Vor(x_k)$ for all $b < k \leq {i_j}$. 
Equivalently, it is the prefix of ${P^*}'$ until it leaves $\Vor(x_{i_{j+1}+1})$.
We denote the decomposition of $P^*$ by $S(P^*)$.

Before describing {\sc SegmentBreakdown} we note that an immediate application of \cref{cor:decomp} implies that the each segment $s_{i_j}$ is a subpath of the fine dual tree $T^*_{x_{i_j}}$.

\begin{corollary}\label{lem:segments-fine}
   For any $0 \leq j \leq k$, the segment $s_{i_j}$ of $P^*$ is identical to $T^*_{x_{i_j}}(e^*_s,e^*_f)$, where $e^*_s$ and $e^*_f$ are the first and last edges of $s_{i_j}$  
\end{corollary}

\begin{figure}[htb]
    \centering
    \includegraphics[width=0.45\textwidth]{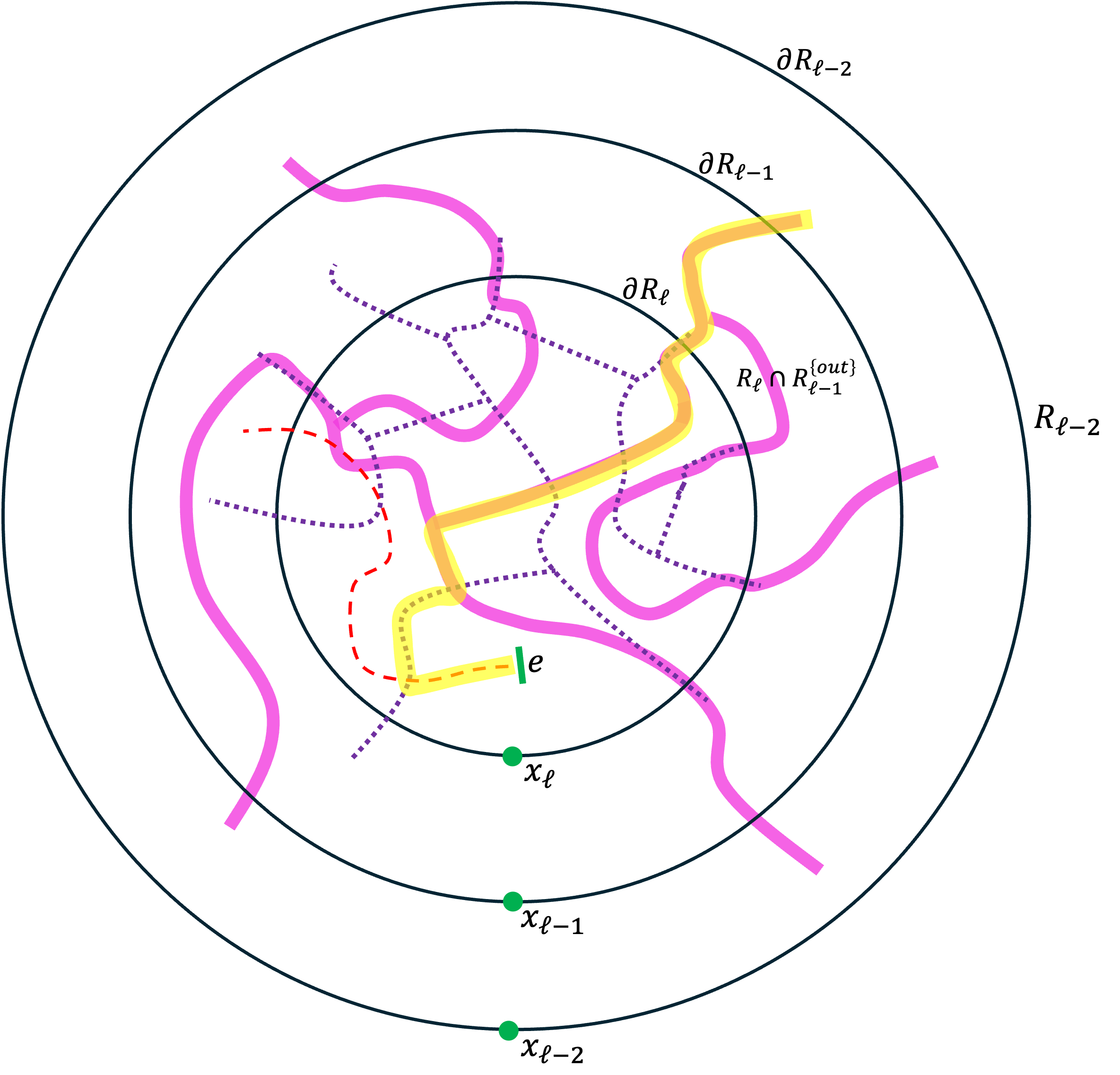}
    \caption{
    A schematic illustration of the decomposition induced by~\cref{cor:decomp}.
    The edge $e$ lies strictly within $\Vor(x_\ell)$ and $\Vor(x_{\ell-1})$.
    The rootward path $P^{*{\ell}}$ of the fundamental cycle of $e$ with respect to $T^*_{x_{\ell}}$ is shown as a dashed line.
    The Voronoi diagram $\VD^*(x_{\ell-1}, R_{\ell}^{out})$ is dotted in purple, and $\VD^*(x_{\ell-2}, R_{\ell-1}^{out})$ is drawn in bold pink.
    Highlighted in yellow is the rootward path $P^{*{(\ell-2)}}$ of the fundamental cycle of $e$ with respect to $T^*_{x_{\ell-2}}$.
    The path $P^{*{(\ell-2)}}$ consists of three concatenated segments: one along the dashed $P^{*\ell}$, one along the dotted $\VD^*(x_{\ell-1}, R_{\ell}^{out})$, and one along the bold $\VD^*(x_{\ell-2}, R_{\ell-1}^{out})$.
    }
    \label{fig:Rootward_Decomposition}
 \end{figure}

\paragraph{The Algorithm {\sc SegmentBreakdown}.}
We describe an algorithm that calculates the decomposition $S(P^*)$ of path $P^*$ into segments, where $\ell$ is the level at which the fine edge $e=uv$ is an edge of an MSSP structure. 
The algorithm begins with identifying the level $\ell$ by starting with the coarse edge $\hat e$ at level $i$. 
If $e$ is an edge of the MSSP structure at level $i$, then $\ell = i$. 
Otherwise, let $x_{i+1}$ be the lowest ancestor of $u$ in $\hat{T}_{x_i}$ that belongs to $\partial R_{i+1}$ (possibly $x_{i+1}=u$). 
We push onto the stack $X$ the pair $(x_{i+1},\VD(x_i,R_{i+1}^{out}))$.  
Note that, by choice, $x_{i+1}$ is the last vertex of $\partial R_{i+1}$ on $T_{x_b}(x_b,u)$.
By \cref{pro:gamma_subtree_level_association}, the coarse edge $\hat e$ (whose level is $i$) stores a pointer to the coarse edge $\hat e_{i+1}$ at level $i+1$ whose associated fine edge is $e$. 
We repeat the process with $\hat e_{i+1}$ taking the place of $\hat e$ and $x_{i+1}$ taking the place of $x_i$ until we find level $\ell$. 

The algorithm {\sc SegmentBreakdown} returns a 
a partition $S$ of $P^*$ into $k \leq \ell-b$ subpaths such that, for each subpath $s' \in S$, there is a level $b \leq j \leq \ell$ such that all fine edge $d'^* \in s'$ are represented by coarse edges in $\hat{T}^*_{x_j}$.
{\sc SegmentBreakdown} does so by locating, for each pair $(x_{i'}, \VD^*)$ in the stack $X$, the first dual edge on the subpath (implied by \cref{cor:decomp}) forming the intersection of $P^*$ with  $\Vor^*(x_{i'})$ (if the intersection is non-empty)\footnote{Note that $\Vor^*(x_{i'})$ is the boundary of $\Vor(x_{i'})$. I.e, the set of duals of edges with exactly one endpoint not in $\Vor(x_{i'})$. Hence the first such edge is the first edge of $P^*$ not strictly in $\Vor(x_{i'})$.}, as follows.

Let $\hat e_\ell$ be the coarse edge of $\hat T_{x_\ell}$ that represents the fine edge $e$.
Let $\hat P^{*\ell}$ be the right rootward path of the fundamental cycle of edge $\hat e^*_\ell$ w.r.t $\hat T^*_{x_\ell}$.
By \cref{lem:dual-coarse-tree}, $\hat P^{*\ell}$ is a coarsening of $P^{*\ell}$.
Since $P^*$ is not accessible to us we begin with $\hat P^{*\ell}$.
We assume that we have a mechanism $\mathcal M$ for finding the first fine edge $d^*$ represented by $\hat P^{*\ell}$ such that $d$ belongs to the boundary of $\Vor(x_i)$, for any $b\leq i\leq \ell$ (or reports that no such edge $d^*$ exists). 
The mechanism also returns the coarse edge  of $\hat T^*_{x_{i-1}}$ that represents $d^*$.
We provide this mechanism $\mathcal M$ in \cref{subsec:findCriticalOnSegment}.

Start by setting a path $\tilde P^*$ to be $\hat P^{*\ell}$, and let $i=\ell$. 
Utilizing mechanism $\mathcal M$ for $\tilde P^*$, we find the first (fine) edge $\tilde e^*$ of $\tilde P^*$ that is not strictly in $\Vor(x_{i-1})$. 
I.e., $\tilde e^*$ is an edge of the boundary of $\Vor(x_{i-1})$.
If $\tilde e^*$ does not exist, we keep $\tilde P^*$ unchanged.
Otherwise, we truncate $\tilde P^*$ just before $\tilde e^*$ so that after the truncation, by \cref{cor:decomp}, $\tilde P^*$ is the maximal common prefix of all $P^{*j}$ for $i-1 \leq j \leq \ell$.
We then decrement $i$ and repeat the process until $\tilde P^*$ is the maximal prefix common to all $P^{*j}$ for $b \leq j \leq \ell$.

Let $(i_1+1)$ be the last index in this process for which $\tilde P^*$ was truncated. Let $\tilde e^*_{i_1}$ be the fine edge on the boundary of $\Vor(x_{i_1+1})$ found at that time, and let $\hat e^*_{i_1}$ be the coarse edge representing $\tilde e^*_{i_1}$ in $\hat T^*_{x_{i_1}}$. We have thus identified the first segment, $s_\ell$ in the decomposition of $P^*$. In this segment $P^*$ coincides with a prefix of a rootward path in $T^*_{x_\ell}$ (by \cref{cor:decomp}).

Now that we have found the relevant level $i_1$, we move on to find the next segment of the decomposition,  $s_{i_1}$.  This segment starts with the fine edge $\tilde e^*_{i_1}$.
By \cref{lem:segments-fine}, the segment $s_{i_1}$ is a path on $P^{*{i_1}}$ and thus is represented by $\hat P^{*{i_1}}$.
We therefore repeat the above process starting by setting $\tilde P^*$ to be the suffix of $\hat P^{*{i_1}}$ that starts with the edge $\tilde e^*_{i_1}$, and only working at levels between $b$ and $i_1$.

We repeat the process until we have partitioned all of $P^*$. For every segment we encountered along the process we iterated over at most $m$ levels, using the mechanism $\mathcal M$ once at each level.  Thus the entire procedure requires $k\cdot O(m) = O(m^2)$ calls to mechanism $\mathcal M$.

The set of segments we found, $\{s_\ell$, $s_{i_1}$ .., $s_{i_{k}}\}$ (where $\ell>i_1>...>i_{k}$), forms the desired decomposition of $P^*$. 
The first segment starts at $e^*$. 
Every other segment begins with a (fine) dual edge on the boundary of $\Vor(x_{i+1})$ in $VD(x_i,R_{i+1}^{out})$, for some $i$.

\subsection{The Procedure {\sc FindCriticalOnSegment}}\label{subsec:findCriticalOnSegment}
We describe the implementation of mechanism $\mathcal{M}$. 
Recall that mechanism $\mathcal{M}$ identifies, for given indices $b\leq i< j\leq \ell$, the first fine edge $d^*$ represented by $\hat{P}^{*j}$ such that $d$ lies on the boundary of $\Vor(x_{i+1})$.
The mechanism also returns the coarse edge representing $d^*$ in $\hat T^*_{x_{i}}$. 

We describe a slightly more general procedure in subroutine {\sc FindCriticalOnSegment}; 
let $s$ be a rootward path in some coarse dual tree  that starts at a fine edge $e^*_{s}$ represented by the coarse edge $\hat e^*_s$, and ends with a fine edge $e^*_{f}$ represented by the coarse edge $\hat e^*_f$.
There exists some efficiently verifiable predicate $\Gamma$ (e.g., belonging to some Voronoi cell) that is satisfied only for a prefix of fine edges of $s$, and the goal of the procedure {\sc FindCriticalOnSegment} is to find the first fine edge $d^*$ of $s$, immediately after this prefix. 

In one use case, corresponding to the implementation of $\mathcal{M}$ used by {\sc SegmentBreakdown}, a fine edge $d^*$ satisfies the predicate $\Gamma$  iff $d^*$ is strictly in some Voronoi cell. 
In another use case, corresponding to the the call in \cref{line:FCcallsFCOS} of  {\sc FindCritical}, $d^*$ satisfies the predicate $\Gamma$ iff $d^*$ is $x_b$-colored with respect to a competitor set $Y$.

Note that in the case that {\sc FindCriticalOnSegment} is used to implement mechanism $\mathcal{M}$ we also need to report the coarse edge $\hat d_i^*$ representing $d^*$ in $\hat T^*_{x_{i}}$.
We find $\hat d_i^*$ by obtaining the colors of the primal endpoints of $d^*$ in $VD(x_i, R_{i+1}^{out})$. 
The coarse edge $\hat d_i^*$ is the Voronoi edge mapped to this pair of colors in the representation of $VD(x_i, R_{i+1}^{out})$.

Each coarse edge of $s$ has one or two fine edge associated with it. 
The procedure {\sc FindCriticalOnSegment} performs a binary search by evaluating the predicate on these associated fine edges to identify the last associated fine edge $e'^*$ that belongs to the prefix. 
The fine edge $e^*_c$ immediately following the prefix belongs to a path $\rho$ of fine edges that starts with $e'^*$ and ends at the fine associated edge after $e'^*$ on $s$. 
The path $\rho$ is represented by either a single or two coarse edges of $s$ since each coarse edge has an associated fine edge.
Let $i$ denote the level of the coarse edges of $s$. 
The path $\rho$ does not cross $\partial R_{i+1}$ because if a coarse edge $\hat \epsilon$ represents a fine edge of $\partial R_{i+1}$ then this fine edge is the associated edge of $\hat \epsilon$. 

There are two cases depending on whether $\rho$ is in $R_{i+1}$ or in $R_{i+1}^{out}$. We can distinguish the case by inspecting whether the endpoint of the first coarse edge representing $\rho$ is in $R_{i+1}$ or not.

\begin{enumerate}
\item If $\rho$ is in $R_{i+1}$ then it is represented by an MSSP structure from part (A).
We can then find $e^*_c$ by an explicit binary search on the path $\rho$ using the MSSP data structure. 

\item Otherwise, $\rho$ is in $R_{i+1}^{out}$, so it is represented by a single coarse Voronoi edge.
This is because coarse edges in $R_{i+1}^{out}$ are Voronoi edges, and because the two associated fine edges of each Voronoi edge are the first and last fine edges on the fine path represented by the Voronoi edge.
Thus, {\sc FindCriticalOnSegment} can find $e^*_c$ on the bisector represented by this single Voronoi edge $\hat e'^*$.

Let $\alpha$ and $\delta$ be the sites whose cells are separated by $\hat e'^*$. 
Recall that this pair of sites is stored in the representation of the Voronoi diagram.
The Voronoi edge $\hat e'^*$ is a segment of the bisector $\beta^*(\alpha,\delta)$ that starts with a fine edge $e^*_{\beta_1}$ and ends with a fine edge $e^*_{\beta_2}$. 
If $\hat e^*_s$ is the coarse edge $\hat e'^*$ then $e^*_s$ is the first fine edge along $\rho$, otherwise the first fine edge along $\rho$ is $e^*_{\beta_1}$. We denote this first fine edge along $\rho$ by $q^*_s$.
Similarly, we identify $q^*_f$, the last fine edge along $\rho$, from either $e^*_f$ or $e^*_{\beta_2}$.
We know that the fine critical edge $e^*_c$ we are looking for is between $q^*_s$ and $q^*_f$ on the bisector $\beta^*(\alpha,\delta)$. 

{\sc FindCriticalOnSegment} then calls the {\sc TreeElimination} on the level $(i+1)$ tree $\hat T_\alpha$, the (coarse) shortest path tree rooted at $\alpha$ in $R_{i+1}^{out}$, with the objective of finding an edge incident to $e^*_c$, the $\Gamma$-critical edge of the subpath of the bisector $\beta^*(\alpha,\delta)$. The edge $e^*_c$ is known to be between $q^*_s$ and $q^*_f$ on $\beta^*(\alpha,\delta)$. 
To be specific, {\sc TreeElimination} is called with level $i+1$, site $\alpha$, competitor set $\{\delta\}$, coarse tree $\hat T_\alpha$, an empty stack, and decision predicate $P_{bisector-\Gamma -Prefix}$, which is described next. 
We note that the predicate $\Gamma$ and the delimiting edges $q^*_s, q^*_f$ are implicitly provided in the decision predicate.

\end{enumerate}

\subsection{Bisector Prefix Predicate}\label{subsec:prefix_on_bisector}

The decision predicate $P_{bisector-\Gamma-Prefix}(e_1, e_2)$ is defined in \cref{alg:bisPreDec}.
We begin by showing that the computation time of the predicate is asymptotically comparable to the computation time of $P_{trichromatic}$, as defined in \cref{alg:triDec}, with the dominant cost incurred by a constant number of distance and point location queries to the parts of the oracle that were constructed so far.
To this purpose, we examine the three non-trivial core operations performed by \cref{alg:bisPreDec}, and show how each can be carried out using a constant number of such queries.

The first core operation is determining the color of a vertex among $\{\alpha, \delta\}$.
This is done by using a distance query from $\alpha$ and from $\delta$ using the distance oracle of $R_{i+1}^{out}$ which has already been computed.

The second core operation is determining if a vertex $w$ belongs to the subtree $T^e$.
This is done by simply examining the following distances in the current subgraph $R_{i+1}^{out}$, available in the distance oracle of $R_{i+1}^{out}$ at this stage: 
\begin{itemize}[itemsep=0pt, topsep=4pt]
    \item $d_{R_{i+1}^{out}}(\alpha, v)$; the distance from the root $\alpha$ to $v$, the leafward vertex of $e$ in $T_\alpha$.
    \item $d_{R_{i+1}^{out}}(\alpha, w)$; the distance from the root $\alpha$ to $w$.
    \item $d_{R_{i+1}^{out}}(v, w)$; the distance from $v$ to $w$.
\end{itemize}
By the uniqueness assumption of shortest paths, the shortest path from $\alpha$ to $w$ is incident to $v$ and therefore $w\in T^e$, if and only if $d_{R_{i+1}^{out}}(\alpha, w) = d_{R_{i+1}^{out}}(\alpha, v) + d_{R_{i+1}^{out}}(v, w)$.

The third core operation is determining, for any pair of fine edges $q^*_1, q^*_2$ on the bisector $\beta^*(\alpha, \delta)$, whether $q^*_1$ precedes $q^*_2$ along $\beta^*(\alpha, \delta)$.
It follows from \cref{lem:VD_shortest_path_color} that the left-to-right ordering in $T_\alpha$ of the $\alpha$-colored primal endpoints of edges on $\beta^*(\alpha, \delta)$ coincides with the edge order along the bisector.
Since the fine tree $T_\alpha$ is not directly available, we derive the left-to-right order of any given vertex pair $(x_1, x_2)$ in $T_\alpha$ by issuing a point location query to the oracle for each vertex, followed by a left-to-right query on an MSSP tree provided by part (A) of the oracle, as follows.

Let $S_1 = (\alpha, s'_{i+2}, s'_{i+3}, \ldots, x_1)$ and $S_2 = (\alpha, s''_{i+2}, s''_{i+3}, \ldots, x_2)$ denote the sequences of sites returned by the point location query, where we append $x_1$ and $x_2$ at the end of the corresponding sequence (unless it already appears as the last vertex).
Let $s'_j = s''_j$ be the last site common to $S_1$ and $S_2$.
Let $s'_{j+1}$ and $s''_{j+1}$ denote the successors of $s'_j$ in $S_1$ and $S_2$, respectively.
The MSSP tree $\mathcal{T}_{s'_j}$ rooted at $s'_j$ in $R_j^{out}$, provided by part (A) of the distance oracle, includes representations of both paths $T_\alpha(s'_j, s'_{j+1})$ and $T_\alpha(s'_j, s''_{j+1})$.
Observe that $T_\alpha(s'_j,s'_{j+1})$ and $T_\alpha(s'_j,s''_{j+1})$ are subpaths of the paths in $T_\alpha$ from $\alpha$ to $x_1$ and $x_2$, respectively.
Therefore, the third core operation reduces to determining the left‑to‑right order between 
$T_{s'_j}(s'_j, s'_{j+1})$ and $T_{s'_j}(s'_j, s''_{j+1})$, which can be directly answered in by a single query to the MSSP tree rooted at $s'_j$.

Next we prove the correctness of $P_{bisector-\Gamma-Prefix}(e_1, e_2)$, namely that when {\sc TreeElimination} is called with $P_{bisector-\Gamma-Prefix}(e_1, e_2)$ it returns an edge $e$ incident to $e^*_c$, if $e^*_c$ exists.
The correctness is actually shown by \cref{lem:TreeElimination_correctness}, provided that we show that the conditions of \cref{lem:TreeElimination_correctness} for $P_{bisector-\Gamma-Prefix}(e_1, e_2)$ are met.
This is done in the following  \cref{lem:BisectorPrefixDecisionCorrectness}.

\alglanguage{pseudocode}
\begin{algorithm}[h]
\caption{$\textsc{BisectorPrefixDecisionRule}(e_1,e_2)$}
\label{alg:bisPreDec}
\textbf{Input:} The ($\alpha$–$\{\delta\}$)-critical edges of $e=uv$ w.r.t. $T_\alpha$; where $e^*_1$ and $e^*_2$ lie on rootward dual paths $P^*_1,P^*_2$ of the fundamental cycle of $e$ in $T^*_\alpha$. $P^*_1, P^*_2$ are disjoint except for the rootmost dual vertex $q^*$.
\\
\textbf{Assumption:} There is a known predicate $\Gamma$, and two known dual edges $q^*_s$ and $q^*_f$ on $\beta^*(\alpha,\delta)$ such that $q^*_s$ is $\Gamma$-positive and $q^*_f$ is $\Gamma$-negative, and there is a single contiguous subpath of $\beta^*(\alpha, \delta)$ between $q^*_s$ and $q^*_f$ that satisfies $\Gamma$.
\\
\textbf{Output:} True if and only if $T^e = T_\alpha(v) \cup \{e\}$ contains an edge incident to the $\alpha$-colored endpoint of the $\Gamma$-critical edge $e^*_c$.
\begin{algorithmic}[1]
    \If{the $\alpha$-colored primal endpoint of $q^*_s$ and the $\alpha$-colored primal endpoint of $q^*_f$ both belong to $T^e$} \label{ln:Both}
        \State \Return true
    \ElsIf{exactly one of the $\alpha$-colored endpoints of $q^*_s$ and the $\alpha$-colored endpoint of $q^*_f$ belongs to $T^e$}
        \State Let $j$ such that one of the following holds: \label{ln:Single}
        \State \hspace{1em} (i) $e^*_j$ lies between $q^*_s$ and $q^*_f$ on $\beta^*(\alpha,\delta)$, or
        \State \hspace{1em} (ii) $e^*_j$ is null and $q^*$ lies between $q^*_s$ and $q^*_f$ on $\beta^*(\alpha,\delta)$
        \State {\bf if} $j$ is null {\bf then} \Return false
        \State {\bf if} $e^*_j$ does not exist {\bf then} set $e^*_j$ to be the last edge of $P^*_{3-j}$. \label{ln:E1_end_of_P2}
        \If{the $\alpha$-colored endpoint of $q^*_s$ belongs to $T^e$}
            \State \Return the negation of $\Gamma(e^*_j)$ \label{ln:Gamma_return}
        \Else
            \State \Return the value of $\Gamma(e^*_j)$
        \EndIf
    \ElsIf{($e^*_1$ and $e^*_2$ are both null) or ($q^*$ is $h^*$)} \label{ln:ConfC_Sub245}
        \State \Return false.
    \ElsIf{exactly one of $e^*_1, e^*_2$ is null} \label{ln:ConfC_Sub3}
        \State Let $j$ be such that $e^*_j$ is null
        \State Set $e^*_j$ to be the last edge of $P^*_{3-j}$ \label{ln:ConfC_Sub23_e*1}
    \EndIf

    \If{both $e^*_1$ and $e^*_2$ are not between $q^*_s$ and $q^*_f$ on $\beta^*(\alpha,\delta)$} \label{ln:ConfC_Order}
        \State \Return false
    \ElsIf{$\Gamma(e^*_1) \neq \Gamma(e^*_2)$} \label{ln:ConfC_Return}
        \State \Return true
    \Else
        \State \Return false
    \EndIf
    
\end{algorithmic}
\end{algorithm}

\begin{lemma}\label{lem:BisectorPrefixDecisionCorrectness}
Let $e_1,e_2$ be the ($\alpha$-$\{\delta\}$)-critical edges of $e=uv$ w.r.t. $T_\alpha$, such that $u$ is $\alpha$-colored (w.r.t. $\{\alpha,\delta\}$) and rootward of $v$.
Assume that a $\Gamma$-critical edge $e^*_c$ exists on the bisector subpath between $q^*_s$ and $q^*_f$. 
\cref{alg:bisPreDec} returns true if and only if there is an edge $\tilde e \in T^e = T_{\alpha}(v) \cup \{e\}$ such that $\tilde e$ is incident to the $\alpha$-colored endpoint of $e^*_c$.
\end{lemma}

\begin{proof}
Let $P^*_1$ and $P^*_2$ be the rootward paths of the fundamental cycle of $e=uv$ w.r.t. $T_{\alpha}$ that contain $e^*_1$ and $e^*_2$ respectively. Note that $P^*_1$ and $P^*_2$  share the common dual vertex $q^*$.
Let $x_s, x_f$ be the $\alpha$-colored endpoints of $q^*_s, q^*_f$, respectively.\\
We distinguish between the following three basic configurations based on the intersection of $x_s, x_f$ with $T^e$:
\begin{enumerate}[label=(\alph*), itemsep=0pt, topsep=4pt]
    \item \label{case:both-in-Te} Both $x_s$ and $x_f$ belong to $T^e$.
    \item \label{case:one-in-Te} $x_s$ belongs to $T^e$, while $x_f$ does not. The case where only $x_f$ belongs is symmetric and omitted.
    \item \label{case:none-in-Te} Neither $x_s$ nor $x_f$ belongs to $T^e$.
\end{enumerate}
For each configuration, we consider five subcases, determined by the existence of the critical edges and the position of $q^*$ (this is similar to the cases described in \cref{par:critical_edge_definition} and analyzed in the proof of \cref{lem:SimpleTrichromaticDecisionCorrectness}).
Recall that $T_\alpha$ is a shortest path tree in $R_{i+1}^{out}$.
Let $h^*$ denote the infinite face of $R_{i+1}^{out}$.
\begin{enumerate}[itemsep=0pt, topsep=4pt]
    \item \label{case:both-exist} Both critical edges exist.
    \item \label{case:one-exists-h} Exactly one critical edge exists and $q^*= h^*$.
    \item \label{case:one-exists-not-h} Exactly one critical edge exists and $q^* \neq h^*$.
    \item \label{case:none-exist-h} None of the critical edges exists and $q^* = h^*$.
    \item \label{case:none-exist-not-h} None of the critical edges exists and $q^* \neq h^*$.
\end{enumerate}
For each of the configuration we consider each of the five subcases.

For the first direction, assume \cref{alg:bisPreDec} returns true.
For each configuration we define an appropriate cycle $C$ in the plane that separates a subtree of $T^e$ from the rest of the graph.
The cycle $C$ is then used in a unified way to prove that the $\alpha$-colored vertex of $e^*_c$ is in that subtree of $T^e$.

\begin{enumerate}[label=(\alph*)]
    \item Both $x_s$ and $x_f$ belong to $T^e$.
    By \cref{ln:Both}, for any subcase in this configuration, \cref{alg:bisPreDec} returns true.
    Define cycle $C$ to start with $T_\alpha(u, x_s)$, continue along $\beta^*(\alpha,\delta)$ from $q^*_s$ until reaching $q^*_f$, then move rootward in $T_\alpha(u, x_f)$.

    \item $x_s$ belongs to $T^e$, and $x_f$ does not belong to $T^e$.
    
    For subcases (1), (2) and (3), w.l.o.g assume in \cref{ln:Single} that $j=1$ (note that since \cref{alg:bisPreDec} returns true, $j$ is not null).
    If $e^*_1$ is null then by \cref{ln:Single} it must be that $q^*$ is on $\beta^*(\alpha, \delta)$ between $q^*_s, q^*_f$, in which case $e^*_1$ is set in \cref{ln:E1_end_of_P2} to be the last edge of $P^*_2$.
    Note that in all three subcases, since \cref{alg:bisPreDec} returns true in \cref{ln:Gamma_return}, $\Gamma(e^*_1)$ is false.
    
    Define cycle $C$ to start with $T_\alpha(u, x_s)$, continue along $\beta^*(\alpha,\delta)$ from $q^*_s$ until reaching $e^*_1$, then travel along $P^*_1$ from $e^*_1$ towards $e^*$.

    For any other subcase in this configuration, \cref{alg:bisPreDec} returns false.
    
    \item Neither $x_s$ nor $x_f$ belong to $T^e$.

    By \cref{ln:ConfC_Sub245}, since \cref{alg:bisPreDec} returns true, subcases (2),(4),(5) do not occur.
    For subcase (3), w.l.o.g let $e^*_1$ be null and by \cref{ln:ConfC_Sub3} we set $e^*_1$ to be the last edge of $P^*_2$.
    Note that now for both subcases (1) and (3), since \cref{alg:bisPreDec} returns true then by \cref{ln:ConfC_Order}, $e^*_1$ and $e^*_2$ are between $q^*_s$ and $q^*_f$ on $\beta^*(\alpha,\delta)$.

    Define cycle $C$ to start with $e$, travel along $P^*_1$ towards $e^*_1$, continue along $\beta^*(\alpha,\delta)$ until reaching $e^*_2$, travel along $P^*_2$ from $e^*_2$ towards $e^*$. Cf. \cref{fig:BisectorPrefixProof_simple-figure}
\end{enumerate}

\begin{figure}[H]
    \centering
    \begin{minipage}{0.28\textwidth}
        \centering
        \includegraphics[width=\linewidth]{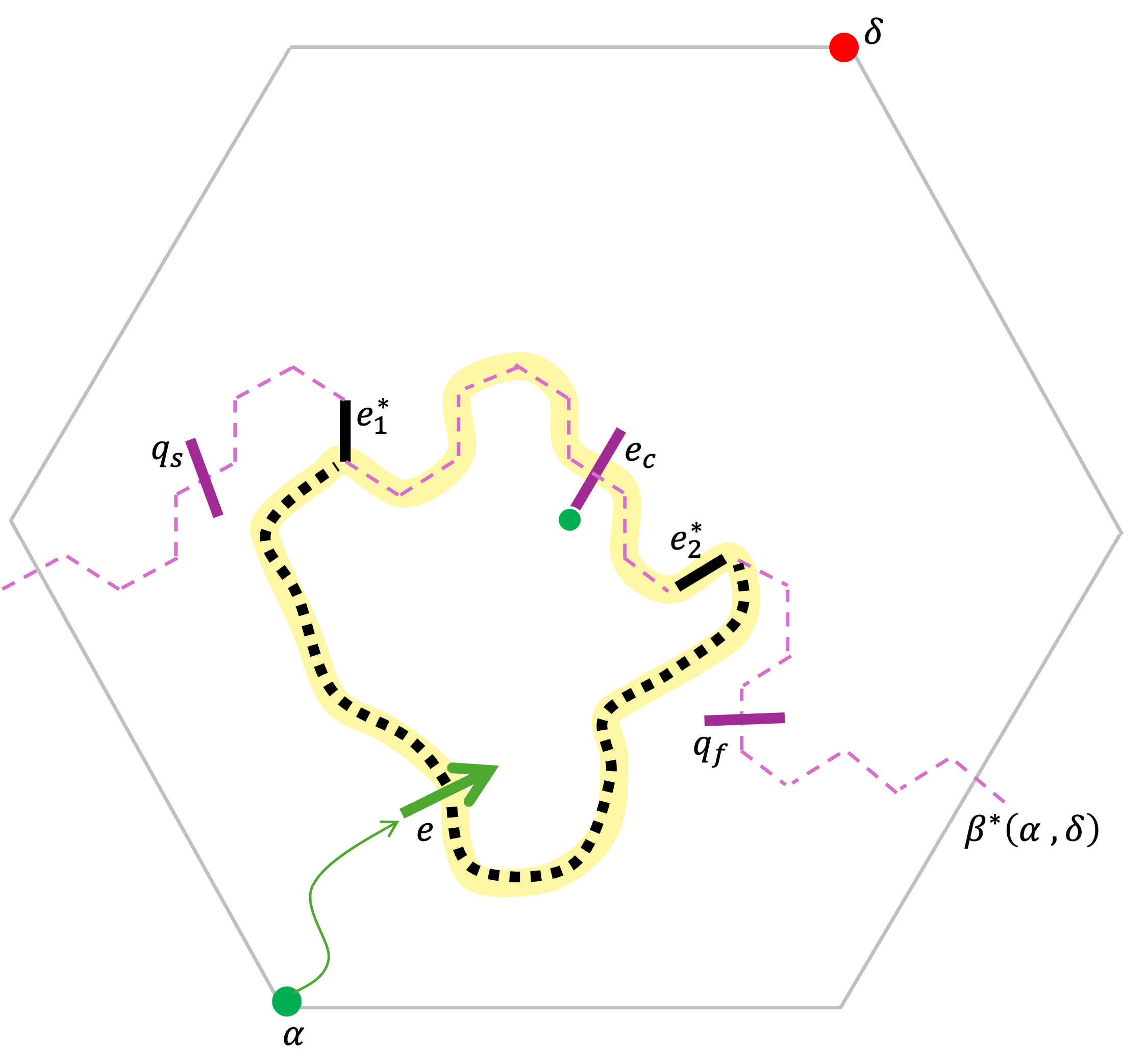}
    \end{minipage}%
    \hfill
    \begin{minipage}{0.7\textwidth}
        \captionof{figure}{%
            \textbf{The cycle $C$ described in the Proof of Lemma~\ref{lem:BisectorPrefixDecisionCorrectness}.} 
            The bold green arrow represents edge $e$ and the dashed black line represents the the green prefixes of paths $P^*_1, P^*_2$ of the fundamental cycle of $e^*$ w.r.t. $T^*_\alpha$.
            The bisector $\beta^*(\alpha, \delta)$ is indicated by a thin dashed purple line.
            The critical edges $e^*_1, e^*_2$ of the fundamental cycle are solid black.
            The primal edges corresponding to $q^*_s, q^*_f$ and to $e^*_c$ are solid purple.
            The $\alpha$-colored vertex of $e^*_c$ is indicated by a bold green node.
            The cycle $C$ is highlighted in yellow.
            The figure corresponds to the positive direction in subcase (1) of configuration (c) where both critical edges exist, and neither $x_s$ nor $x_f$ belong to $T^e$.
        }
        \label{fig:BisectorPrefixProof_simple-figure}
    \end{minipage}
\end{figure}

In all cases, since the the cycle $C$ consists only of portions of $P^*_1, P^*_2$, paths in $T^e$ and $\beta^*(\alpha,\delta)$, and since no edge of $T_\alpha$ is dual to an edge of $\beta^*(\alpha, \delta)$, then $C$ encloses only vertices of a subtree of $T^e$.
By construction, $C$ also encloses the $\alpha$-colored endpoints of all edges along $\beta^*(\alpha,\delta)$ between some edge satisfying $\Gamma$ and some edge that does not satisfy $\Gamma$.
Thus, $C$ encloses the $\alpha$-colored endpoint of $e^*_c$.
Therefore, the $\alpha$-colored endpoint of $e^*_c$ belongs to $T^e$ and some edge $\tilde e \in T^e$ is incident to it.

For the second direction, assume \cref{alg:bisPreDec} returns false.
We again define for each configuration an appropriate cycle $C$ in the plane.
Since no edge of $T_\alpha$ is dual to an edge of $\beta^*(\alpha, \delta)$, it follows that any subtree of $T_\alpha$ rooted at any single vertex or edge, and in particular $T^e$, can contain the $\alpha$-colored endpoints of at most a single segment of edges along the $\beta^*(\alpha,\delta)$ bisector.
Let $S$ be the maximal segment of edges along the $\beta^*(\alpha,\delta)$ bisector, between $q^*_s$ and $q^*_f$, with $\alpha$-colored endpoints in $T^e$.
The cycle $C$ is defined so that for all edges of $S$, $C$ encloses their $\alpha$-colored endpoints.

\begin{enumerate}[label=(\alph*)]
    \item Both $x_s$ and $x_f$ belong to $T^e$.
    None of the subcases applies because in this configuration, by \cref{ln:Both} \cref{alg:bisPreDec} returns true.

    \item $x_s$ belongs to $T^e$, and $x_f$ does not belong to $T^e$.

    Consider the fundamental cycle $C_e$ of $e$ w.r.t. $T^*_\alpha$.
    Since $C_e$ encloses $x_s$ but not $x_f$, $C_e$ must intersect the bisector $\beta^*(\alpha, \delta)$ along its non-empty segment between $q^*_s$ and $q^*_f$.
    Hence, a critical edge must exist, so subcases (4) and (5) are not possible. 
    Moreover, either one of the rootward paths that form $C_e$ contains a critical edge between $q^*_s$ and $q^*_f$, or $q^*$ is on $\beta^*(\alpha, \delta)$ between $q^*_s$ and $q^*_f$.
    Thus, in subcases (1), (2) and (3), $j$ of \cref{ln:Single} is not null, and we assume w.l.o.g that $j=1$.
    If $e^*_1$ is null then by \cref{ln:Single} it must be that $q^*$ is on $\beta^*(\alpha, \delta)$ between $q^*_s, q^*_f$, in which case $e^*_1$ is set in \cref{ln:E1_end_of_P2} to be the last edge of $P^*_2$.
    Note that in all three subcases, since \cref{alg:bisPreDec} returns false in \cref{ln:Gamma_return}, $\Gamma(e^*_1)$ is true.
    
    Define cycle $C$ to start with $T_\alpha(u, x_s)$, continue along $\beta^*(\alpha,\delta)$ from $q^*_s$ until reaching $e^*_1$, then travel along $P^*_1$ from $e^*_1$ towards $e^*$.

    Note that $C$ encloses all the $\alpha$-colored endpoints of edges on $\beta^*(\alpha,\delta)$ between $q^*_s$ and $e^*_1$.
    Also note that no edge on $\beta^*(\alpha,\delta)$ between $e^*_1$ and $q^*_f$ has an $\alpha$-colored endpoint in $T^e$.
    Thus, $S$ is enclosed by $C$.
    Since $q^*_s$ and $e^*_1$ both satisfy $\Gamma$, then by the monotonicity of $\Gamma$ along $\beta^*(\alpha,\delta)$ between $q^*_s$ and $q^*_f$, it follows that all of edges of $S$ satisfy $\Gamma$.
    Thus, the $\alpha$-colored endpoint of $e^*_c$ does not belong to $T^e$ and no edge $\tilde e \in T^e$ is incident to it.

    \item Neither $x_s$ nor $x_f$ belong to $T^e$.

    For subcase (2) and (3) w.l.o.g let $e^*_1$ be null and set $e^*_1$ to be the last edge of $P^*_2$, as set in \cref{ln:ConfC_Sub23_e*1}.

    We first consider the following:
    \begin{itemize}
        \item subcases (1) and (3), with the additional property that both $e^*_1$ and $e^*_2$ do not lie between $q^*_s$ and $q^*_f$ on $\beta^*(\alpha,\delta)$. This case returns false in \cref{ln:ConfC_Order} of \cref{alg:bisPreDec}.
        \item subcases (2), (4), (5). This case returns false in \cref{ln:ConfC_Sub245} of \cref{alg:bisPreDec}.
    \end{itemize}
    Recall that $S$ is the single maximal contiguous segment of edges along $\beta^*(\alpha,\delta)$ between $q^*_s$ and $q^*_f$ with $\alpha$-colored endpoints in $T^e$.
    Let $C_e$ denote the fundamental cycle of $e$ with respect to $T^*_\alpha$.  
    Neither $x_s$ nor $x_f$ belong to $T^e$ so $C_e$ encloses neither $x_s$ nor $x_f$.
    Hence, it follows by \cref{lem:VD_shortest_path_color} that if $C_e$ intersects the bisector $\beta^*(\alpha,\delta)$ outside the interval between $q^*_s$ and $q^*_f$, then it cannot also intersect it within that interval.
    In any subcase with no critical edges, $C_e$ does not intersect the bisector $\beta^*(\alpha,\delta)$ anywhere.
    Therefore, in each of the considered cases listed above, $S$ is empty and no vertex in $T^e$ is incident to an edge between $q^*_s$ and $q^*_f$ on $\beta^*(\alpha,\delta)$.
    Thus, $\alpha$-colored endpoint $e^*_c$ does not belong to $T^e$ and no edge $\tilde e \in T^e$ is incident to it.
    
    The remaining subcases are (1) and (3) with the additional property that both $e^*_1$ and $e^*_2$ lie between $q^*_s$ and $q^*_f$ on $\beta^*(\alpha,\delta)$. 
    This case returns false in \cref{ln:ConfC_Return} of \cref{alg:bisPreDec} if both $e^*_1, e^*_2$ satisfy $\Gamma$, or both do not.
    Define cycle $C$ to start at $e$, travel along $P^*_1$ towards $e^*_1$, continue along $\beta^*(\alpha,\delta)$ until reaching $e^*_2$, travel along $P^*_2$ from $e^*_2$ towards $e^*$.

    Note that $C$ encloses all the $\alpha$-colored endpoints of edges on $\beta^*(\alpha,\delta)$ between $e^*_1$ and $e^*_2$.
    Also note that no edge on $\beta^*(\alpha,\delta)$ between $q^*_s$ and $e^*_1$ or between $e^*_2$ and $q^*_f$ has an $\alpha$-colored endpoint in $T^e$.
    Thus, $S$ is enclosed by $C$.
    Since $e^*_1$ and $e^*_2$ either both satisfy $\Gamma$ or both do not, then by the monotonicity of $\Gamma$ along $\beta^*(\alpha,\delta)$ between $q^*_s$ and $q^*_f$, it follows that either all of edges of $S$ satisfy $\Gamma$ or all of edges of $S$ do not satisfy $\Gamma$.
    Thus, the $\alpha$-colored endpoint of $e^*_c$ does not belong to $T^e$ and no edge $\tilde e \in T^e$ is incident to it.
    
\end{enumerate}

\end{proof}

\subsection{Running Time Analysis}\label{subsec:analysis}

We now analyze the running time of {\sc TreeElimination} for finding a trichromatic face in a 3-site Voronoi diagram, proving the following.

\begin{theorem}\label{thm:tc-time}
    Let $Q$ be the query time of the distance oracle to be constructed.
    For any $R_i\in \mathcal{R}$, any 3 sites on $\partial R_i$, and any Voronoi diagram $\VD^*$ in $R_i^{out}$ w.r.t. these sites,
    The time to find the trichromatic face of $\VD^*$ is $(\log n)^{O(m)} \cdot Q$.
\end{theorem} 

Consider a call to {\sc TreeElimination}$(i,x_b,Y,\hat T,X,P)$, where $b$ is the recursive level of the site $x_b$, and $i\geq b$ is the recursive level of the tree $\hat T$ to be eliminated. 
We want to analyze the total running time of such a call to {\sc TreeElimination}. 

Let $T(n',i,\bottom)$ denote an upper bound on the maximum running time over any invocation of {\sc TreeElimination} on a coarse tree of size at most $n'$ at recursion level at least $i$, where the recursive level of the site $x_b$ is at least $\bottom$. Note that $i \geq b$.
Similarly, let $F(i,\bottom)$ denote an upper bound on the maximum running time over any invocation of {\sc FindCritical} on a coarse dual cycle at level at least $i$, with a site $x_b$ whose recursive level is at least $\bottom$.

We first argue that $T(n',i,\bottom)$ and $F(i,\bottom)$ are well defined. That is, that every call to {\sc TreeElimination} or to {\sc FindCritical} terminates. 
This is because every recursive call ever made either strictly decreases $n'$ without changing $i$ and $b$, or strictly increases $i$ without changing $b$, or strictly increases $b$. By inspecting the base cases, this implies that every call terminates. 
To derive recurrence relations for $T(n',i,b)$ and $F(i,\bottom)$, consider the call to {\sc TreeElimination} that achieves the maximum running time. 

As long as the tree contains more than a single edge, the procedure performs $O(m)$ distance queries and a single call to {\sc FindCritical} at level $i$, and then makes a single recursive call to {\sc TreeElimination}, but with a tree that is smaller by a constant factor. Hence,

\begin{equation}\label{eq:T1}
T(n',i,\bottom) \leq T\left(\frac{n'}{k},i,\bottom\right) + mQ + F(i,\bottom)  ,\quad \forall n' > 1,
\end{equation}
where $k>1$ is some constant.

When the tree contains just a single edge, the procedure determines a level $\ell'>i$ in $O(m)$ time, and calls {\sc FindCritical} at level $\ell'$. 
This is repeated $O(1)$ times, and then followed by a single recursive call to {\sc TreeElimination} with a new coarse tree at level at least $i+1$. Therefore, 

\begin{equation}\label{eq:T2}
T(1,i,\bottom) \leq T(n,i+1,\bottom) +  O(F(i,\bottom)) + O(m),\quad \forall i < m
\end{equation}
The base case is
\begin{equation}\label{eq:T_base}
T(1,m,\bottom) = O(1)
\end{equation}
because at level $m$ every coarse edge is in fact a fine edge.

We proceed to derive a recurrence for $F(i,\bottom)$.
{\sc FindCritical} performs exactly two iterations (one for each rootward path of the fundamental cycle). 
In each iteration it first calls {\sc SegmentBreakdown}.
Then, it eliminates the segments using binary search
that performs a total of $O(\log n)$ distance queries.
Finally, it calls {\sc FindCriticalOnSegment}.

{\sc SegmentBreakdown} performs no direct computation but only makes $O(m^2)$ invocations to {\sc FindCriticalOnSegment} (to realize mechanism $\mathcal M$).
Each call to {\sc FindCriticalOnSegment} eliminates the edges of the segment using binary search that makes $O(\log n)$ distance queries, and then performs at most one call to {\sc TreeElimination}, in which the recursion level is at least one greater than $\bottom$. Therefore, 

\begin{equation}\label{eq:FFS}
F(i,\bottom) \leq O(m^2)\cdot (\Otild(Q)+T(n,\bottom+1,\bottom + 1))
\end{equation}
Here, $\Otild$ suppresses polylogarithmic factors in $n$, the size of the input graph.

The base case is
\begin{equation}\label{eq:F_base}
F(m,m) = \Otild(Q),
\end{equation}
because in this case the only segment is at level $m$, so a binary search is performed on fine edges.

We now resolve the set of recurrence relations.
From \cref{eq:T1} and \cref{eq:T2} we get
\begin{equation}\label{eq:TR}
T(n',i,\bottom) \leq T(n,i+1,\bottom) + \Otild(mQ + F(i,\bottom))
\end{equation}
Repeatedly substituting \cref{eq:TR} into itself, using the fact that $F(i,b) \geq  F(i',b)$ for all $i'>i$, and incorporating \cref{eq:T1} and the base case of \cref{eq:T_base} yields
\begin{equation}\label{eq:T_F_only}
T(n',i,\bottom) \leq \Otild {(m \cdot(mQ + F(i,\bottom)))}
\end{equation}
Applying \cref{eq:FFS} or \cref{eq:F_base} into \cref{eq:T_F_only} gives us
\begin{align}
T(n',i,\bottom) & \leq  \Otild {(m \cdot(mQ + m^2 (Q +  T(n,\bottom+1,\bottom +1))} \nonumber\\
& = \Otild {(m^3Q +  m^3T(n,\bottom+1,\bottom +1))}\label{eq:T_F_recursive}
\end{align}
Repeating the substitution of \cref{eq:T_F_recursive}, and then applying \cref{eq:T1} unfolds the  recurrence down to the base case. Finally, incorporating the base case from \cref{eq:T_base}, and using the fact that $m\leq \log n$ in any recursive decomposition, results in 
\begin{equation}\label{eq:T_F_full}
T(n',i,\bottom) = (\log n)^{O(m)} \cdot Q
\end{equation}

\section{Near Optimal Construction of the Oracle of \cite{ourJACM}}\label{sec:jacm_complete}

To complete the proof of \cref{thm:main}, we discuss in \cref{subsec:complementary_element_D} how to construct the remaining components (D)–(E) for each Voronoi diagram used in the distance oracle of \cite{ourJACM} in $\Otild(S\cdot m \cdot Q)$, where $S$ is the number of sites of the Voronoi diagram, and $Q$ is the query time of the oracle.
This is dominated by the $S\cdot (\log n)^{O(m)} \cdot Q$ time of the construction of a Voronoi diagram of parts (B) and (C) with $S$ sites described in \cref{sec:complete}.

Together, these suffice for proving \cref{thm:main} and to construct the oracle of \cite{ourJACM} for all $m = o(\log n / \log \log n)$, but not for $m \in \Omega(\log n / \log \log n) \cap o(\log n)$.
In \cref{subsec:complementary_optimized_compute}
we improve the construction of parts (B) and (C) from \cref{sec:complete} for  $m \in \Omega(\log n / \log \log n) \cap o(\log n)$, which covers all tradeoffs for the oracle of \cite{ourJACM}.

\subsection{Efficient constructions of parts (D)--(E) of the oracle of \cite{ourJACM}}\label{subsec:complementary_element_D} 

Recall that the oracle works with a recursive $r$-division with $m$ levels.
In fact, there is no need to discuss part (E) and its construction in detail because the construction  already provided in \cite{ourJACM} is independent of all the other parts, and takes $\Otild(mn)$ time. 
This is dominated by the $n^{1+o(1)}$ construction time of the other parts of the oracle, for any $m = o(\log n)$, which covers all tradeoffs considered in \cite{ourJACM}.

We therefore only discuss part (D), which is defined in \cite{ourJACM} as follows:

\begin{enumerate}
  \item[(D)] \textbf{Site Tables and Side Tables.} For each Voronoi diagram $\VD^*(u,R_i^{out})$ from parts (B) or (C), and for every node $f^*$ in its centroid decomposition (where the primal vertices of $f^*$ are $y_j$, and their corresponding sites are $s_j$, $j\in\{0,1,2\}$), let $R_{i'}$ be an ancestor region of $R_i$ at some level $i'\ge i$. Then, for each such $R_{i'}$ and for every $j\in\{0,1,2\}$, store the ordered pair $(q,x)$, where:
    \begin{itemize}
      \item $q$ is the first vertex on the shortest $s_j$--to--$y_j$ path (in $R_i^{out}$) that lies on the boundary $\partial R_{i'}$, and 
      \item $x$ is the last vertex on the same path that lies on $\partial R_{i'}$.
    \end{itemize}
    In addition, we store the distance $\dist_G(u,x)$. Note that if the shortest $s_j$--to--$y_j$ path does not intersect $\partial R_{i'}$, then no $(q,x)$ pair exists; in that case, we store a single bit indicating whether $R_{i'}^{\mathrm{out}}$ lies to the right or to the left of the site--centroid--site chord\footnote{The site--centroid--site chord is the path consisting of the concatenation of the shortest path from $s_j$ to $y_j$, the edge $y_j y_{j-1}$ of the face $f$, and the reverse of the shortest path from $s_{j-1}$ to $y_{j-1}$.} 
    \[
    (s_j,\ldots,y_j,\, y_{j-1},\ldots,s_{j-1})
    \]
    in $R_{i'}^{\mathrm{out}}$.
  
\end{enumerate}

Part (D) is constructed for each Voronoi diagram of parts (B) and (C) immediately after  constructing each Voronoi diagram, level by level, starting at level $m$. 
We provide a procedure to compute all ingredients of both the site table and the side table for a given Voronoi diagram of $S$ sites in $\Otild (S\cdot m\cdot Q)$ time.

\paragraph{Site tables.}
Let $\VD^*(u, R_i^{out})$ be some Voronoi diagram of $R_i^{out}$ for the sites of $\partial R_i$ w.r.t. additive weights from some source $u\in R_i$; as computed in part (B) or (C).
Let $f^*\in \VD^*(u, R_i^{out}), R_{i'},y_j,s_j$ be as defined above. 
Let $P$ be the shortest $s_j$--to--$y_j$ path (in $R_i^{out}$), and let $(q_{i'}, x_{i'})$ be the first and last vertex on $P$ that lies on the boundary $\partial R_{i'}$.
Let $k>i$ be the smallest level for which $y_j\in R_{k}$.

We first observe the following monotonicity property of the vertices $x_{i'}$ along $P$, as shown in \cref{fig:last_vertices_order(part_d)}. 

\begin{proposition}\label{prop:last_vertices_order(part_d)}
    For all $i < i_1 < i_2 < k$, $x_{i_1}$ appears on $P$ before $x_{i_2}$.
    For all $k \leq i_1 < i_2 \leq m$, $x_{i_1}$ appears on $P$ after $x_{i_2}$ (if it exists).
\end{proposition}

\begin{figure}[htb]
    \centering
    \includegraphics[width=0.5\textwidth]{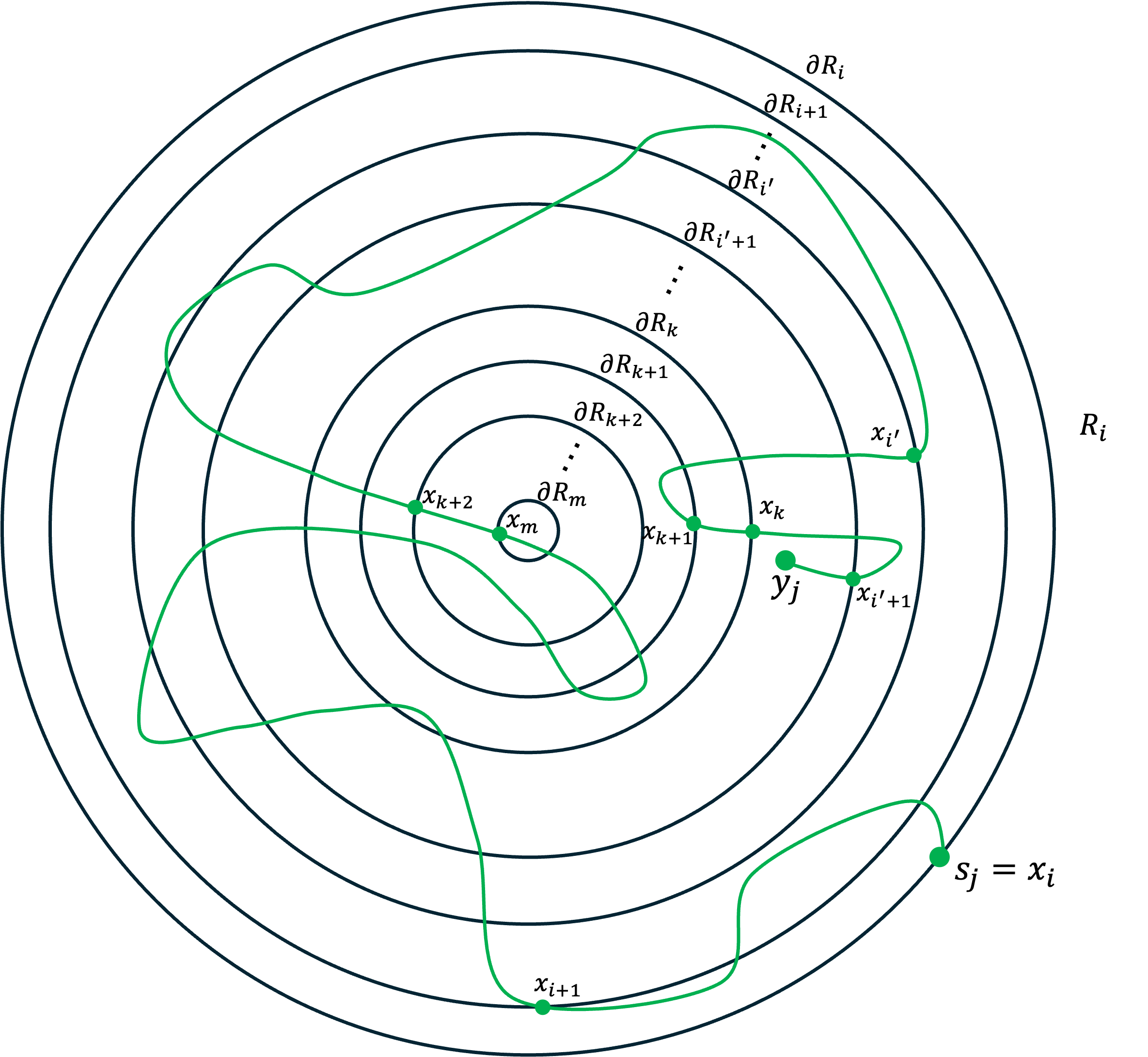}
    \caption{
    A schematic illustration of ~\cref{prop:last_vertices_order(part_d)}.
    The shortest path from $s_j\in \partial R_i$ to $y_j\in R_k$ is shown in green.
    The boundaries of the regions $R_i$ through $R_m$ in the graph decomposition are depicted as nested concentric circles.
    The last vertices $x_{i'}$ for each region are labeled explicitly along the path.
    The ordering of the $x_{i'}$ vertices along the path reflects the monotonicity stated in the proposition: for $i < i_1 < i_2 < k$, the vertices appear in increasing order along the path, and for $k \leq i_1 < i_2 \leq m$, the order reverses.
    }\label{fig:last_vertices_order(part_d)}
 \end{figure}

We first find $x_{i'}$ and $\dist_G(u,x_{i'})$, for each $k > i' > i$.
This is done using a single point location query (as described in~\cref{par:point_location_discussion}) for $y_j$ in $\VD^*(s_j,R_{i+1}^{out})$.
This point location query takes $O(Q)$ time, and is possible since at the current stage of computing the parts for level $i$, all components of the oracle have been computed for all levels greater than $i$.
We decompose $P$ into segments $P_0 = P(s_j,x_{i+1}), P_1 = P(x_{i+1},x_{i+2}), \dots, P_{k-2} = P(x_{k-2},x_{k-1}), P_{k-1} = P(x_{k-1},y_j)$.

We shall use the following markings on edges of the MSSP data structures of part (A) and of the coarse trees.
Each edge that represents a vertex at some level $0 < \ell \leq m$ is marked as such an edge (there is a separate mark for each level).
In addition, each shortcut edge of an MSSP tree is marked as such.
The markings are done when constructing these trees and only add a multiplicative $O(m)$ factor to the construction time.

We next find $x_k$.
We first try to find $x_k$ on $P_{k-1}$. 
Since $y_j$ is a vertex of the MSSP tree $T_{x_{k-1}}$, we query $T_{x_{k-1}}$ for the leafmost ancestor of $y_j$ that is a shortcut edge. 
If such a shortcut edge exists, then its leafward endpoint is $x_k$.
The distance $\dist_{T_{x_{k-1}}}(x_{k-1}, x_k)$ together with the previously calculated $\dist_G(u,x_{k-1})$ yields $\dist_G(u,x_k)$.
We break $P_{k-1}$ into two segments $P_{k-1} = P(x_{k-1},x_k)$, and $P_{k} = P(x_{k},y_j)$.

If $x_k$ was not found in this way, then we proceed as in the following description of the general case of finding $x_{i'}$ for $i'>k$ after having found all the last vertices on $P$ for levels less than $i'$.

We would like to find $x_{i'}$ for $i' \geq k$. 
By \cref{prop:last_vertices_order(part_d)}, $x_{i'}$ is not on the last segment in the decomposition of $P$.
To see this, recall that if $i'=k$ we are in the case that $x_{k}$ is before $x_{k-1}$ on $P$, and if $i'>k$, by \cref{prop:last_vertices_order(part_d)}, $x_{i'}$ appears before $x_k$ on $P$.
We work segment by segment in reverse order.
To handle the segment $P_\ell = P(x_\ell,x_{\ell+1})$, 
we check whether $x_{i'}$ lies on the part of $P$ that is represented by the $x_{\ell}$-to-$x_{\ell+1}$ path in $\hat T_{x_{\ell+1}}$. 
Let $z_{\ell+1}z'_{\ell+1}$ be the leafmost $i'$-marked ancestor edge of $x_{\ell+1}$ in $\hat T_{x_{\ell+1}}$.
If $z_{\ell+1}z'_{\ell+1}$ exists then $x_{i'}$ lies on the part of $P$ that is represented by the $z_{\ell+1}$-to-$z'_{\ell+1}$ path in $\hat T_{z_{\ell+1}}$.
We repeat the same process increasing level by level until we have found $z'_{i'}$ which is, by definition, $x_{i'}$.

If $z_{\ell+1}z'_{\ell+1}$ does not exist then $x_{i'}$ is not on $P_\ell$, and we move on to $P_{\ell-1}$. The entire process of finding $x_{i'}$ takes $\Otild(m) = \Otild(1)$ time, so $\Otild(m^2)$ for finding all of the $x_i$'s.

To find  $q_{i'}$ for each $i' > i$ (if it exists) we use a symmetric process that works along the segments $P_0,P_1, \dots, $ of $P$ in increasing order, where we query for the rootmost level-$i'$ marked ancestor edge instead of the leafmost such edge.

Hence, the time to compute all the site tables for a single Voronoi diagram with $S$ sites is $\Otild(S \cdot(Q + m^2))$

\paragraph{Side tables.} 
Let $\chi_j$ be the $(s_j,\ldots,y_j,\, y_{j-1},\ldots,s_{j-1})$ site-centroid-site chord. 
Let $R_{i'}$ be an ancestor region of $R_i$ at some level $i'\geq i$ such that $\partial R_{i'}$ does not intersect $\chi_j$.
We need to compute a single bit indicating whether $R_{i'}^{out}$ lies to the right or to the left of $\chi_j$.
Since $\partial R_{i'}$ is disjoint from $\chi_j$, this is equivalent to determining whether an arbitrary vertex $z$ on $\partial R_{i'}$ is left or right of $\chi_j$.

Let $z$ be a vertex on $\partial R_{i'}$. 
We use the point locating mechanism for $z$ in $\VD^*(u, R_i^{out})$.
Note that, as component $D$ has not yet been computed for the Voronoi diagram in this level, we cannot perform a direct call to the point location procedure of \cite{ourJACM}.
Instead, we perform a call to the point locate mechanism of \cite{CharalampopoulosGMW19} only in the current $i$ level, utilizing the distance query mechanism of \cite{ourJACM} for each of the sites evaluated.
Namely, the point location mechanism of \cite{CharalampopoulosGMW19} evaluates $O(\log n)$ sites, derived from the centroid decomposition of $\VD^*(u, R_i^{out})$, as prospective sources for a shortest path to $z$ w.r.t. additive distances.
In inspecting the shortest path distance from each prospective site $s'\in \partial R_i$, we call the point location mechanism of \cite{ourJACM} over $\VD^*(s', R_{i+1}^{out})$ in $O(Q)$ time.
The runtime of the overall point location query for $z$ in $\VD^*(u, R_i^{out})$ is therefore  $\Otild(Q)$.

Let $s_i\in \partial R_i$ be the site leading to $z$, which is returned by the point location query.
We need to determine whether vertex $z$ is left or right of $\chi_j$.
Recall that $u$ is the source vertex of the Voronoi diagram for which we are computing the side tables.
If $s_i \neq s_j$ and $s_i \neq s_{j-1}$ then 
consider the three paths $T_u(s_i,z), T_u(s_j,y_j), T_u(s_{j-1},y_{j-1})$.
The three paths are disjoint since each path does not intersect $\partial R_i$ after the first vertex.
Therefore, whether vertex $z$ is left or right of $\chi_j$ is determined by the cyclic order of $s_i, s_j, s_{j-1}$ along $\partial R_i$.
Cf. \cref{fig:side_tables}.
Otherwise, w.l.o.g. $s_i = s_j$, and whether vertex $z$ is left or right of $\chi_j$ is equivalent to whether $z$ is left or right of $y_j$ in $T_{s_j}$. 
We have already described how to determine this in $\Otild(Q)$ time in \cref{subsec:prefix_on_bisector}.

Hence, the total time to compute the side table for a Voronoi diagram with $S$ sites is $\Otild(S \cdot m \cdot Q)$, which dominates the time for constructing the site tables.

\begin{figure}[H]
    \centering
    \begin{minipage}{0.3\textwidth}
        \centering
        \includegraphics[width=\linewidth]{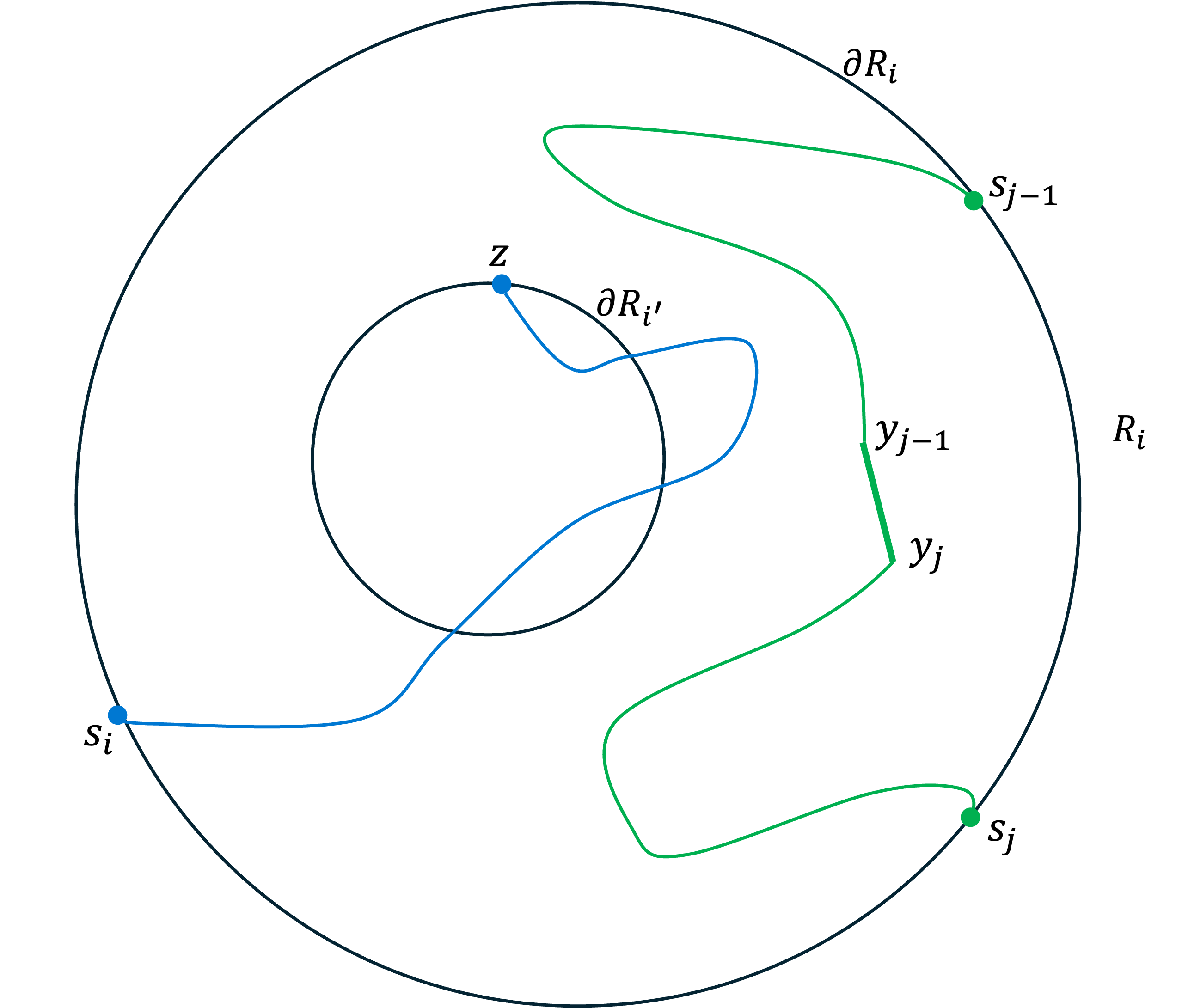}
    \end{minipage}%
    \hfill
    \begin{minipage}{0.7\textwidth}
        \captionof{figure}{%
            A schematic diagram for the computation of the side tables. 
            The boundaries of the regions $R_i$ and $R_{i'}$ in the graph decomposition are shown.
            The $\chi_j$ chord is indicated by a green line.
            The edge $y_j-y_{j-1}$ is strictly inside $R_{i'}$.
            The shortest path from $s_i$ to $z$ is indicated by a blue line.
            }
        \label{fig:side_tables}
    \end{minipage}
\end{figure}

\subsection{Constructing parts (B)--(C) of the oracle of 
\cite{ourJACM} for $m \in \Omega(\log n / \log \log n)$}\label{subsec:complementary_optimized_compute}

The analysis in \cref{subsec:analysis} concluding the algorithm description of \cref{sec:complete} yields a total runtime of $(\log n)^{O(m)}\cdot Q$
for finding a single trichromatic face at level $i$ of the recursive decomposition, where $Q$ denotes the distance oracle query runtime for all levels $i<i'\leq m$.
Recall that the distance oracle of \cite{ourJACM} provides $Q=\Otild(2^m)$.
Thus, for any $m=o(\log n / \log \log n)$ the analysis of \cref{sec:complete} results in an $n^{o(1)}$ time algorithm for finding each trichromatic face, and suffices to produce an overall  construction time of $n^{1+o(1)}$ for all Voronoi diagrams in parts (B) and (C) of the oracle.
This, however, is not the case for $m \in \Omega(\log n / \log \log n) \cap o(\log n)$ where the algorithm of \cref{sec:complete} for finding a single trichromatic face requires $n^{\Omega(1)}$ time, thereby taking $n^{1+\Omega(1)}$ time for constructing parts (B) and (C) of the oracle. 
In this subsection we resolve this issue.

\paragraph{The scaffolding mechanism.}
To overcome this problem when $m \in \Omega(\log n / \log \log n) \cap o(\log n)$ we shall invoke the procedure of \cref{sec:complete} to find a trichromatic face in a Voronoi diagram of a region $R_i$ of the recursive decomposition $\mathcal R$ that has $m$ levels, using a surrogate oracle with a recursive decomposition that has only $O(\sqrt{\log n}) \in o(\log n/\log\log n)$ levels, but still includes the region $R_i$. 
The query time of this surrogate oracle is $2^{O(\sqrt{\log n})}$.
This way, finding the trichromatic face will take 
just $(\log n)^{O(\sqrt{\log n})}\cdot 2^{O(\sqrt{\log n})} = n^{o(1)}$ time.

Let $r_1, r_2, \dots r_m$ be the parameters of the recursive decomposition $\mathcal R$ of the oracle that we need to construct.
Recall that we can afford to compute the additive distances for all the Voronoi diagrams required for the oracle of \cite{ourJACM} with decomposition $\mathcal R$.

Let $m'=\sqrt m$ and let $\mathcal R'$ be the recursive decomposition with parameters $r'_1, r'_2, \dots r'_{m'}$ such that $r'_j = r_{\lfloor j\cdot m/m'\rfloor}$. 
We call $\mathcal R'$ the scaffolding decomposition.
We construct the oracle of \cite{ourJACM} for $\mathcal R'$, which we call the scaffolding oracle. 
The construction requires $(\log n)^{O(m')}2^{O(m')} = n^{o(1)}$ time for finding a single trichromatic face of each Voronoi diagram in the oracle, so a total of $n^{1+o(1)}$ time to construct all parts of the scaffolding oracle. 

We work in batches of $m'$ consecutive levels of the $m$-level decomposition $\mathcal R$, starting at level $m$.
For the $j$'th batch, the surrogate decomposition $\mathcal R'_j$ is given by the parameters $$r'_j, r_{\lfloor j\cdot m/m' \rfloor +1}, r_{\lfloor j\cdot m/m' \rfloor +2}\dots, r_{\lfloor (j+1)\cdot m/m' \rfloor -1}, r'_{j+1}, r'_{j+2}\dots r'_{m'}.$$
That is, we insert $m/m'$ consecutive levels of $\mathcal R$ into the beginning of the suffix of the scaffolding decomposition that starts with $r'_{j+1}$.
We construct the oracle of \cite{ourJACM} for $\mathcal R'_j$, which we call the $j$'th surrogate oracle.
We already have all the components of the surrogate oracle for levels of $\mathcal R'_j$ 
starting from $r'_{j+1}$ (including), so we only need to compute the parts of the oracle for the batch of $m/m'$ levels. 
This is done using the algorithm of \cref{sec:complete} in $(\log n)^{O(m'+m/m')}2^{O(m'+m/m')} = n^{o(1)}$ time for finding a single trichromatic face of any Voronoi diagram in each of the remaining levels the oracle, so a total of $n^{1+o(1)}$ time to construct all parts of the surrogate oracle.

Doing this for all $m'$ batches takes $O(m')\cdot n^{1+o(1)} = n^{1+o(1)}$ time.
Along this computation we have computed all the Voronoi diagrams for the oracle of \cite{ourJACM} for all the regions of the $m$-level decomposition $\mathcal R$.

\newpage
\bibliographystyle{alpha}
\bibliography{references}

\end{document}